\documentclass[a4paper,12pt]{article}
\usepackage{amsmath,amsthm,amsfonts,amssymb,bm,mathrsfs}
\usepackage[protrusion=true,expansion=true]{microtype}
\usepackage{mathtools}
\usepackage{graphicx}
\usepackage{xcolor}
\usepackage{float}
\usepackage[super,square,numbers,sort&compress]{natbib} 
\usepackage{fancyhdr}
\usepackage{chemarrow}
\usepackage{extarrows}
\usepackage{comment}
\usepackage{enumerate}
\usepackage{wrapfig}
\usepackage{bibentry,natbib}
\setcitestyle{authoryear,open={(},close={)}}
\usepackage{geometry}
\usepackage[T1]{fontenc}
\usepackage{footmisc}
\usepackage{sgame}
\usepackage{color}
\usepackage{setspace}

\usepackage{fancyhdr}
\usepackage{hhline}
\usepackage[super]{nth}
\usepackage[nice]{nicefrac}
\usepackage{multicol}
\usepackage{tikz}
\usetikzlibrary{decorations.pathreplacing}
\usetikzlibrary{arrows}
\tikzstyle{block}=[draw opacity=0.7,line width=1.4cm]
\usepackage{pgfplots}
\pgfplotsset{width=10cm,compat=1.9}
\usepgfplotslibrary{external}

\usepackage{amsfonts}
\usepackage{lipsum}

\usepackage{lmodern}
\usepackage{enumitem}
\usepackage{dsfont}
\usepackage{appendix}

\usepackage{multibib}
\newcites{Appendix}{References for Online Appendix}%

\newtheorem{theorem}{Theorem}

\newtheorem{axiom}{Axiom}

\newtheorem{example}{Example}

\newtheorem{lemma}{Lemma}

\newtheorem{proposition}{Proposition}
\newtheorem{corollary}{Corollary}

\newtheorem{defn}{Definition}

\newcommand{\bE}{\mathbb{E}}

\newcommand{\bR}{\mathbb{R}}

\newcommand{\cB}{\mathcal{B}}

\newcommand{\cH}{\mathcal{H}}

\newcommand{\cP}{\mathcal{P}}

\newcommand{\cT}{\mathcal{T}}

\newcommand{\cp}{\mathsf{c}}

\newcommand{\hra}{\rightharpoonup}

\newcommand{\hrra}{\rightharpoonup\mathrel{\mspace{-15mu}}\rightharpoonup}
\newcommand{\ux}{\overline{x}}
\newcommand{\lx}{\underline{x}}
\newenvironment{restatedlemma}[1]
  {\par\medskip\noindent\textbf{Lemma~\ref{#1}.}\itshape\ }
  {\par\medskip}

\DeclareMathOperator{\supp}{supp}

\usepackage{hyperref}
\definecolor{Redish}{HTML}{890F0F}
\hypersetup{colorlinks=true,linkcolor=Redish,filecolor=Redish,urlcolor=Redish,citecolor=Redish,}
\newcommand*{\fullref}[1]{\hyperref[{#1}]{\ref{#1} \nameref{#1}}} 
\newcommand*{\figref}[1]{\hyperref[{#1}]{Figure \ref{#1}}}
\newcommand*{\secref}[1]{\hyperref[{#1}]{Section \ref{#1}}}
\newcommand*{\tabref}[1]{\hyperref[{#1}]{Table \ref{#1}}}
\newcommand*{\Equationref}[1]{\hyperref[{#1}]{Equation \ref{#1}}}
\newcommand*{\equationref}[1]{\hyperref[{#1}]{equation \ref{#1}}}
\newcommand*{\propref}[1]{\hyperref[{#1}]{Proposition \ref{#1}}}
\newcommand*{\defref}[1]{\hyperref[{#1}]{Definition \ref{#1}}}
\newcommand*{\lemmaref}[1]{\hyperref[{#1}]{Lemma \ref{#1}}}
\newcommand*{\thmref}[1]{\hyperref[{#1}]{Theorem \ref{#1}}}
\newcommand*{\cororef}[1]{\hyperref[{#1}]{Corollary \ref{#1}}}
\newcommand*{\appenref}[1]{\hyperref[{#1}]{Appendix \ref{#1}}}
\newcommand*{\factref}[1]{\hyperref[{#1}]{Fact \ref{#1}}}
\newcommand*{\assumref}[1]{\hyperref[{#1}]{Assumption \ref{#1}}}

\def\sym#1{\ifmmode^{#1}\else\(^{#1}\)\fi}

\usepackage{multicol}
\usepackage{multirow}

\newcommand{\tr}{\textcolor{red}}

\newcommand{\RN}[1]{%
  \textup{\uppercase\expandafter{\romannumeral#1}}%
}

\title{Decision Making Under Multidimensional Risk}
\author{Shaowei Ke\thanks{Department of Economics, China Europe International Business School. Email: shaoweike@ceibs.edu.} \and Mu Zhang\thanks{Department of Economics, University of Michigan. Email: muzhang@umich.edu.}
\thanks{This paper subsumes an earlier paper titled ``Multidimensional Choices under Uncertainty'' by the same authors, and partially subsumes \cite{Zhang23}. We are grateful to the co-editors and three anonymous referees, whose comments and suggestions have substantially improved the paper.
We are also grateful to Tilman B\"{o}rgers, David Dillenberger, Amanda Friedenberg, Faruk Gul, David Kreps, Fabio Maccheroni, David Miller, Xiaosheng Mu, Efe Ok, Pietro Ortoleva, Wolfgang Pesendorfer, Debraj Ray, Ran Spiegler, and participants at numerous seminars and conferences for helpful comments.}}
\date{\today}
\begin{document}

\onehalfspacing
\maketitle
\thispagestyle{empty}

\begin{abstract}
    Choice alternatives are often multidimensional and risky. We introduce and axiomatize the \textit{structured multidimensional expected utility} representation, a unified framework that generalizes existing approaches to evaluating such alternatives. The representation uses a \textit{rooted clustered tree} to organize the joint, separate, and conditional evaluation of risk across dimensions within a common structure. We analyze the uniqueness of the representation and characterize useful special cases. We apply the representation to inequality across individuals, groups, and generations and to multisource income, characterizing the implications of bracketing for stochastic dominance and the avoidance of multidimensional risk.
\end{abstract}


\newpage
\pagenumbering{arabic} 

\section{Introduction}
\label{sect_intro}
Decision makers often face alternatives that are complex and uncertain. Evaluating such alternatives requires taking into account both their multiple dimensions and the associated risk. For example, a decision maker may need to assess a product with uncertain attributes, a job with uncertain future payoffs, or a policy that generates an uncertain income distribution across individuals. Although evaluating risky multidimensional alternatives is a fundamental and ubiquitous task in economics, there is no consensus on how such alternatives should be evaluated.

To illustrate, consider the following example. Let $(x_1,x_2)$ denote the income levels of individuals 1 and 2, respectively. A policymaker is evaluating a policy that yields $(0,1)$ and $(1,0)$ with equal probability. She dislikes inequality  across individuals. Let $u(\cdot,\cdot)$ be an increasing, concave, and symmetric function. One way to evaluate the policy is to first use $u$ to assess each possible income profile and then take the expectation: $\tfrac{1}{2}u(0,1)+\tfrac{1}{2}u(1,0)$. It is well known that this approach captures ex post inequality aversion. Alternatively, she may first compute each individual's expected income, which is $1/2$, and then use $u$ to evaluate the profile of expected income levels: $u(1/2,1/2)$. This approach captures ex ante inequality aversion. These two approaches to evaluating multidimensional risk are both useful yet mutually incompatible.\footnote{See, among others, \cite{fleurbaey2010assessing}; \cite{GKPS10,GKPS12}; \cite{FL2012fairness}; and \cite{Saito13}. The first formula captures ex post inequality aversion, or inequality of outcome: $(1/2,1/2)$ is better than having $(0,1)$ or $(1,0)$ with equal probability. The second formula captures ex ante inequality aversion, or inequality of opportunity: having $(0,1)$ or $(1,0)$ with equal probability is better than having either $(0,1)$ or $(1,0)$ with certainty. Each formula misses the other concern. The second treats the mixture as $(1/2,1/2)$, while the first is linear in probabilities and, by symmetry, does not strictly prefer the mixture to either $(0,1)$ or $(1,0)$ with certainty.}

This issue is not unique to inequality aversion. It arises in many other contexts, including evaluating risky consumption bundles, dynamic choice under risk, and ambiguity aversion.\footnote{See \hyperref[OA_opposite]{Online Appendix \RN{1}.1} for a detailed discussion of these examples.} The common theme is that there are two opposite approaches to evaluating a risky multidimensional alternative:
 (i) first-aggregation-then-expectation (FATE)---the decision maker first aggregates across dimensions for each realization and then takes the expectation; and (ii) first-expectation-then-aggregation (FETA)---the decision maker first  takes the expectation within each dimension and then aggregates across dimensions.\footnote{Throughout the paper, ``aggregation across dimensions'' means combining values or utilities associated with different dimensions, whereas ``taking expectation'' means probability-weighted averaging over possible realizations, although expectation is, mathematically, a form of aggregation. This terminology does not require a distinction between risky and risk-free dimensions; any dimension may be deterministic or risky.}
Both approaches may seem reasonable but yield sharply different behavioral implications.




However, the evaluation of a risky multidimensional alternative is not limited to the two approaches discussed above. A third possibility is the recursive approach, in which the decision maker does not evaluate all dimensions of the alternative simultaneously, whether before or after taking expectation as in the FATE and FETA approaches. Instead, she evaluates them recursively, conditioning on realizations of some dimensions to evaluate another. A familiar example comes from dynamic choice. In \cite{KrepsPorteus78} and \cite{EpsteinZin89}, the decision maker evaluates multiperiod risk recursively. Consider a two-period case. For each realization of period-$1$ consumption $x_1$, she evaluates the conditional expected utility of period-$2$ consumption given $x_1$, denoted by $U_{x_1}$. She then aggregates $x_1$ and $U_{x_1}$, possibly in a nonadditive way, and finally takes the expectation of the resulting value with respect to $x_1$.


The examples above are drawn from settings that make the respective approaches particularly transparent: inequality for the FATE and FETA approaches, and dynamic choice for the recursive approach.\footnote{Section \ref{sect_ineq} applies the recursive approach in the inequality setting, in which it captures a preference for intergenerational mobility.} Despite the different economic contexts of the introductory examples, the three approaches point to the same underlying issue---how dimensions enter the evaluation: Some are evaluated jointly, some are evaluated separately, and some are evaluated conditional on realizations of others.

Motivated by this observation, we introduce a representation that organizes these three forms of evaluation through a common graphical structure. Within this framework, the FATE, FETA, and recursive approaches are three extreme cases, while intermediate cases combine these forms of evaluation across dimensions. Our goal is to make this structure explicit and characterize the preferences that can be represented in this form. We then analyze the representation and its key special cases and study its applications.

We begin by introducing the structure that underlies our unified framework, the \textit{rooted clustered tree} (RCT). An RCT is a directed rooted tree whose non-root vertices  form a partition of the dimensions. The role of the RCT can be understood through the three approaches described above. Risk across dimensions in the same vertex is evaluated jointly, as under the FATE approach. Risk in a vertex is evaluated conditional on realizations of the dimensions in its ancestor vertices, as under the recursive approach. Risk across dimensions in different branches is evaluated separately, as under the FETA approach. That is, conditional on the relevant ancestor realizations, the correlation between any two such dimensions is \textit{behaviorally irrelevant}: Changing that correlation while holding fixed the conditional marginal distribution of each dimension does not affect the decision maker's choices. Figure \ref{simple_trees} illustrates three examples of RCTs with two dimensions.


\begin{figure}[ht]
    \centering
    \begin{tikzpicture}[scale=1]
        \node at (0,0) {$o$};
        \node at (0,-1.5) {$1,2$};
        \draw[-stealth] (0,-.3) -- (0,-1.2);

        \node at (3,0) {$o$};
        \node at (2,-1.5) {$1$};
        \node at (4,-1.5) {$2$};
        \draw[-stealth] (3,-.3) -- (2,-1.2);
        \draw[-stealth] (3,-.3) -- (4,-1.2);

        \node at (6,0) {$o$};
        \node at (6,-1.5) {$1$};
        \node at (6,-3) {$2$};
        \draw[-stealth] (6,-.3) -- (6,-1.2);
        \draw[-stealth] (6,-1.8) -- (6,-2.7);
    \end{tikzpicture}
    
    \caption{Suppose there are two dimensions, $1$ and $2$, and $o$ is an auxiliary vertex representing the root. The left-hand RCT evaluates dimensions $1$ and $2$ jointly. The middle RCT evaluates them separately. The right-hand RCT evaluates dimension $2$ conditional on dimension $1$.}
    
    \label{simple_trees}
\end{figure}
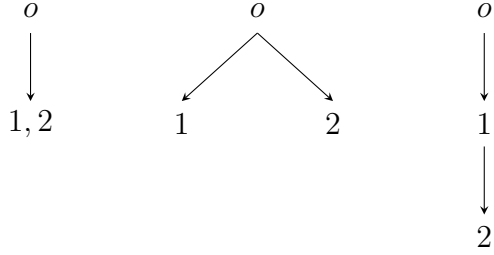

The RCT does more than accommodate the different forms of evaluation described above; it shows that they share a common directional structure. Think of an arrow from dimension $i$ to dimension $j$ as indicating that dimension $j$ can be evaluated conditional on the realization of dimension $i$. Both the conditional risk of dimension $j$ and the utility used to evaluate it may depend on the realization of dimension $i$. If arrows are present in both directions, either dimension can serve as the conditioning dimension, so the risk across the two dimensions is effectively evaluated jointly. The RCT therefore places the two dimensions in the same vertex, and the two arrows are absorbed into that vertex. If only the arrow from $i$ to $j$ is present, the vertex containing $i$ is an ancestor of the vertex containing $j$, and dimension $j$ is evaluated conditional on dimension $i$. If neither arrow is present, the two vertices lie in different branches and the dimensions are evaluated separately. Thus, joint, conditional, and separate evaluation reflect whether the same directional relation holds in both directions, in one direction, or in neither direction.

This common structure also gives RCTs a close connection to \textit{directed acyclic graphs} (DAGs), which are widely used in the literature on causal inference, including work on imperfect causal perception.\footnote{\label{causality}For recent work that adopts DAGs to study causal inference, see, among others, \cite{spiegler2016bayesian,spiegler2020behavioral,spiegler20}; \cite{eliaz2020model}; and \cite{EllisThysen24}.} We make this connection precise after \thmref{thm_main}.

Given the RCT, we introduce the main representation of the decision maker's preference: the \textit{structured multidimensional expected utility} (SMEU) representation. The representation describes how the decision maker evaluates a \textit{lottery}, which is a risky multidimensional alternative, using the recursive structure of an RCT. At each vertex, for every realization of the dimensions in that vertex, she combines that realization---possibly in a nonadditive way---with the evaluations associated with its child vertices. She then evaluates the resulting risk conditional on realizations of the dimensions in its ancestor vertices. Both the conditional risk and the utility used in this evaluation may depend on those realizations. This process continues until all vertices have been evaluated, yielding the utility for the lottery. Thus, the SMEU representation extends the preferences of \cite{KrepsPorteus78} and \cite{EpsteinZin89} from temporal risk to multidimensional risk, with the RCT specifying how dimensions enter the recursive evaluation.

The value of the SMEU framework goes beyond nesting the FATE, FETA, and recursive representations (formally defined in Section \ref{sect_repre}) and including intermediate cases that combine their features. It also allows us to connect familiar economic concerns that are usually studied separately. To illustrate this value, we consider a common setting in which a policymaker evaluates inequality across individuals, groups, and generations. As is well known, the FATE and FETA representations capture ex post and ex ante inequality aversion, respectively. We show that, within this broad setting, the recursive representation can capture a preference for intergenerational mobility and that intermediate SMEU representations can combine ex post, ex ante, and intergenerational concerns.\footnote{\label{mobility_group}For intergenerational mobility, see, among others, \cite{becker1979equilibrium}; \cite{solon1992intergenerational}; \cite{zimmerman1992regression}; and \cite{corak2013income}. For group inequality (both within-group and between-group inequality), see, among others, \cite{Formby89}; \cite{Gottschalk97}; \cite{LaFerrara02}; \cite{Lemieux06}; \cite{Elbers08}; \cite{Burstein19}; and \cite{Darity22}.}

We analyze the uniqueness of the RCT. A preference may have SMEU representations with different RCTs. We define a simplicity criterion that first minimizes each dimension's depth and then, holding these depths fixed, minimizes the extent to which dimensions are evaluated separately in different vertices. We show that the \textit{simplest} RCT is unique.

We axiomatize the SMEU representation and its key special cases, including representations that combine features of any two of the FATE, FETA, and recursive representations. The axioms are built around two binary relations over dimensions, both revealed from choices. One identifies when the correlation between two dimensions is behaviorally irrelevant; the other identifies when choices are consistent with evaluating one dimension conditional on realizations of another. A key axiom, RCT perfection, ensures that these two relations fit together consistently. Suppose one dimension is evaluated conditional on the realizations of two others. To compare lotteries, the decision maker must take into account both the conditional evaluation associated with each pair of realizations and the probability of that pair. The two conditioning dimensions must therefore be evaluated jointly, or one must be evaluated conditional on the other. RCT perfection imposes the corresponding consistency requirement: Choices must be consistent with evaluating at least one of the two conditioning dimensions conditional on the other.

Finally, we apply the generalized bracketing representation to multisource income. The decision maker partitions income sources into brackets, evaluates risk jointly within each bracket, and aggregates the resulting evaluations across brackets. We characterize how bracketing affects violations of stochastic dominance and show that coarser bracketing makes dominance easier to satisfy. We also characterize the decision maker's \textit{avoidance of multidimensional risk}---when she prefers a given distribution of aggregate income to be generated by a single source rather than spread across several sources---and show that finer bracketing can strengthen this preference.

\subsection{Related Literature}\label{sect_lit}
Many papers have studied multivariate risk, but most remain within expected utility theory and focus on analyzing measures of risk attitude.\footnote{See, among others, \cite{Keeney73}; \cite{KihlstromMirman74,KihlstromMirman81}; \cite{Richard75}; \cite{Duncan77}; \cite{Karni79}; \cite{Schlee90}; \cite{LevyLevy91}; \cite{Grant95}; and \cite{ERS07}.} Some papers depart from expected utility theory but do so in ways more aligned with classic non-expected-utility analyses (e.g., \cite{Karni89}). Our approach is complementary. The SMEU representation need not satisfy independence. When independence fails, the failure reflects how dimensions are evaluated jointly, separately, or conditional on realizations of others, rather than the kind of one-dimensional departure from expected utility highlighted by the Allais paradox.

\cite{DillenbergerLu26} develop a complementary approach that allows risk domains to be subjectively defined, whereas we take the product structure over dimensions as exogenous. They characterize domain-specific pure risk aversion, and their intra-domain independence condition is related to our unidimensional independence: Both impose independence locally while allowing it to fail globally.

The FETA representation can capture narrow bracketing and choice behavior commonly described in the literature as correlation neglect,\footnote{\label{lit_bracket}See, among others, \cite{TverskyKahneman81}; \cite{Thaler85}; \cite{RLR99}; \cite{BHT2006bracketing}; \cite{RabinWeizsacker09}; \cite{LevyRazin15}; \cite{EnkeZimmermann19}; \cite{V2023bracketing}; \cite{Zhang23}; \cite{ellis2024revealing}; and \cite{Camara21}.} while the SMEU representation allows more general forms of bracketing and more general ways in which correlations across dimensions affect choice. Section \ref{section_app} studies the avoidance of multidimensional risk under generalized bracketing \citep{H2021issues}, and the framework is also related to mental accounting, in which some dimensions may be evaluated separately from others \citep{Thaler85,Thaler99}. Our paper generalizes \cite{Zhang23} from two dimensions to multiple dimensions and accommodates history-dependent risk attitudes in recursive models. By studying computationally tractable decision rules, \cite{Camara21} characterizes a dynamic bracketing model that generalizes narrow bracketing in a way that differs from ours. Relative to our analysis, independence is maintained in \cite{Camara21}.

Our paper offers a new way to resolve two long-standing tensions in the literature on inequality aversion and ambiguity aversion. In both settings, there are two well-known but behaviorally incompatible modeling approaches: the ex ante and the ex post.\footnote{\label{lit_amb}For ambiguity aversion, see \cite{Raiffa61}; \cite{DominiakSchnedler11}; \cite{Saito15}; \cite{ORR19}; \cite{KeZhang20}; and \cite{BHL22}.} For inequality, the ex ante approach corresponds to inequality of opportunity, whereas the ex post approach corresponds to inequality of outcomes. For ambiguity, the ex ante approach allows randomization to hedge ambiguity, whereas the ex post approach does not. Prior work combines the two approaches in other ways. For example, \cite{Saito13,Saito15} study weighted averages of them, and \cite{KeZhang20} further generalize that framework. Relatedly, \cite{feldman2024disentangling} study risky social decisions using a flexible model that combines ex ante and ex post fairness concerns, together with risk attitudes and other-regarding preferences. Our focus is different: We provide a representation of multidimensional risk that nests FATE and FETA and also captures recursive and intermediate RCT structures. Under the generalized bracketing representation, the ex post approach is applied within each bracket and the ex ante approach across brackets, thereby combining the two within a single model. For inequality, Section \ref{sect_ineq} shows how this representation combines ex post inequality within groups and ex ante inequality across groups, and how other SMEU representations combine inequality concerns across individuals, groups, and generations (see Footnote \ref{mobility_group} for related literature).

Our paper is related to the literature on dynamic preferences under risk. \cite{DDGO20} show that exponentially discounted expected utility implies that the decision maker is risk-seeking over time lotteries, contrary to experimental evidence. Their solution (see also \cite{KihlstromMirman81} and \cite{DGO2020SI}) aggregates across time via exponential discounting and then applies a Bernoulli index before taking expectation, which corresponds to a special case of our FATE representation. Conversely, \cite{Selden78} and \cite{SeldenStux78} adopt the opposite order, similar to special cases of the FETA representation. Our framework nests these approaches, as well as the recursive approach (\citealp{KrepsPorteus78}; \citealp{EpsteinZin89}). We briefly discuss time lotteries as a possible future application in Section \ref{sect_end}.

We are not the first to study the tension between the FATE and FETA representations in multiple contexts. In a setting in which each alternative is described by a matrix of attribute values, \cite{LRW2023deceptive} highlight the conflict between row-first (row-monotonic) and column-first (column-monotonic) aggregations and provide guidance on when each is appropriate.\footnote{\cite{fleurbaey2010assessing} and \cite{MP2015multi} also consider this setting. Related to \cite{LRW2023deceptive}, \cite{MP2015multi} analyze row- and column-separability and monotonicity conditions.} When rows are interpreted as states of the world, this setting becomes more comparable to ours. Row-first aggregation corresponds to our FATE representation and column-first aggregation corresponds to our FETA representation. Our focus, however, is not on choosing between the two, but on developing and applying a representation that nests both and also accommodates recursive evaluation and intermediate cases.


If we interpret different dimensions in our setup as different sources of uncertainty, our paper is also related to \cite{ErginGul09} and \cite{CC-VMMM21}. In both papers, the decision maker's risk attitude may be source-dependent and she may evaluate risks source-by-source before aggregating across sources, conceptually similar to our generalized bracketing representation. However, unlike in our framework, how states are grouped into sources is exogenously fixed. \cite{chew2023rich} also develop a source-dependent extension of expected utility, but source dependence in that case is modeled through different mixture operators rather than through dimensions.

The rest of the paper is organized as follows. Section \ref{sect_repre} introduces the RCT and the SMEU representation. Section \ref{sect_ineq} illustrates the FATE, FETA, recursive, and intermediate SMEU representations in a setting with inequality across individuals, groups, and generations. Section \ref{sect_unique} analyzes  uniqueness of the RCT and the Bernoulli indices. Sections \ref{sect_axiom} and \ref{sect_special} axiomatize the SMEU representation and its key special cases. Section \ref{section_app} applies generalized bracketing to multisource income, and Section \ref{sect_end} concludes.

\section{Setup and Representation}
\label{sect_repre}
For an arbitrary set $Z$, let $\Delta(Z)$ denote the set of all simple lotteries (probability measures with a finite support) on $Z$. Let $I=\{1,\dots, N\}$ be a finite set of integers with $N>1$. For every $i\in I$, let $X_i=[\underline{x}_i,\overline{x}_i]\subseteq\bR$ be nondegenerate.
Let $X=\bigtimes_{i\in I} X_i$. Generic elements of $X$ are called consequences. Generic elements of $\Delta(X)$ are called lotteries.

Fix any $A\subseteq I$. Let $X_A=\bigtimes_{i\in A} X_i$. We use $x,y,z$ to denote generic elements of $X_A$ and $p,q,r,s$ to denote generic elements of $\Delta(X_A)$.\footnote{The terms ``consequences'' and ``lotteries'' are reserved for elements in $X$ and $\Delta(X)$.} Let $\supp(p)$ denote the support of $p\in\Delta(X_A)$. We denote $p\in\Delta(X_A)$ that yields $x\in X_A$ with certainty by $\delta_x$. When there is no risk of confusion, we identify $\delta_x$ with $x$, and identify a subscript or a superscript $A\subseteq I$ with $i$ if $A=\{i\}$ and with $-i$ if $A=\{i\}^\mathsf{c}$.\footnote{For any $A$, we define $A^\mathsf{c} := I \setminus A$. That is, complements are always taken relative to $I$, even if $A \not\subseteq I$.} For any $p,q\in\Delta(X_A)$ and $\alpha\in[0,1]$, we write $p\alpha q$ as shorthand for the convex combination $\alpha p + (1-\alpha)q\in\Delta(X_A)$.


\medskip

\textbf{Marginal and Conditional Distributions.} For any $A\subseteq B\subseteq I$ and $p\in\Delta(X_B)$, we use $p_A\in\Delta(X_A)$ to denote $p$'s marginal distribution on $A$, and use $x_A\in X_A$ to denote the restriction of $x\in X_B$ to $A$. For any $A,B\subseteq C\subseteq I$ such that $A\cap B=\emptyset$, $p\in\Delta(X_C)$, and $x\in X_B$, let $p_{A|x}\in\Delta(X_A)$ denote $p$'s conditional marginal distribution on $A$ given $x$, and let $\bE^p_{A|x}$ denote the expectation operator under distribution $p_{A|x}$. We write $\bE^p_A$ if $B=\emptyset$ and $\bE^p$ if $A=C$.
For any disjoint subsets of $I$, $A_1,\dots, A_n$, with $\bigcup_{i=1}^n A_i=A\subseteq I$, and any $p_i\in\Delta(X_{A_i})$ for every $i\in\{1,\dots,n\}$, we use $(p_1,\dots,p_n)$ to denote the unique $q\in\Delta(X_A)$ such that  $q(x)=p_1(x_{A_1})\times\dots\times p_n(x_{A_n})$ for every $x\in X_A$. 

\medskip

\textbf{Preference.} The decision maker has a preference $\succsim$ over $\Delta(X)$. Its asymmetric and symmetric parts are denoted by $\succ$ and $\sim$, respectively. For any nonempty $A\subseteq I$ and $x\in X_{A^\mathsf{c}}$, we define the conditional preference $\succsim_x$ on $\Delta(X_A)$ such that for any $p,q\in\Delta(X_A)$, we have $p\succsim_x q\iff (p,x)\succsim (q,x)$. We define $\succ_x$ and $\sim_x$ similarly.

\subsection{Multidimensional Risk and RCTs}\label{sect_rct}
Let $\cP$ be a partition of $I$. An RCT is an out-tree on $\cP_o:=\cP \cup\{o\}$, in which $o$ is an additional vertex representing the root.\footnote{A directed graph is an out-tree if (i) the underlying undirected graph is connected and acyclic, and (ii) there is a unique directed path from the root to every other vertex of the graph.}  Denote such a tree by $\cT=(\cP,E)$, in which $E$ is the set of directed edges.\footnote{A directed edge is represented by an ordered pair $(A,B)\in \cP_o\times\cP_o$, which denotes an edge pointing from vertex $A$ to vertex $B$.} Let $\mathbb{T}$ denote the set of all RCTs.\footnote{In \hyperref[number]{Online Appendix \RN{3}.3}, we analyze how the number of RCTs depends on $N$.} Each vertex of an RCT (except the root) corresponds to a set of dimensions, which motivates the term ``clustered.''

\medskip

\textbf{Tree Relations.} By definition, $A$ is the (unique) parent of $B$ if $(A,B)\in E$. For any $A\in \cP_o$, denote by $c(A)=\{B\in \cP: (A,B)\in E\}$ the set of $A$'s children. Denote by $a(A)$ the set of ancestors of $A$---that is, its parent, parent's parent, and so on. Let $\bar a(A)=\cup_{B\in a(A)\setminus\{o\}} B$ be the set of dimensions contained in $a(A)$. We define $d(A)$ as the descendants of $A$---its children, children's children, and so on---and $\bar d(A)$ as the set of dimensions contained in $d(A)$. The set $(A\cup \bar d(A))^\mathsf{c}$ will play an important role in what follows. It captures $A$'s ancestor dimensions together with all dimensions from branches of the RCT other than the one containing $A$.



\medskip

\textbf{Interpretation.} An RCT records how dimensions are organized in the evaluation of a lottery. Dimensions in the same vertex are evaluated jointly. A vertex is evaluated conditional on the realizations of dimensions in its ancestor vertices. If two vertices lie in different branches, their evaluations are not conditioned on each other; they are combined only through the utility at their common ancestor. Thus, the RCT specifies which marginal and conditional distributions induced by the lottery enter the evaluation. Note that the RCT does not introduce alternative probability measures: Whenever the SMEU representation below uses a marginal or conditional distribution, it uses the one induced by the lottery being evaluated. See Figure \ref{simple_trees} for three simple examples.

\subsection{The SMEU Representation}
Several standard notational conventions are useful for understanding the definition of our main representation. First, for any set $A$, we interpret $X_A$ as $X_{A\cap I}$. This is convenient because in some expressions $A$ may include the root $o$. Second, if we encounter $x\in X_A$ with $A=\emptyset$, then $x$ will be ignored in the expression. Let $(\cP, E)$ be an RCT. This convention has the following implications:
\begin{enumerate}
    \item For any $p_{A|x}$ with $x\in X_B$ and $B=\emptyset$, we identify $p_{A|x}$ with $p_A$.
    \item For any function $f_{x}$ with $x\in X_A$ and $A=\emptyset$, we identify $f_{x}$ with $f$.
    \item For any function $f:X_{\bar a(A)}\times X_A\times\bR^{c(A)}\to\bR$ with $A\in\cP$ and $\bar a(A)=\emptyset$, we identify the domain of $f$ with $X_A\times \bR^{c(A)}$.
\end{enumerate}
Third, in the last implication above, if $c(A)=\emptyset$ instead, $f$'s domain is identified with $X_{\bar a(A)}\times X_A$. Last, for any $A\in \cP$, $x\in X_{\bar a(A)}$, $p\in \Delta(X)$, and function $f$, let $\bE^p_{A|x}\:f=0$ if $x\not\in \supp(p_{\bar a(A)})$.

\begin{defn}
\label{def_heu}The preference has an SMEU representation if there exist an RCT $\cT=(\cP,E)$ and functions $u^o:\bR^{c(o)}\to\bR$ and $u^A:X_{\bar a(A)}\times X_{A}\times\bR^{c(A)}\to\bR$ for every $A\in \cP$ such that defining recursively $U^A_x:\Delta(X)\to\bR$ for all $A\in\cP$ and $x\in X_{\bar a(A)}$ by
\begin{equation}
    \label{kp}U^A_x(p)=\bE^p_{A|x}\:u^A(\:x,\:y,\:(U^B_{(x,y)}(p))_{B\in c(A)}\:),
\end{equation}
and the function $U^o:\Delta(X)\to\bR$ by 
\begin{equation}
    \label{kp2}U^o(p)=u^o(\:(U^B(p))_{B\in c(o)}\:),
\end{equation}
the following statements hold for all $A\in\cP_o$, $x\in X_{\bar a(A)}$, and $z\in X_{(A\cup \bar d(A))^\mathsf{c}}$ with $z_{\bar a(A)}=x$:
\begin{enumerate}
    \item $p\succsim_z q\iff U^A_x(\delta_z, p)\geqslant  U^A_x(\delta_z, q)$ for all $p,q\in \Delta(X_{A\cup \bar d(A)})$; and
    \item $U^A_x(\delta_{(z,y)})$ is continuous and strictly increasing in $y\in X_{A\cup \bar d(A)}$.
\end{enumerate}
We denote the SMEU representation by $(\cT,(u^A)_{A\in\cP_o})$.
\end{defn}

The first condition in the definition, that $U^A_x$ represents $\succsim_z$, implies the usual representation condition: For any lotteries $p,q\in\Delta(X)$, $p\succsim q$ if and only if $U^o(p)\geqslant  U^o(q)$---that is, the function $U^o$  represents $\succsim$ and is derived recursively as in \cite{KrepsPorteus78}. The second condition requires that the representations of conditional preferences are continuous and monotone in the absence of risk. This condition is useful in the proofs and may be relaxed. It ensures that the set of consequences is sufficiently rich so that  we can construct lotteries with different supports that are indifferent to a given lottery.

Although the SMEU representation appears complex, its main idea is straightforward and parallels the recursive structure of \cite{KrepsPorteus78} and \cite{EpsteinZin89}. In those models, the order of recursion is fixed by the exogenous order of time. In the SMEU representation, by contrast, the recursive structure is encoded in the RCT, which is revealed from choice behavior.

Take any vertex $A\in\cP$ in the RCT and realized values of the ancestor dimensions $x\in  X_{\bar a(A)}$. Equation (\ref{kp}) evaluates a lottery $p$ as follows: 
\begin{equation*}   U^A_x(p)=\underbrace{\bE^p_{A|x}}_{\text{\footnotesize $\begin{array}{c}
        \text{conditional expectation}\\
        \text{with respect to the}\\
        \text{present dimensions}
    \end{array}$}}
    u^A(x, \underbrace{y}_{\text{\footnotesize$\begin{array}{c}
        y\in X_{A},\\
        \text{the present}\\
        \text{dimensions}
    \end{array}$}}
    ,\underbrace{(U^B_{(x,y)}(p))_{B\in c(A)}}_{\text{\footnotesize$\begin{array}{c}      
        \text{future utilities}\\
        \text{conditional on }(x,y)
    \end{array}$}}),
\end{equation*}
in which the terms ``present dimensions'' and ``future utilities'' reflect the analogy to time in recursive utility models: The dimensions in vertex $A$ are evaluated at the current stage of the recursive evaluation, while the utilities associated with the child vertices are evaluated at the next stage, analogous to the future. Note that there may be multiple dimensions evaluated simultaneously at a given stage, because a vertex may contain more than one dimension.\footnote{Under the time analogy, multiple dimensions may be evaluated within a single period---for instance, when the decision maker consumes multiple types of goods in one period.} Likewise, there may be multiple ``futures,'' since a vertex may have several children. At the root, the overall utility is given by (\ref{kp2}), which aggregates the utilities of the top-level clusters.

The FATE, FETA, and recursive representations can be obtained by imposing simple restrictions on the RCT of the SMEU representation. Let $\mathbb{T}(\succsim)$ denote the set of all RCTs $\cT\in\mathbb{T}$ such that $(\cT,(u^A)_{A\in\cP_o})$ is an SMEU representation of $\succsim$ for some $(u^A)_{A\in\cP_o}$. We say that $\succsim$ has an SMEU representation with a \textit{unique} RCT if $\mathbb{T}(\succsim)$ is a singleton.

\begin{defn}\label{def_special1}
We say that the preference $\succsim$ has 
\begin{itemize}
    \item a FATE representation if it has an SMEU representation with an RCT such that $c(o)=\{I\}$;
    \item a FETA representation if it has an SMEU representation with a unique RCT such that $c(o)=\{\{1\},\dots,\{N\}\}$;
    \item a recursive representation if it has an SMEU representation with a unique RCT such that $|c(A)|\leqslant 1$ for all $A\in\cP_o$.
\end{itemize}
\end{defn}

Figure \ref{simple_trees} illustrates examples of the RCTs corresponding to the FATE, FETA, and recursive representations shown on the left, in the middle, and on the right, respectively. It is immediate that $\succsim$ has a FATE representation if and only if there exists a continuous and strictly increasing function $u:X\to\bR$ such that, for all $p,q\in\Delta(X)$,
$$
p\succsim q\iff \bE^p\:u(x)\geqslant \bE^q\:u(x).
$$
A FETA representation takes the form 
\[U(p) = v(\:\bE^p_1u^1(x_1),\dots,\bE^p_Nu^N(x_N)\:)\]
for all lotteries $p$, in which $v:\bR^I\to\bR$ and $u^i:X_i\to\bR$ for each $i$ are continuous and strictly increasing.
These formulations correspond to the traditional definitions of the FATE and FETA representations. A recursive representation with $N=2$ and $c(o)=\{\{1\}\}$ takes the form
$$
U(p)=\bE^p_1\:u(x_1,\:\bE^p_{2|x_1}v_{x_1}(x_2))
$$
for all lotteries $p$. Given a recursive representation with RCT $(\cP,E)$, each element of $\cP$ is a singleton and has at most one child. It is as if the unique dimension $i_1$ in $c(o)$ comes first relative to the remaining dimensions, the unique $i_2$ in $c({i_1})$ comes next relative to the rest, and so on. In this sense, the recursive representation is analogous to those in \cite{KrepsPorteus78} and \cite{EpsteinZin89}. The ordering of dimensions, however, is not imposed by calendar time; it is encoded by the RCT and revealed from choices.

\section{Inequality Across Groups and Generations}\label{sect_ineq}
While the FATE, FETA, and recursive representations have been applied in various contexts---and often at least two of them have been discussed together---our theory unifies them within a single framework and introduces intermediate cases that bridge the gaps between them. This section illustrates the value of that unification in a common environment for social evaluation: policies that affect wealth across individuals, groups, and generations. In this environment, it is well known that the FATE and FETA representations capture ex post and ex ante inequality aversion, respectively. We show below that the recursive representation captures a preference for intergenerational mobility. A long-horizon environmental policy, for example, may naturally involve all three concerns: the distribution of prospects among members of future generations, inequality in their realized wealth, and the dependence of their wealth on that of earlier generations. Other SMEU representations can combine these and related within-group and between-group concerns (see Footnote \ref{mobility_group} for related literature).

For expositional clarity, the examples below isolate these concerns using simplified versions of this environment. The Introduction uses two individuals to distinguish ex post and ex ante inequality; here we use two generations to isolate intergenerational mobility; and later examples add groups and multiple individuals within generations. Each example retains only the dimensions needed for the concern at hand, but all involve policy-induced joint distributions over wealth levels.\medskip

\textbf{Preference for Intergenerational Mobility.} First, it is natural to ask what a recursive representation may capture in this setting. Let $N=2$, $x_1\in X_1$ denote the wealth of the old generation, and $x_2 \in X_2$ denote that of the young generation. Consider  
\[
\begin{aligned}
p &= \tfrac{1}{2}\delta_{(0,0)} + \tfrac{1}{2}\delta_{(1,1)}, \quad 
q = \tfrac{1}{2}\delta_{(1,0)} + \tfrac{1}{2}\delta_{(0,1)}, \\[4pt]
\text{and }r &= p\tfrac{1}{2}q 
   = \tfrac{1}{4}\delta_{(0,0)} + \tfrac{1}{4}\delta_{(0,1)} 
   + \tfrac{1}{4}\delta_{(1,0)} + \tfrac{1}{4}\delta_{(1,1)}.
\end{aligned}
\]
Under policy $p$, both generations are rich or poor together. Under $q$, one generation is rich if and only if the other is poor. Under $r$, the two generations' wealth levels are independent, reflecting intergenerational mobility.

A policymaker may strictly prefer $r$ to both $p$ and $q$. This is impossible under the FATE representation, since it is expected-utility and $r$ is a mixture of $p$ and $q$, and thus $r$'s utility must lie between those of $p$ and $q$. It is also impossible under the FETA representation given by
$v(\bE_1^s\:u_1,\:\bE_2^s\:u_2)$ for all $s$, because $p_i = q_i = r_i=\delta_0\frac{1}{2}\delta_1$ for $i = 1, 2$.

However, under a recursive representation 
$$U(p)=\bE^p_1\:u(x_1,\:\bE^p_{2|x_1}v(x_2)),$$ it is possible that $r \succ p, q$. The reason is that under $r$, conditional on any realization of $x_1$, there remains  nontrivial risk over $x_2$---whereas under $p$ or $q$, no such conditional risk exists. If the policymaker values such conditional risk positively, then she may rank $r$ strictly above $p$ and $q$. \medskip

\textbf{Intergenerational and Group Inequality.}
Other special cases of the SMEU representation can be particularly useful for modeling intergenerational and group inequality. Intergenerational inequality concerns how economic outcomes and opportunities persist or change across generations, while group inequality focuses on inequality within and between social or demographic groups.

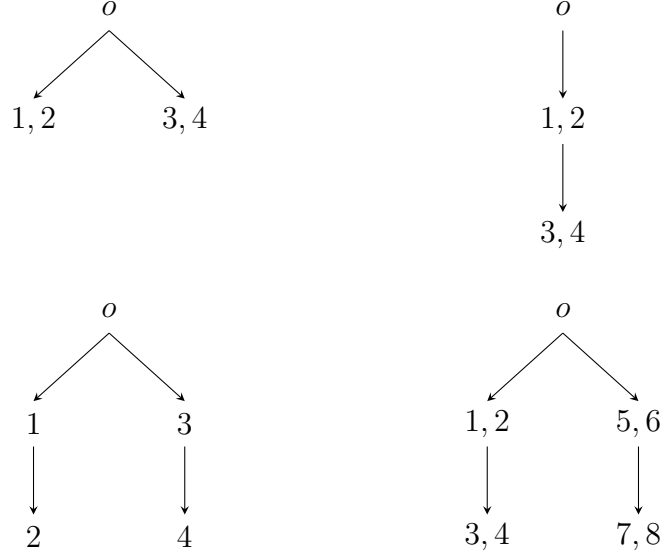
\begin{figure}[ht]
    \centering
    \begin{tikzpicture}[scale=1]
        \node at (0,0) {$o$};
        \node at (-1,-1.5) {$1,2$};
        \node at ( 1,-1.5) {$3,4$};
        \draw[-stealth] (0,-.3) -- (-1,-1.2);
        \draw[-stealth] (0,-.3) -- ( 1,-1.2);

        \node at (6,0) {$o$};
        \node at (6,-1.5) {$1,2$};
        \node at (6,-3.0) {$3,4$};
        \draw[-stealth] (6,-.3) -- (6,-1.2);
        \draw[-stealth] (6,-1.8) -- (6,-2.7);

        \node at (0,-4) {$o$};
        \node at (-1,-5.5) {$1$};
        \node at ( 1,-5.5) {$3$};
        \node at (-1,-7) {$2$};
        \node at ( 1,-7) {$4$};
        \draw[-stealth] (0,-4.3) -- (-1,-5.2);
        \draw[-stealth] (0,-4.3) -- ( 1,-5.2);
        \draw[-stealth] (-1,-5.8) -- (-1,-6.7);
        \draw[-stealth] ( 1,-5.8) -- ( 1,-6.7);

        \node at (6,-4) {$o$};
        \node at (5,-5.5) {$1,2$};
        \node at (7,-5.5) {$5,6$};
        \node at (5,-7) {$3,4$};
        \node at (7,-7) {$7,8$};
        \draw[-stealth] (6,-4.3) -- (5,-5.2);
        \draw[-stealth] (6,-4.3) -- (7,-5.2);
        \draw[-stealth] (5,-5.8) -- (5,-6.7);
        \draw[-stealth] (7,-5.8) -- (7,-6.7);
    \end{tikzpicture}
    \caption{The upper left shows an RCT of a generalized bracketing representation that combines the FATE and FETA representations. The upper right shows a generalized recursive representation that combines the FATE and recursive representations. The lower left shows a recursive bracketing representation that combines the FETA and recursive representations. These three representations will be formally defined and characterized in Section \ref{sect_special}. The lower right shows an RCT that combines features of all three extreme cases of the SMEU representation---the FATE, FETA, and recursive representations.}
    \label{complex_trees}
\end{figure}

We begin by showing how a special case of the SMEU representation---one that combines features of the FATE and FETA representations and will be formally called a generalized bracketing representation (see Section \ref{sect_special})---can capture group inequality. Suppose $N=4$, and let $x_i\in X_i$ denote individual $i$'s wealth. Let $A=\{1,2\}$ and $B=\{3,4\}$ represent two social groups. Consider a policymaker whose SMEU representation is  
\[
U(p)=v(\bE^p_A\: u^A,\:\bE^p_B\: u^B)
\]
for all lotteries $p$. Under this utility function, the dimensions are partitioned into two subsets $A$ and $B$. The policymaker jointly evaluates risk within each subset and then aggregates across subsets. Since $U(p)$ depends only on the marginal distributions over $A$ and $B$, correlations between the two subsets do not affect the evaluation. The corresponding RCT is shown in the upper left of Figure \ref{complex_trees}.

For simplicity, let $\delta_{(1,1,0,0)}\sim\delta_{(0,0,1,1)}$ and $v,u^A,u^B$ be strictly concave. Then, the policymaker exhibits aversion to ex post inequality within each social group ($B$ and $A$, respectively):
\[
\delta_{(1/2,1/2)}\succ_{(x_1,x_2)}\tfrac{1}{2}\delta_{(0,1)}+\tfrac{1}{2}\delta_{(1,0)}\quad\text{and}\quad
\delta_{(1/2,1/2)}\succ_{(x_3,x_4)}\tfrac{1}{2}\delta_{(0,1)}+\tfrac{1}{2}\delta_{(1,0)}.
\]
At the same time, she is averse to ex ante inequality between social groups:
\[
\tfrac{1}{2}\delta_{(1,1,0,0)}+\tfrac{1}{2}\delta_{(0,0,1,1)}\succ\delta_{(1,1,0,0)}\sim\delta_{(0,0,1,1)}.
\]
This special case of the SMEU representation therefore constitutes a natural intermediate case between the FATE and FETA representations.\footnote{The same idea can be applied to settings with ambiguity. See Section \ref{sect_lit} and \hyperref[OA_opposite]{Online Appendix \RN{1}.1} for further discussion.}

Given the previous discussion, it is straightforward to see how SMEU representations associated with the RCTs in Figure \ref{complex_trees} capture different aspects of inequality. The upper-right RCT may correspond to an SMEU representation that combines features of the FATE and recursive representations. This captures a preference for intergenerational mobility, as in our earlier example, with multiple individuals in each generation. Note that this RCT differs from that of a recursive representation. The lower-left RCT combines features of the FETA and recursive representations: The population is divided into two social groups, each containing multiple generations, and the policymaker may value intergenerational mobility within groups. In this case, however, each generation within a group contains only one individual. The lower-right RCT extends this by combining all three---FATE, FETA, and recursive representations---and captures multiple social groups, multiple generations, and multiple individuals within each generation.

\section{Uniqueness of the Simplest RCT}\label{sect_unique}
To what extent are the RCT and the Bernoulli indices in an SMEU representation unique?\footnote{With an abuse of terminology, we refer to all $(u^A)_{A\in\cP_o}$ as Bernoulli indices.} Our first observation is that, fixing the RCT $\cT=(\cP,E)$ of an SMEU representation, the uniqueness of the corresponding Bernoulli indices $(u^A)_{A\in\cP}$ is similar to that in the expected utility theory: Roughly, for all $A\in\cP$, $x\in X_{\bar a(A)}$ and $v\in\bR^{c(A)}$, the Bernoulli index $u^A(x,\cdot,v)$ is unique up to a positive affine transformation. When $A=o$, the representation equation for $u^o$ does not involve an expectation operator. Therefore, $u^o$ is only unique up to a monotone transformation. These arguments are standard. We leave the details to \hyperref[OA_results]{Online Appendix \RN{3}.1}.

More importantly, we want to analyze the uniqueness of the RCT. We defined uniqueness of the RCT in Section \ref{sect_repre}. In general, however, the RCT in an SMEU representation need not be unique. For instance, if $\succsim$ has an expected utility representation with an additively separable Bernoulli index, then, because it is both linear in probabilities and separable across dimensions, every $\cT\in\mathbb{T}$ appears in some SMEU representation of $\succsim$.

Although the RCT is not unique in general, the \textit{simplest} RCT is uniquely identifiable. For each RCT $\cT=(\cP,E)$ and $i\in I$, define $\kappa(i)$ as the depth of the vertex containing $i$ in the RCT---that is, the number of edges from the root $o$ to that vertex. Formally, for all $i\in A\in \cP$, let $\kappa(i) = |a(A)|$. When needed, we add a superscript to mappings such as $c^{\cT}$, $a^{\cT}$, $\bar a^{\cT}$, $d^{\cT}$, $\bar d^{\cT}$, and $\kappa^{\cT}$ to emphasize their dependence on the underlying RCT.

\begin{defn}\label{def_simple}
    We say that an RCT $\cT=(\cP,E)$ is simpler than another RCT $\cT'=(\cP',E')$ if the following conditions hold:
    \begin{enumerate}
        \item For all $i\in I$, we have $\kappa^{\cT}(i)\leqslant \kappa^{\cT'}(i)$.
        \item If $\kappa^{\cT}(i)= \kappa^{\cT'}(i)$ for all $i\in I$, then $\cP'$ is finer than $\cP$.\footnote{For two partitions $\cP$ and $\cP'$ of some finite set $A$, we say that $\cP'$ is finer than $\cP$ if for every $A'\in\cP'$ there exists some $A\in \cP$ such that $A'\subseteq A$.}
    \end{enumerate}
\end{defn}

Definition \ref{def_simple} should be read lexicographically. The first condition concerns depth: $\cT$ must place every dimension weakly closer to the root than $\cT'$ does. The second condition is used only when the two RCTs have the same depth for every dimension. In that case, $\cT$ is simpler if its partition $\cP$ of $I$ into vertices is coarser than the corresponding partition $\cP'$ of $\cT'$. Thus, an RCT can be simpler because it is shallower, even if it is not coarser.


All RCTs in $\mathbb{T}(\succsim)$ appear in SMEU representations of the same preference. Simplicity is therefore a selection criterion, not an additional behavioral restriction. It selects one interpretation of the same choice behavior. The first condition says that we do not interpret the evaluation of a dimension as involving additional layers of recursive conditioning when the same preference can be represented without those layers. Once depths are fixed, the second condition says that we do not interpret dimensions as being evaluated separately in different vertices when the same preference can instead be represented by evaluating them jointly in a single vertex. The simplicity order is generally incomplete. Nevertheless, the following theorem shows that, among all RCTs consistent with the preference, there exists a unique simplest one.

\begin{theorem}\label{thm_unique_h2}
If $\succsim$ has an SMEU representation, then there exists a unique $\cT^*\in \mathbb{T}(\succsim)$ such that $\cT^*$ is simpler than $\cT$ for all $\cT\in \mathbb{T}(\succsim)$.
\end{theorem}

We call $\cT^*$ in the above theorem the \textit{simplest} RCT of $\succsim$. The theorem says that once we use Definition \ref{def_simple} to select among the RCTs in $\mathbb{T}(\succsim)$, the selected RCT is unique. Thus, even when the preference has several SMEU representations with different RCTs, the theorem identifies a unique benchmark RCT.


For example, suppose that $\succsim$ has an expected utility representation with an additively separable Bernoulli index. Then every $\cT\in\mathbb{T}$ appears in some SMEU representation of $\succsim$: Different RCTs provide different ways of interpreting the same preference. The simplest RCT is $\cT^*=(\cP^*,E^*)$, where $\cP^*=\{I\}$ and $E^*=\{(o,I)\}$. In this case, the theorem selects the RCT under which the preference is interpreted as involving neither nontrivial recursive conditioning nor separate evaluations across vertices: All dimensions are evaluated jointly in a single vertex.

\section{Axiomatic Characterization}\label{sect_axiom}
We introduce axioms on the decision maker's preference that characterize the SMEU representation. We begin with three standard ones.

\begin{axiom}\label{axiom_WO}
    (Weak Order) The preference $\succsim$ is complete and transitive.
\end{axiom}

\begin{axiom}\label{axiom_M}
    (Consequence Monotonicity) For all $x,y\in X$, if $x\geqslant  y$ and $x\neq y$, then $\delta_x\succ \delta_y$.
\end{axiom}

For each single dimension, we impose independence.

\begin{axiom}\label{axiom_one}
    (Unidimensional Independence) For all $i\in I$, $p,q,r\in\Delta(X_i)$, $x\in X_{-i}$, and $\alpha\in(0,1)$, we have $p\succ_x q\implies p\alpha r\succ_x q\alpha r$.
\end{axiom}

While independence holds within each dimension, it may fail across dimensions. The main axioms below describe such departures from independence using two binary relations over the set of dimensions $I$, both revealed from choice behavior. The first relation identifies when the correlation between two dimensions is behaviorally irrelevant.

\begin{defn}\label{def_correlation}
We say that $i\perp j$ if $i\neq j$ and for all $p,q\in \Delta(X)$ such that \emph{\text{(i)}} $p_{i|z}=q_{i|z}$ and $p_{j|z}=q_{j|z}$ for all $z\in \supp(p_{\{i,j\}^\mathsf{c}})$ and \emph{\text{(ii)}} $p_{\{i,j\}^\mathsf{c}} = q_{\{i,j\}^\mathsf{c}}$, we have $p\sim q$.
\end{defn}

The idea behind this definition is simple. Fix two dimensions $i$ and $j$. Suppose that two lotteries $p$ and $q$ have the same joint distribution over $\{i,j\}^{\mathsf c}$ and, conditional on each realization of $\{i,j\}^{\mathsf c}$, the same marginal distributions over $i$ and over $j$, respectively. Then the only remaining difference between $p$ and $q$ is the conditional correlation between dimensions $i$ and $j$. If the decision maker is indifferent between every such pair of lotteries, we write $i\perp j$. Thus, $i\perp j$ means that the conditional correlation between $i$ and $j$ is behaviorally irrelevant: Changing this correlation while holding the relevant marginal distributions fixed cannot affect choices.

The binary relation $\perp$ extends naturally to subsets of $I$.  For any disjoint sets $A,B\subseteq I$, we write $A\perp B$ if $i\perp j$ for all $i\in A$ and $j\in B$. In this case, correlations between dimensions in $A$ and dimensions in $B$ are behaviorally irrelevant in the pairwise sense. In addition, for any $B\subsetneq A\subseteq I$, we say that $B$ is \textit{minimally separable} from $A$ if $B\perp (A\setminus B)$ and $B'\not\perp (A\setminus B')$ for all nonempty $B'\subsetneq B$. There may be multiple minimally separable subsets of a given set $A$. If $B$ is minimally separable from $A$, then changes in the correlation between dimensions in $B$ and dimensions in $A\setminus B$ do not affect choices, while no smaller nonempty subset of $B$ has this property.

If $A\perp B$, how should the decision maker evaluate risk over $A\cup B$? A natural assumption is that the risks in $A$ and $B$ are evaluated separately. The next axiom imposes this separability requirement at the minimal level: If $B$ is minimally separable from $A$, then the risk in $B$ can be evaluated separately from the risk in $A\setminus B$. In deterministic environments, this idea parallels the classical notion of weak separability in \cite{goldman1964note} and \cite{Gorman1968}.


\begin{axiom}\label{axiom_S}
    (Uncorrelated Separability) For all $B\subsetneq A\subseteq I$ such that $B$ is minimally separable from $A$, we have $p\succsim_x q \iff (p_B, r)\succsim_x (q_B,r)$ for all $x\in X_{A^{\cp}}$, $r\in \Delta(X_{A\setminus B})$, and $p,q\in \Delta(X_A)$ such that $p_{A\setminus B}=q_{A\setminus B}$.\footnote{Recall that for any $p\in \Delta(X_A)$ and $B\subsetneq A$, $p_{A\setminus B}$ denotes the marginal distribution of $p$ on $A\setminus B$. Note that $p$ need not be written as $(p_{A\setminus B},p_B)$, because $p_{A\setminus B}$ and $p_B$ may be correlated. } 
\end{axiom}

Take any set $B$ that is minimally separable from $A$. Axiom \ref{axiom_S} requires the comparison between $p_B$ and $q_B$ to be invariant to both the risk in the remaining dimensions $A\setminus B$ and the correlation between $B$ and $A\setminus B$. The axiom therefore imposes a natural separability condition: The risk in $B$ can be evaluated separately from the risk in $A\setminus B$.

This separability property should be understood in the weak, rather than additive, sense. In deterministic environments, weak separability means that the ranking of alternatives within one block is independent of the values of variables outside that block. Under standard conditions, this property yields a representation in which each block is first evaluated separately and the resulting block utilities are then  combined through a general aggregator; the aggregator need not be additive \citep{goldman1964note,Gorman1968}. Additive separability, as in \cite{D1960additive}, requires stronger invariance conditions: Separability must hold not only for each individual block, but also for arbitrary unions of blocks. In this sense, Axiom \ref{axiom_S} follows the logic of  weak, rather than additive, separability.

The second binary relation identifies when choices are consistent with evaluating one dimension  conditional on realizations of another. 

\begin{defn} \label{order}
     We say that $i\hra j$ if the following implication holds. Suppose $x^k,y^k\in X$ and $\pi^k\in[0,1)$ for $k=1,\ldots,n$ satisfy  \emph{\text{(i)}} $\sum_{k=1}^n\pi^k=1$, \emph{\text{(ii)}} $(x^k_i)_{k=1}^n$ are pairwise distinct and so are $(y^k_i)_{k=1}^n$, and \emph{\text{(iii)}} for all $k$, $x^k_{\{i,j\}^\mathsf{c}} = y^k_{\{i,j\}^\mathsf{c}} $ and $ x^k\succsim y^k$. Then we have $\sum_{k=1}^n\pi^k\delta_{x^k}\succsim \sum_{k=1}^n\pi^k\delta_{y^k}$.
\end{defn}


Definition \ref{order} uses the behavioral content of the recursivity axiom in \cite{CE1991recursive} to define the revealed relation $\hra$. In \cite{CE1991recursive}, recursivity is a state-by-state dominance requirement on a recursively constructed choice domain. In their setting, a choice alternative is a temporal lottery, and after a state is realized, the continuation object is itself a temporal lottery. To see the behavioral content of the axiom, suppose there are $n$ possible states, occurring with probabilities $\pi^1,\dots,\pi^n$, and consider two temporal lotteries. If state $k\in\{1,\dots,n\}$ is realized, the first temporal lottery delivers $x^k$, while the second delivers $y^k$. Recursivity requires that, if $x^k\succsim y^k$ for all $k$, then the first temporal lottery is weakly preferred to the second. Thus, the axiom can also be viewed as a form of dynamic consistency: A preference that holds after every realized state must be respected before  the state is realized.\footnote{As noted by \cite{CE1991recursive}, recursivity is closely related to the notion of dynamic consistency in \cite{johnsen1985structure}.}

In our choice domain, which is not recursively constructed, Definition \ref{order} applies the same state-by-state dominance test locally.  The realizations of dimension $i$ play the role of states, and the test identifies when choices are consistent with evaluating dimension $j$ conditional on realizations of dimension $i$. The condition 
$x^k_{\{i,j\}^{\mathsf c}}=y^k_{\{i,j\}^{\mathsf c}}$ holds fixed all dimensions other than $i$ and $j$ in each state-by-state comparison. Thus, the test isolates the relationship between these two dimensions.

The pairwise-distinctness requirement on dimension $i$ is particularly important. Its role is clearest in the case of two ``states'' and two dimensions. In a recursively constructed domain, state labels remain part of the object: Even if two states lead to the same realization of dimension $i$, they are still treated as distinct states. In our domain, by contrast, suppose that $x_i^1=x_i^2=0$. Then the mixture $\frac{1}{2}\delta_{x^1}+\frac{1}{2}\delta_{x^2}$ records only that, conditional on dimension $i$ taking value $0$, dimension $j$ is distributed according to $\frac{1}{2}\delta_{x_j^1}+\frac{1}{2}\delta_{x_j^2}$. Thus, a decision maker who evaluates dimension $j$ conditional on dimension $i$ faces a single conditional risk, rather than two separate state-contingent realizations $x_j^1$ and $x_j^2$ of dimension $j$, as in the recursively constructed domain. Consequently, without pairwise distinctness, forming the mixture over $k$ can alter the conditional risk of dimension $j$ given dimension $i$. The state-by-state dominance test would then no longer isolate the intended conditional-evaluation relation. See Example \ref{eg2} for a related illustration of how mixtures may affect conditional risk.

We impose independence whenever it is compatible with $\hra$. Therefore, when dimension $i\in A$ can serve as the conditioning dimension for evaluating risk in $A$, we strengthen the state-by-state dominance test above to an independence requirement.

\begin{axiom}\label{axiom_RI}
    (Conditional Independence) For all $i\in A\subseteq I$, if $i\hra j$ for all $j\in A$, then for all $\alpha\in(0,1)$, $x\in X_{A^\mathsf{c}}$, and $p,q,r,s\in\Delta(X_A)$ such that $\supp(p_i)\cap \supp(r_i)=\emptyset$ and $\supp(q_i)\cap \supp(s_i)=\emptyset$, we have $p\succ_x q$ and $r\sim_x s\implies p\alpha r\succ_x q\alpha s$. 
\end{axiom}

Fix $i\in A$ such that $i\hra j$ for all $j\in A$, and fix $x\in X_{I\setminus A}$. Axiom \ref{axiom_RI} states that, for lotteries over $A$ with the dimensions in $I\setminus A$ fixed at $x$, a strict preference for $p$ over $q$ is preserved when $p$ and $q$ are mixed with indifferent lotteries $r$ and $s$, respectively. Thus, once $i\hra j$ for all $j\in A$, the axiom strengthens the state-by-state dominance test in Definition \ref{order} into an independence requirement for lotteries over $A$. The support-disjointness requirements play the same role as the pairwise-distinctness condition in Definition \ref{order}: They prevent the mixture from pooling different terms under the same realization of dimension $i$.

The next axiom imposes a consistency requirement on the two revealed relations $\hra$ and $\perp$.

\begin{axiom}\label{axiom_H}
    (RCT Perfection) For all $i^1,\dots, i^n\in I$ such that $i^k\not\perp i^{k+1}$ for all $k=1,\dots,n-1$, if $i^1\hra i^\ell$ and $i^n\hra i^\ell$ for all $\ell=2,\dots,n-1$, then either $i^1\hra i^n$ or $i^n\hra i^1$. 
\end{axiom}

To interpret RCT perfection, it is useful to begin with the case of two dimensions. In our framework, dimensions $1$ and $2$ can be evaluated in three ways: using the joint distribution of $(x_1,x_2)$; using the marginal distribution of one dimension together with the conditional distribution of the other; or using the two marginal distributions separately.\footnote{One could study richer models in which the distributions used in evaluation differ from those induced by the lottery. We do not pursue this possibility. In our model, whenever a marginal or conditional distribution is used, it is the one induced by the lottery being evaluated.} Axiom \ref{axiom_H} imposes the corresponding consistency requirement: If the correlation between dimensions $1$ and $2$ is behaviorally relevant, as captured by $1\not\perp 2$, then the evaluation must use information beyond their separate marginal distributions. Hence, the decision maker must either evaluate the two dimensions jointly  or use a marginal-and-conditional decomposition. In terms of the revealed relation $\hra$, this requires that at least one of $1\hra 2$ and $2\hra 1$ hold. The remaining case, in which neither relation is revealed, corresponds to an evaluation  based solely on the two marginal distributions.

Now consider three dimensions. Suppose $1\not\perp 2$, $2\not\perp 3$, $1\hra 2$, and $3\hra 2$. Intuitively, under these conditions, dimension $2$ is evaluated conditional on both dimensions $1$ and $3$. For each realized pair $(x_1,x_3)$, this conditional evaluation is based on the conditional distribution $p_{2|(x_1,x_3)}$. Comparing lotteries therefore requires considering both the conditional evaluation associated with each pair $(x_1,x_3)$ and the probability of that pair. These probabilities form the joint distribution $p_{\{1,3\}}$, which must therefore enter the comparison. This joint distribution can be incorporated either directly, by evaluating dimensions $1$ and $3$ jointly, or through a marginal-and-conditional decomposition, by evaluating one conditional on the other. In terms of the revealed relation $\hra$, this means that at least one of $1\hra 3$ and $3\hra 1$ must hold. This is precisely what Axiom \ref{axiom_H} requires in the three-dimensional case. The same logic applies to longer chains.\footnote{Although $\hra$ and $\perp$ cannot be exactly elicited using finite datasets, RCT perfection can be falsified using finite datasets. See \hyperref[falsify]{Online Appendix \RN{3}.4}.}

The next axiom is continuity. The standard continuity axiom requires that for every lottery $p$, the set of lotteries that are weakly better than $p$ and the set of lotteries that are weakly worse than $p$ are closed. This notion of continuity may be too demanding in our theory. Consider the following example.
 
\begin{example}
    \label{eg2}Let $N=2$. Suppose the decision maker's utility function is
    $$U(p)=\sum_{x_1}p_1(x_1)u(x_1, p_{2|x_1}),$$
    in which $u(x_1,p_2)=x_1+(\sum_{x_2}p_2(x_2)\sqrt{x_2})^2$. That is, the utility of $(x_1,p_2)$ equals the sum of $x_1$ and the certainty equivalent of $p_2$ under the square-root Bernoulli index.
    
    Consider a lottery that yields $(\varepsilon,0)$ and $(0,4)$ with equal probability. As $\varepsilon$ converges to $0$ its utility converges to $2$, but the lottery converges in distribution to $(\delta_0, q_2)$ with utility $u(0,q_2)=1$, in which $q_2=\delta_{0}\frac{1}{2}\delta_4$.
\end{example}

This discontinuity arises because, in the limit, the mixture changes the conditional risk of dimension $2$ given realization $0$ of dimension $1$. For any nonzero $\varepsilon$, conditional on $x_1=0$, the lottery for dimension $2$ is $\delta_4$. By contrast, when $\varepsilon$ is $0$, the lottery for dimension $2$ conditional on $x_1=0$ becomes $\tfrac{1}{2}\delta_0+\tfrac{1}{2}\delta_4$. Thus, the discontinuity comes from the change in the conditional risk used to evaluate dimension $2$.\footnote{It is possible to test whether people's choice behavior exhibits this kind of discontinuity. For example, in the multisource income setting with initial wealth (see Section \ref{section_app}), one could vary the correlation between initial wealth and income and use econometric methods to test whether such changes systematically affect observed choice behavior.}

This contrasts with standard dynamic choice models such as \cite{KrepsPorteus78},  \cite{EpsteinZin89}, and \cite{CE1991recursive}. In those models, choice alternatives are constructed recursively according to an exogenous linear order of dimensions that represent time periods. This domain structure simplifies the characterization and makes it straightforward to impose continuity. In our choice domain, by contrast, the recursive structure is not built into the domain; it is revealed from choice behavior. As Example \ref{eg2} illustrates, this difference can make full continuity too strong for SMEU.

Nonetheless, the following weaker notion of continuity is independent of the observation behind Example \ref{eg2} and should remain valid in our theory.

\begin{axiom}\label{axiom_C}
    (Weak Continuity) For all $p,q,r\in \Delta(X)$, the sets $\{\alpha\in[0,1]:p\alpha q\succsim r\}$ and $\{\alpha\in[0,1]:r\succsim p\alpha q\}$ are closed in $[0,1]$, and the sets $\{x\in X:\delta_x\succsim p\}$ and $\{x\in X:p\succsim \delta_x\}$ are closed in $X$.
\end{axiom}

The main representation theorem is below. 

\begin{theorem} \label{thm_main}
   The preference satisfies weak order, consequence monotonicity, unidimensional independence, uncorrelated separability, conditional independence, RCT perfection, and weak continuity if and only if it has an SMEU representation.
\end{theorem}

\medskip
\textbf{Revealed Relations and the RCT.}
We conclude this section by clarifying the relationship between the RCT structure and the revealed binary relations $\hra$ and $\perp$. A natural first interpretation is that $i\hra j$ corresponds to dimension $i$ being an ancestor of $j$ in the RCT, while $i\perp j$ corresponds to dimensions $i$ and $j$ being placed in different branches. This interpretation is correct in one direction. If an SMEU representation has $i$ as an ancestor of $j$ in the RCT, then $i\hra j$; and if it places $i$ and $j$ in different branches, then $i\perp j$ (see \lemmaref{lemma_nece_hra} and \lemmaref{lemma_nece_perp}).

The converse is more subtle. A revealed relation may be redundant: It may fail to correspond to the associated tree relation in any SMEU representation. Nevertheless, such redundancy is limited. If $i\not\perp j$, then $i$ and $j$ must remain in the same branch in every SMEU representation. Moreover, if in addition $i\hra j$, then there exists an SMEU representation in which $i$ is an ancestor of $j$ in the RCT. Conversely, if neither $i\hra j$ nor $j\hra i$ holds, then $i$ and $j$ must be placed in different branches in every SMEU representation. Thus, redundancy can arise only when both $i\hra j$ and $i\perp j$ hold. In that case, one of the two revealed relations may be redundant for constructing the RCT.

To see how such redundancy is resolved, and to prepare for the examples below, it is useful to understand the construction of the RCT behind the proof of \thmref{thm_main}. At any given step, suppose that the dimensions under consideration form a set $A\subseteq I$. The construction proceeds in one of two ways. First, if there exists a nonempty subset $B\subsetneq A$ such that $B\perp(A\setminus B)$, then the minimally separable subsets of $A$ form a partition of $A$ (see \lemmaref{lemma_finest}). The construction places the elements of this partition in different branches. Thus, whenever two dimensions $i$ and $j$ belong to different elements of the partition, the natural first interpretation of $i\perp j$ is realized in the RCT. If they belong to the same element, this step does not determine their relationship, and the construction continues within that element.

The second way to proceed uses $\hra$. If there exists a dimension $i\in A$ such that $i\hra j$ for all $j\in A$, let $M$ denote the set of all such dimensions. The construction can place $M$ in one vertex, with the remaining dimensions in $A$, if any, placed in descendant subtrees. Now consider $i,j\in A$ with $i\hra j$. If $i\in M$ and $j\notin M$, the natural first interpretation of $i\hra j$ is realized in the RCT. If $j\in M$, then $i\in M$ as well by \lemmaref{lemma_transitive}; in this case, the RCT places $i$ and $j$ in the same vertex, reflecting the fact that both $i\hra j$ and $j\hra i$ are present, rather than selecting one of them as a strict ancestor relation. If neither $i$ nor $j$ belongs to $M$, this step does not determine their relationship, and the construction continues with the remaining dimensions $A\setminus M$.

Axioms \ref{axiom_WO}--\ref{axiom_C} ensure that this recursive construction is well defined. At every non-singleton set $A\subseteq I$ reached in the construction, at least one of the two ways described above is applicable (see \lemmaref{lemma_suff_SUB}). Hence the construction can be iterated. If two dimensions $i$ and $j$ remain together until the construction reaches the set $\{i,j\}$, then one of the two ways must again apply to this binary set. The first way realizes the natural first interpretation of $i\perp j$, while the second realizes that of either $i\hra j$ or $j\hra i$.

We now illustrate the two possible forms of redundancy. The first example shows that $i\hra j$ need not imply that $i$ is an ancestor of $j$ in the RCT.

\begin{example}\label{eg_hra} 
Let $N=5$. Suppose the decision maker's utility function is 
\begin{align*}
        U(p)= \bE^{p}_{\{1,3\}}\:x_1x_3 
        + \bE^p_1\:u(x_1,p_{5|x_1}) 
        + \bE^p_2\:u(x_2,p_{4|x_2}),
\end{align*}
in which $u$ takes the functional form in Example \ref{eg2} (the first argument plus the certainty equivalent of the second under the square-root Bernoulli index). We can verify that $2\hra 3$. However, dimension $2$ is not an ancestor of $3$ in the RCT. To see why, start at the root. Consider  whether the second way of the construction applies: Is there a dimension $i$ such that $i\hra j$ for all $j\in I$? Dimension $1$ is the natural candidate. It is linked to dimension $3$ through the first term and to dimension $5$ through the conditional evaluation of $5$ given $1$. However, dimension $4$ is evaluated conditional on dimension $2$, not on dimension $1$. In fact, there is no dimension $i$ such that $i\hra j$ for all $j\in I$.

Thus, the construction cannot proceed by the second way at the root. It therefore proceeds by the first way. The minimally separable subsets of $I$ are $\{1,3,5\}$ and $\{2,4\}$, which implies that dimensions $2$ and $3$ are placed in different branches. Thus, the revealed relation $2\hra 3$ is redundant for constructing the RCT: In any SMEU representation, the vertex containing dimension $2$ is not an ancestor of the vertex containing dimension $3$.
\end{example}

The second example illustrates the opposite form of redundancy: $i\perp j$ need not imply that $i$ and $j$ are placed in different branches.

\begin{example}\label{eg_perp} 
Let $N=4$. Suppose the decision maker's utility function is 
\begin{align*}
        U(p)= \bE^{p}_{\{1,2\}}\:x_1x_2 
        + \bE^{p}_{\{2,3\}}\:x_2x_3 
        + \bE^p_1\:u(x_1,p_{4|x_1}),
\end{align*}
in which $u$ again takes the functional form in Example \ref{eg2}. We can verify that $1\perp 3$. However, dimensions $1$ and $3$ are not placed in different branches in the RCT. To see why, start again at the root. If the first way of the construction were available, the minimally separable subsets of $I$ would form a nontrivial partition of $I$. Consider the element of this partition that contains dimension $1$. Because the first and third terms of the utility function are chosen so that $1\not\perp 2$ and $1\not\perp 4$, respectively, this element must also contain dimensions $2$ and $4$. Since it contains dimension $2$ and $2\not\perp 3$, it must also contain dimension $3$. Hence this element must contain all four dimensions, so no nontrivial separability partition is possible at the root.

The construction therefore proceeds by the second way. In this example, dimension $1$ serves as the conditioning dimension at the root. Thus, the revealed relation $1\perp 3$ is redundant for constructing the RCT: In any SMEU representation, the vertex containing dimension $1$ and the vertex containing dimension $3$ are not placed in different branches; rather, the former is always an ancestor of the latter.
\end{example}

\medskip
\textbf{RCTs and DAGs.}
We now make precise the connection between RCTs and DAGs noted in the Introduction. At the level of its vertices, an RCT is itself a DAG: It is a directed rooted tree whose non-root vertices contain subsets of dimensions.

There is a second connection at the level of individual dimensions. Given an RCT, we can construct the following directed graph in which each dimension is a vertex. For any distinct dimensions $i,j\in I$, include a directed edge from $i$ to $j$ if and only if either $i$ and $j$ belong to the same RCT vertex, or the RCT vertex containing $i$ is an ancestor of the vertex containing $j$. Both cases have the same interpretation: Dimension $j$ can be evaluated conditional on the realization of dimension $i$, and both the conditional risk and the utility used in this evaluation may depend on that realization.

There are three possible types of relation between any two dimensions. If both directed edges between $i$ and $j$ are present, the two dimensions belong to the same RCT vertex and are evaluated jointly. If exactly one directed edge is present, say from $i$ to $j$, the RCT vertex containing $i$ is an ancestor of the vertex containing $j$, and dimension $j$ is evaluated conditional on dimension $i$. If neither edge is present, the two dimensions lie in different branches and are evaluated separately; their correlation is behaviorally irrelevant. Thus, joint, conditional, and separate evaluation arise from the same directional relation, as explained in the Introduction. Clearly, the RCT can be recovered from the resulting directed graph over dimensions.

This construction illustrates the second connection between RCTs and the DAGs used in the literature cited in Footnote \ref{causality}. The commonality is structural: In both frameworks, the presence or absence of a directed edge records whether the corresponding relation is part of the decision maker's subjective representation. In DAG models of imperfect causal perception, the edges capture the probabilistic or causal relations that the decision maker takes into account. In our model, an edge from dimension $i$ to dimension $j$ records that dimension $j$ can be evaluated conditional on the realization of dimension $i$. The edges therefore describe how multidimensional risk is evaluated; they need not be interpreted as an objective causal structure or as a mistaken perception of one.

At the graph-theoretic level, Axiom \ref{axiom_H} is analogous to the perfection, or no-collider, property in DAG models of perceived causality: In that literature, any two perceived causes of the same variable must themselves be connected; see, among others, \cite{spiegler20} and \cite{EllisThysen24}.

\section{Special Cases}
\label{sect_special}
Our framework allows us to characterize several well-known utility representations in a novel way and to derive useful generalizations of them.

We begin by characterizing the FATE, FETA, and recursive representations in Definition \ref{def_special1}. While expected utility theory already provides a characterization of the FATE representation, we offer an alternative characterization within our framework. This new perspective on the FATE representation facilitates the derivation of its natural generalization. The same applies to the FETA and recursive representations.

\begin{theorem}\label{thm_special}
    Suppose the preference $\succsim$ has an SMEU representation. It has a FATE representation if and only if $i\hra j$ for all $i,j\in I$. It has a FETA representation if and only if $i\hra j\iff i=j$ for all $i,j\in I$. It has a recursive representation if and only if $\hra$ is a linear order.
    
\end{theorem}

Note that the gap between the RCT structure and the revealed binary relations $\hra$ and $\perp$ noted after \thmref{thm_main} is resolved in these three extreme cases.

In the FATE representation, $i\hra j$ for all $i,j\in I$. Thus, at the root, every dimension can serve as a conditioning dimension for every other dimension. The second way in the construction described after \thmref{thm_main} can therefore place all dimensions in a single vertex, yielding an RCT with $c(o)=\{I\}$. This is exactly the FATE representation. The bidirectional nature of $\hra$ is represented by evaluating the relevant dimensions in the same vertex. In addition, if $k\perp l$ for some $k,l\in I$, this relation is not redundant either. When $i\hra j$ for all $i,j\in I$, one can construct an RCT in which $k$ and $l$ remain together until the construction reaches the binary set $\{k,l\}$. At that binary set, the coexistence of $k\perp l$ and $k\hra l$ implies an additively separable expected utility evaluation over these two dimensions. Hence there is an RCT that places $k$ and $l$ in different branches. Thus, in the FATE representation, neither $\hra$ nor $\perp$ is redundant.

In the FETA representation, $i\hra j\iff i=j$. Hence, for any two distinct dimensions $i$ and $j$, neither $i\hra j$ nor $j\hra i$ holds. By the discussion after \thmref{thm_main}, $i$ and $j$ must then be placed in different branches in every SMEU representation. Since this holds for all distinct dimensions, the unique RCT has all dimensions separated at the root: $c(o)=\{\{1\},\dots,\{N\}\}.$
This is exactly the FETA representation. Thus, there is no nontrivial $\hra$ relation to be redundant, and every revealed $\perp$ relation is represented by branch separation.

Finally, suppose that $\hra$ is a linear order. Relabel the dimensions so that $i\hra j$ if and only if $i\leqslant j$.\footnote{This relabeling is only notational. In a recursive representation, the revealed order $\hra$ need not coincide with any preexisting order of the dimensions. We provide additional results in \hyperref[time_order]{Online Appendix \RN{3}.2} for the case in which, under the original labeling, $i\hra j$ if and only if $i\leqslant j$. This case may be relevant, for example,  when the dimensions represent ordered time periods.} The second way in the construction can then be applied repeatedly: Dimension $1$ is placed as an ancestor of the remaining dimensions, dimension $2$ is placed as an ancestor of the remaining dimensions after dimension $1$ is removed, and so on. This yields the recursive RCT: $c(o)=\{\{1\}\}$ and $c(\{i\})=\{\{i+1\}\}$ for $i=1,\dots,N-1$. If $k\perp l$ for some $k,l\in I$ with $k<l$, then one can show that this would imply the reverse revealed relation $l\hra k$, but since $k<l$ implies $k\hra l$, this contradicts $\hra$ being a linear order. Thus, in the recursive case, the revealed order $\hra$ coincides with the ancestor order in the unique RCT, and no $\perp$ relation exists.


We can use the same approach to derive and characterize other useful special cases of the SMEU representation, which naturally generalize the FATE, FETA, and recursive representations. Examples of utility functions corresponding to these generalizations already appeared in Section \ref{sect_ineq}, with their associated RCTs illustrated in Figure \ref{complex_trees}.

\begin{defn}\label{def_special2}
    The preference $\succsim$ has
    \begin{itemize}
        \item a generalized bracketing representation if it has an SMEU representation such that $\bigcup_{A\in c(o)}A=I$;
        \item a generalized recursive representation if it has an SMEU representation such that $|c(A)|\leqslant 1$ for all $A\in \cP_o $;
        \item a recursive bracketing representation if it has an SMEU representation such that $|A|=1$ for all $(\cP,E)\in \mathbb{T}(\succsim)$ and $A\in\cP$.
    \end{itemize}
\end{defn}
    
These three cases combine the three extreme representations in pairwise ways. A generalized bracketing representation combines the FATE and FETA representations: Dimensions are partitioned into brackets, risk is evaluated jointly within each bracket, and the resulting bracket-level evaluations are then aggregated across brackets. This case is useful in applications in which some dimensions are naturally evaluated together while others are evaluated separately. The multisource-income application in Section \ref{section_app} is based on this representation.

A generalized recursive representation combines the FATE and recursive representations: The RCT remains a chain, but each vertex may contain multiple dimensions. This case is useful whenever recursive evaluation is appropriate but a single vertex in the recursive chain may involve several dimensions. In dynamic settings, for example, this flexibility allows a vertex to contain several dates, rather than requiring the vertices to be the standard equally spaced calendar periods imposed by the modeler. When evaluating a risky consumption stream, a decision maker may treat today as one vertex, the rest of the week as another, and the rest of the month as a third. Which dates belong to the same vertex can, in principle, be estimated from choices.

A recursive bracketing representation combines the FETA and recursive representations: Each vertex contains a single dimension, but the vertices need not form a single chain. Relative to the FETA representation, it allows ancestor relations among some dimensions; relative to the recursive representation, it allows branching. This case is useful when risk is evaluated at the level of individual dimensions, but the RCT need not impose a single recursive order over all dimensions: Some dimensions may lie on the same chain, while others lie in different branches. The next result characterizes the above representations.

\begin{theorem}\label{thm_special2}
    Suppose the preference $\succsim$ has an SMEU representation. It has a generalized bracketing representation if and only if $\hra$ is symmetric. It has a generalized recursive representation if and only if $\hra$ is complete. It has a recursive bracketing representation if and only if $\hra$ is antisymmetric.
    \end{theorem}

The characterization is easiest to see from the RCT restrictions in Definition \ref{def_special2}. For any preference with an SMEU representation, it can be shown that $\hra$ is reflexive and transitive. Thus, imposing symmetry, completeness, or antisymmetry yields three familiar structures. If $\hra$ is symmetric, then it is an equivalence relation on dimensions. Its equivalence classes are exactly the brackets in a generalized bracketing representation. If $\hra$ is complete, then it has the same structure as a complete and transitive preference relation: Dimensions satisfying both $i\hra j$ and $j\hra i$ form classes analogous to indifference classes, and these classes are ordered by $\hra$. These ordered classes are exactly the vertices of the chain in a generalized recursive representation. Finally, if $\hra$ is antisymmetric, then two distinct dimensions cannot be placed in the same vertex, because dimensions in the same vertex would satisfy both $i\hra j$ and $j\hra i$. This yields the recursive bracketing representation.

\section{Generalized Bracketing in Multisource Income}\label{section_app}
In this section, we apply the generalized bracketing representation to multisource income, focusing on two implications: how bracketing is related to stochastic dominance and whether the decision maker prefers a given distribution of aggregate income to be concentrated in a single source rather than spread across several sources.


Let $X_i=Z=[-b,b]$ for every $i \in I$, with $b > 0$. Denote $Z_+ = [0, b]$ and $Z_- = [-b, 0]$. Each element of $Z$ represents a monetary outcome, which can indicate either a gain or a loss depending on its sign. Each dimension $i \in I$ represents an income source, such as salary, investment returns, or rental income.
For any $A \subseteq I$ and $p \in \Delta(X_A)$, let $f[p]$ denote the distribution of total income generated by dimensions in $A$. Formally, $f[p](z)=\sum_{x\in X_A} p(x)\mathds{1}\{\sum_{i\in A}x_i=z\}$, in which $\mathds{1}$ is the indicator function. When $A = I$, $f[p]$ represents the distribution of total income. Throughout this subsection, we assume that the relevant utility index used to evaluate such $f[p]$ is defined on $[-Nb, Nb]$.

Assume that $\succsim$ has the following generalized bracketing representation:
\begin{equation}\label{eq_money3}
    U(p)= \sum_{i=1}^n\text{CE}(f[p_{A_i}],\:v^{A_i}),
\end{equation}
in which $\{A_i\}_{i=1}^n$ is a partition of $I$, representing how the decision maker brackets different income sources.\footnote{In \hyperref[OA_money]{Online Appendix \RN{2}}, we axiomatize the representation  (\ref{eq_money3}). Unlike the standard additive expected utility formulation 
$\hat{U}(p)=\sum_{i=1}^n \bE^{f[p_{A_i}]}\:v^{A_i}(x)$ used in \cite{V2023bracketing} and \cite{Camara21}, (\ref{eq_money3}) sums certainty equivalents. The difference matters even without risk: With two singleton brackets and strictly concave $v$, additive expected utility ranks $(1,1)$ above $(2,0)$ because $2v(1)>v(0)+v(2)$, although both vectors deliver total income $2$ with certainty; by contrast, (\ref{eq_money3}) assigns both utility $2$.} The function $v^{A_i}$ is defined on $[-|A_i|b, |A_i|b]$, continuous, and strictly increasing, capturing the decision maker's risk attitude toward income risk within bracket $A_i$. Assume further that each $v^{A_i}$ is twice continuously differentiable, so that the Arrow--Pratt measure of absolute risk aversion is well defined.

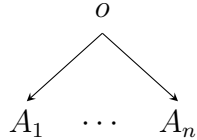
\begin{figure}[ht]
    \centering
    \begin{tikzpicture}[scale=1]
        \node at (0,-1.5) {$o$};

        \node at (-1,-3.0) {$A_1$};
        \node at ( 0,-3.0) {$\cdots$};
        \node at ( 1,-3.0) {$A_n$};

        \draw[-stealth] (0,-1.8) -- (-1,-2.7);
        \draw[-stealth] (0,-1.8) -- ( 1,-2.7);
    \end{tikzpicture}
    \caption{The RCT associated with equation (\ref{eq_money3}).}
    \label{multisource_tree}
\end{figure}

Under this representation, the decision maker partitions income sources into brackets, computes the certainty equivalent of total income within each bracket, and then aggregates these certainty equivalents using summation.





\subsection{Dominance and Bracketing: Beyond the Narrow Case} \label{sect_sd}

\cite{RabinWeizsacker09} show that decision makers frequently violate stochastic dominance in risky choice tasks, and that such violations are consistent with a form of narrow bracketing. The generalized bracketing representation generalizes this insight: It accommodates a broader range of bracketing structures and allows comparative statics on how behavior changes as brackets become finer or coarser.

For any $p,q\in \Delta(Z)$, we say that $p$ \textit{(first-order) stochastically dominates} $q$, denoted by $p\succ_{FOSD}q$, if $p\neq q$ and $\sum_{x\leqslant z}q(x)\geqslant  \sum_{x\leqslant z}p(x)$ for all $z\in Z$. We say that $\succsim$ satisfies \textit{dominance} if $f[p]=f[q]$ implies $p\sim q$ and $f[p]\succ_{FOSD} f[q]$ implies  $p\succ q$ for all $p,q\in\Delta(X)$. Although commonly imposed in economics, this property is often violated in experiments.\footnote{Relatedly, in dynamic choice settings, \cite{ecta2017recursive} also show that many widely used recursive models fail to satisfy a version of dominance.} Consider the following example.

\begin{example}\label{eg_dominance}
Consider the following pair of decisions. The risks in the two decisions are resolved independently, and both choices affect your overall payment. Please examine both decisions and indicate your preferred options.\smallskip

	{\it Decision 1: Choose between}
	
	~~~~~{\it A. A sure gain of \$2.40.}
	
	~~~~~{\it B. A 25\% chance to gain \$10.00, and a 75\%  chance to gain \$0.}

	{\it Decision 2: Choose between}
	
	~~~~~{\it C. A sure loss of \$7.50.}
	
	~~~~~{\it D. A 75\%  chance to lose \$10.00, and a 25\% chance to lose \$0.}
\end{example}
\cite{RabinWeizsacker09} and \cite{TverskyKahneman81} show that at least 28\% of participants choose $A$ in decision 1 and $D$ in decision 2. This choice pattern, however, violates dominance: The total-income distribution from choosing $B$ and $C$ stochastically dominates that from choosing $A$ and $D$,
$$\frac{3}{4} \delta_{-7.50} + \frac{1}{4} \delta_{2.50}\succ_{FOSD}\frac{3}{4} \delta_{-7.60} + \frac{1}{4} \delta_{2.40},$$
since the former always yields $\$0.1$ more. Such violations cannot be rationalized by models that evaluate lotteries based on the distribution of total income---including many that may violate dominance.\footnote{See, among others, \cite{OR1985disappointment};  \cite{RES1986disappointment}; \cite{AER2007reference}; \cite{M2022simplicity}; and \cite{ P2022simplicity}.} By contrast, evaluating the two decisions separately makes the pattern intuitive: Choosing $A$ reflects risk aversion over gains, while choosing $D$ reflects risk seekingness over losses; both are well-documented behaviors.

The next result formalizes the tension between bracketing and dominance: Satisfying dominance leaves no room for any nontrivial form of bracketing, not merely narrow bracketing. Indeed, dominance holds only when equation \eqref{eq_money3} involves a single bracket.

\begin{proposition}\label{prop_dominance}
    Suppose $\succsim$ is represented by (\ref{eq_money3}). Then $\succsim$ satisfies dominance if and only if it is represented by $U(p)=\mathbb{E}^{f[p]}\:v$ for some $v$.
\end{proposition}

By contrast, if we restrict attention to lotteries in which income sources are independent across all dimensions---as in Example \ref{eg_dominance}---then dominance can still hold under nontrivial bracketing. Again, this is not limited to narrow bracketing: It holds for any bracketing structure in (\ref{eq_money3}), provided that the Bernoulli indices exhibit constant absolute risk aversion (CARA).\footnote{We refer to increasing and decreasing absolute risk aversion as IARA and DARA, respectively.} This result, made possible by our more flexible SMEU framework, generalizes the insights in
\cite{RabinWeizsacker09} and \cite{mpst2020background}.

We say that $\succsim$ satisfies \textit{independent-sources dominance} if  $f[p]=f[q]$ implies $p\sim q$ and $f[p]\succ_{FOSD} f[q]$ implies $p\succ q$ for all $p,q\in \Delta(X)$ such that $p=(p_1,\dots,p_N)$ and $q=(q_1,\dots,q_N)$.

\begin{proposition}\label{prop_dominance2}
    Suppose $\succsim$ is represented by (\ref{eq_money3}). Then $\succsim$ satisfies independent-sources dominance if and only if it is represented by either $U(p)=\mathbb{E}^{f[p]}\:v$ for some $v$, or $U(p)=\sum_{i=1}^n \text{CE}(f[p_{A_i}],\:v)$  for some $v$ that exhibits CARA.
\end{proposition}

Given the above results, it is natural to ask whether coarser bracketing structures make it easier for preferences to satisfy dominance. Since a preference may have multiple generalized bracketing representations, we focus on the one with the \textit{coarsest} bracketing structure. This coarsest structure corresponds to $c^{\cT^*}(o)$, in which $\cT^*$ denotes the simplest RCT of $\succsim$ characterized in \thmref{thm_unique_h2}.

For any subset of income sources $A\subseteq I$, we say that $\succsim$ satisfies \textit{dominance on $A$} if  $f[p]=f[q]$ implies $p\sim q$ and $f[p]\succ_{FOSD} f[q]$ implies  $p\succ q$ for all $p,q\in \Delta(X)$ such that $p_i=q_i=\delta_0$ for all $i\not\in A$. For two decision makers with preferences $\succsim^1$ and $\succsim^2$, respectively, we say that $\succsim^1$ is \textit{more dominance-consistent} than $\succsim^2$ if, for all $A\subseteq I$, whenever $\succsim^2$ satisfies dominance on $A$, so does $\succsim^1$.

\begin{proposition}\label{prop_dominance3}
    Suppose $\succsim^1$ is represented by (\ref{eq_money3}) with the coarsest bracketing structure of income sources $\{A_i\}_{i=1}^n$.
    \begin{enumerate}
        \item For all $B\subseteq I$, $\succsim^1$ satisfies dominance on $B$ if and only if $B\subseteq A_i$ for some $i$.
        \item Suppose $\succsim^2$ is also represented by (\ref{eq_money3}) with the coarsest bracketing structure of income sources $\{A'_i\}_{i=1}^m$. Then $\succsim^1$ is more dominance-consistent than $\succsim^2$ if and only if $\{A'_i\}_{i=1}^m$ is finer than $\{A_i\}_{i=1}^n$.
    \end{enumerate}
\end{proposition}

This result formalizes the idea that coarser bracketing of income sources makes dominance easier to satisfy.

\subsection{Avoidance of Multidimensional Risk}\label{sect_amr}
Decision makers may have clear preferences over whether income comes from a single source or  multiple sources.\footnote{For instance, \cite{H2021issues} studies the case in which integrating risk across sources can be cognitively demanding, potentially leading individuals to avoid choices that involve multidimensional risk.} To formalize this idea, we say that $p\in \Delta(X)$ is a \textit{source-$i$} lottery if there exists $r\in \Delta(Z)$ such that $p_i=r$ and $p_j=\delta_0$ for all $j\neq i$.  We say that $\succsim$ exhibits \textit{avoidance of multidimensional risk} if $q\succsim p$ for all $i\in I$, $p=(p_1,\dots,p_N)\in \Delta(X)$, and source-$i$ lottery $q$ satisfying $f[q]=f[p]$.\footnote{If we strengthen this definition by removing the requirement that $p=(p_1,\dots,p_N)$, we can show that the decision maker must place all income sources into one bracket.} Similarly, we say that $\succsim$ exhibits avoidance of multidimensional risk in gains (resp.\ losses) if the same condition holds for all $p=(p_1,\dots,p_N)\in \Delta(Z_+^{N})$ (resp.\ $p=(p_1,\dots,p_N)\in \Delta(Z_-^{N})$).

\begin{proposition}\label{prop_avoidance}
    Suppose $\succsim$ is represented by (\ref{eq_money3}). Then:
    \begin{enumerate}
        \item It satisfies avoidance of multidimensional risk if and only if  it is represented by either \emph{(i)} $U(p)=\mathbb{E}^{f[p]}\:v$ for some $v$ or \emph{(ii)} $U(p)=\sum_{i=1}^n \text{CE}(f[p_{A_i}],\:v)$ for some $v$ that exhibits CARA. 
        \item It satisfies avoidance of multidimensional risk in gains if and only if  it is represented on $\Delta(Z_+^{N})$ by either \emph{(i)} $U(p)=\mathbb{E}^{f[p]}\:v$ for some $v$ or \emph{(ii)} $U(p)=\sum_{i=1}^n \text{CE}(f[p_{A_i}],\:v)$ for some $v$ that exhibits DARA. 
        \item It satisfies avoidance of multidimensional risk in losses if and only if  it is represented  on $\Delta(Z_-^{N})$ by either \emph{(i)} $U(p)=\mathbb{E}^{f[p]}\:v$ for some $v$ or \emph{(ii)} $U(p)=\sum_{i=1}^n \text{CE}(f[p_{A_i}],\:v)$ for some $v$ that  exhibits IARA.
    \end{enumerate}
\end{proposition}

This result applies not only to narrow bracketing, but also to general forms of bracketing. An immediate implication of \propref{prop_avoidance} is that if the decision maker \textit{strictly} prefers to consolidate some risky multisource income into a single source, then she must exhibit nontrivial bracketing.

Next, we turn to a comparative notion of avoidance of multidimensional risk. \propref{prop_avoidance} implies that, on the full domain $\Delta(X)$, strict avoidance of multidimensional risk cannot arise. To obtain a meaningful comparative notion, we therefore restrict attention below to gains. The case of losses is symmetric and hence omitted.

Consider two decision makers whose preferences $\succsim^1$ and $\succsim^2$ are represented by (\ref{eq_money3}). Let $\{(A_i,v^{A_i})\}_{i=1}^n$ and $\{(B_i,\hat v^{B_i})\}_{i=1}^m$ denote their coarsest bracketing structures and Bernoulli indices, respectively.
If $\succsim^2$ satisfies avoidance of multidimensional risk in gains, then \propref{prop_avoidance} implies that either $m=1$ or we can normalize $\hat v^{B_i}\equiv \hat v$ on $x\in Z_+$ for some $\hat v$ that exhibits DARA. We say that $\succsim^1$ exhibits \textit{stronger avoidance of multidimensional risk in gains} than $\succsim^2$ if $q\succsim^2 p \Longrightarrow q\succsim^1 p$ for all $i\in I$, $p=(p_1,\dots,p_N)\in \Delta(Z_+^{N})$, and source-$i$ lottery $q\in \Delta(Z_+^{N})$.

\begin{proposition}\label{prop_OA_compare}
    Suppose $\succsim^2$ is represented on $\Delta(Z_+^{N})$ by (\ref{eq_money3}) with parameters $\{(B_i,\hat v)\}_{i=1}^m$ in which $\hat v$ exhibits strictly DARA for $x\in Z_+$. Then for any $\succsim^1$ that is also represented on $\Delta(Z_+^{N})$  by (\ref{eq_money3}),  it exhibits stronger avoidance of multidimensional risk in gains than $\succsim^2$ if and only if $\succsim^1$ can be represented  on $\Delta(Z_+^{N})$ by (\ref{eq_money3}) with parameters $\{(A_i,\hat v)\}_{i=1}^n$ in which $\{A_i\}_{i=1}^{n}$ is finer than $\{B_i\}_{i=1}^{m}$.
\end{proposition}

This result thus formalizes the idea that finer bracketing of income sources leads to stronger avoidance of multidimensional risk in gains.

\medskip
\textbf{Further Extensions.} The analysis above focuses on risk across income sources and abstracts from background risk. The SMEU framework can also incorporate risky background wealth while preserving the generalized bracketing structure across income sources. For example, let dimension $1$ represent background wealth, and suppose that each income bracket is evaluated conditional on its realization:
\begin{equation}\label{eq_money}
    U(p)= \bE^p_1\:u(w\:+\:\sum_{i=1}^n\text{CE}(f[p_{A_i|w}],\:v^{A_i}_w)),
\end{equation}
in which $\{A_i\}_{i=1}^n$ is a partition of $I\setminus\{1\}$, representing how the decision maker brackets different income sources. The functions $u$ and $v^{A_i}_w$ capture the decision maker's risk attitudes toward background risk and the income risk within each bracket $A_i$, respectively. We allow $v^{A_i}_w$ to depend on both background wealth and the specific income bracket. This specification provides a tractable and flexible way to study how background risk, wealth-dependent within-bracket risk attitudes, and the degree of integration between income brackets and background wealth affect choice behavior.

This recursive-conditioning perspective also bears on calibration questions in risk taking. \cite{R2000calibration} argues that, under expected utility theory, rejecting modest small-stakes gambles over a wide range of wealth levels implies implausibly strong risk aversion toward large-stakes gambles. Reference dependence and narrow framing have been proposed as responses to this critique (see \cite{R2000calibration} and \cite{BHT2006bracketing}). In our terminology, narrow framing corresponds to a FETA representation that evaluates the gamble separately from background wealth. Now consider a recursive specification such as
\[
    U(p)=\bE^p_1\:u\big(w+\text{CE}(p_{2|w},v_w)\big),
\]
in which dimension $1$ is background wealth $w$ and $p_{2|w}$ is the conditional distribution of the gamble given background wealth. The function $u$ aggregates background wealth and the certainty equivalent of the gamble conditional on that wealth. Under such a utility function, the small-stakes decision can also be locally insensitive to wealth. This suggests one way to formulate \citeauthor{R2000calibration}'s (\citeyear{R2000calibration}) critique within the SMEU framework. This is close in spirit to \cite{F2015calibration} and \cite{sarver2018dynamic}, who study \citeauthor{R2000calibration}'s (\citeyear{R2000calibration}) critique in recursive non-expected utility models on recursive choice domains.

\section{Conclusion}
\label{sect_end}
This paper develops the SMEU representation for preferences over risky multidimensional alternatives. The representation consists of an RCT and Bernoulli indices. The RCT organizes the dimensions into vertices and specifies the ancestor and branch structure among these vertices; the Bernoulli indices determine the corresponding evaluations. This structure allows the model to accommodate joint evaluation, separate evaluation across branches, and recursive evaluation within a single representation.

The SMEU framework unifies the FATE, FETA, and recursive representations and also includes intermediate cases. We illustrate this unification in two settings. The first is a policy-evaluation example centered on inequality. This example shows how ex post inequality, ex ante inequality, intergenerational mobility, and combinations of these concerns can be represented within the same framework. The second is an application to multisource income, in which generalized bracketing helps analyze dominance and the avoidance of multidimensional risk. Several further directions in this context can be developed. For example, one can add risky background wealth and formulate calibration questions through recursive conditioning.

Time lotteries provide another possible future application. A time lottery has two dimensions: the prize and the date at which the prize is paid. Under the standard exponentially discounted expected-utility model, as \cite{DDGO20} emphasize, it is difficult to accommodate both stochastic impatience and aversion to random payment dates. The SMEU framework provides a potential way to address this issue by allowing the prize and timing dimensions to enter the representation in different ways. One possibility is to bracket prize risk and timing risk separately. Another is to evaluate timing risk conditional on the realized prize. In both cases, stochastic impatience and aversion to random payment dates may coexist.

Future research could more formally examine this application and explore other settings that involve multidimensional risk, such as consumption–investment decisions, mental accounting, and social welfare analysis. Moreover, developing econometric methods to estimate the SMEU representation would be valuable, as it could help reveal the RCTs that decision makers employ in different contexts.

\appendix

\newpage

\section{Proofs}\label{appen2}
{
\setlength{\abovedisplayskip}{5pt}
\setlength{\belowdisplayskip}{5pt}
\setlength{\abovedisplayshortskip}{3pt}
\setlength{\belowdisplayshortskip}{3pt}
\setlength{\jot}{3pt}

\vspace{-8pt}

\subsection{Proof of \texorpdfstring{\thmref{thm_main}}{Theorem 2}}

We begin with some notation. For any $i,j\in I$, denote $i\hrra j$ if $i\hra j$ and $j\not\hra i$.  Fix an RCT $\cT=\{\cP,E\}$. For any 
$i\in I$, let $H(i)=A\in\cP$ if $i\in A$ and let $M(A)=\{i\in A: i\hra A\}$. 
For any $i\in A\subseteq I$, denote $i\hra A$ if $i\hra j$ for all $j\in A$.

Checking the necessity of the axioms is routine (yet nontrivial in our case). We leave the details of the necessity proof to  \hyperref[OA_proof_necessity]{Online Appendix \RN{4}.1}. Below we only describe lemmas that are useful for proving other results of the paper. The first lemma shows the existence of an SMEU representation in which all vertices in the RCT are singletons.

\vspace{-5pt}
\begin{lemma}\label{lemma_nece_singleton}
If $\succsim$ has an SMEU representation, then there exists $(\hat\cP,\hat E)\in \mathbb{T}(\succsim)$ in which $\hat\cP=\{\{1\},\dots,\{N\}\}$.
 \end{lemma}
\vspace{-5pt}

The second lemma characterizes the implications of $i\hra A$ for some $i\in A$. 

\vspace{-5pt}

\begin{lemma}\label{lemma_nece_RI}
    Suppose that $\succsim$ has an SMEU representation. If $i\in A$ and $i\hra A$, then  for all $z\in X_{A^\mathsf{c}}$, there is an SMEU representation of  $\succsim_z$ in which $A\subseteq H(i)\cup \bar d(H(i))$.
\end{lemma}
\vspace{-5pt}

The last two lemmas provide sufficient conditions for $i\perp j$ and $i\hra j$.

\vspace{-5pt}

\begin{lemma}\label{lemma_nece_perp}
    Suppose that $\succsim$ has an SMEU representation.  If  $j\not\in H(i)\cup\bar d(H(i))$ and $i\not\in H(j)\cup\bar d(H(j))$ in some $\cT\in \mathbb{T}(\succsim)$, then $i\perp j$.
\end{lemma}

\vspace{-10pt}

\begin{lemma}\label{lemma_nece_hra}
    Suppose that $\succsim$ has an SMEU representation. If $j\in H(i)\cup\bar d(H(i))$ in some $\cT\in \mathbb{T}(\succsim)$, then $i\hra j$. 
\end{lemma}

\vspace{-5pt}

For the rest of the proof we focus on  the sufficiency of Axioms \ref{axiom_WO}-\ref{axiom_C}. 

\medskip

\textit{Step 1: Preliminary results.} 

\vspace{5pt}

We present several lemmas that will be useful in later steps. Most proofs of these lemmas are deferred to \hyperref[OA_proof_lemmas]{Online Appendix \RN{4}.2}. 
 The first lemma shows that  $\hra$ is transitive.

\vspace{-5pt}

\begin{lemma}\label{lemma_transitive}
Suppose that $\succsim$ satisfies Axioms \ref{axiom_WO}-\ref{axiom_one} and \ref{axiom_C}. If $i\hra j$ and $j\hra l$, then $i\hra l$.
\end{lemma}

\vspace{-5pt}

\begin{proof}[Proof of \lemmaref{lemma_transitive}]
    For all  $x^k,y^k\in X$ and $\pi^k\in (0,1)$ with $k=1,\dots,n$ such that $\sum_{k=1}^n\pi^k=1$, $x^k_i\neq x^{k'}_i, y^k_i\neq y^{k'}_i$ for all $k\neq k'$, $x^k_{\{i,l\}^\mathsf{c}} = y^k_{\{i,l\}^\mathsf{c}} $,  and $ x^k\succsim y^k$ for all $k$, we need to show that $\sum_{k=1}^n\pi^k\delta_{x^k}\succsim \sum_{k=1}^n\pi^k\delta_{y^k}$. Denote $w^k=x^k_{\{i,j,l\}^\mathsf{c}}\in X_{\{i,j,l\}^\mathsf{c}}$ for all $k$.
    
   By relabeling the superscripts, there exists $m\leqslant  n$ such that $(x^k_i,x^k_l)=(y^k_i,y^k_l)$ (and hence $x^k=y^k$) for all $m<k\leqslant n$ and $(x^k_i,x^k_l)\neq (y^k_i,y^k_l)$ for all $1\leqslant k\leqslant m$. For each $k\leqslant m$, Axiom \ref{axiom_M}  implies that either $x^k\succ y^k$, or $x^k\sim y^k$ and
   $(x^k_i, x^k_l), (y^k_i, y^k_l)\not\in \{(\ux_i,\ux_l), (\lx_i,\lx_l)\}$. Without loss of generality, assume $x^k_1<x^k_2<\cdots<x^k_m$. 

    We construct a new sequence $(\hat x^k)_{k=1}^n$ as follows. First, define $\hat x^k=x^k$ for all $k>m$. Then we construct $\hat x^k$  for any $k\leqslant m$ recursively. Suppose we have defined $\hat x^{k'}$ for all $k'>k$ such that $ x^{k'}\succsim \hat x^{k'}\succsim y^{k'}$ for all $k'>k$ and $\sum_{t=1}^n\pi^t \delta_{x^t}\succsim \sum_{t=1}^k\pi^t \delta_{x^t} + \sum_{t=k+1}^n\pi^t \delta_{\hat x^t}$. We want to define $\hat x^k$. 
    If $(x^k_i, x^k_j)\not\in \{(\ux_i,\ux_j), (\lx_i,\lx_j)\}$ or $x^k\succ y^k$, then by Axiom \ref{axiom_M}  and Axiom \ref{axiom_C}, we can find $\hat x^k=(w^k,\hat x^k_i,\hat x^k_j, x^k_l)$ such that (i) $\hat x^k_i\not\in \{x^t_i: t<k\}\cup \{\hat x^t_i: t>k\}\cup \{y^t_i:t=1,\dots,n\}\cup\{\ux_i,\lx_i\}$, (ii) $\hat x^k_j\not\in \{x^t_j: t<k\}\cup \{\hat x^t_j: t>k\}\cup \{y^t_j:t=1,\dots,n\}\cup\{\ux_j,\lx_j\}$, and (iii) $ x^{k}\succsim \hat x^{k}\succsim y^{k}$ . Since $i\hra j$, \[\sum_{t=1}^n\pi^t \delta_{x^t}\succsim\sum_{t=1}^{k}\pi^t \delta_{x^t} + \sum_{t=k+1}^n\pi^t \delta_{\hat x^t}\succsim \sum_{t=1}^{k-1}\pi^t \delta_{x^t} + \sum_{t=k}^n\pi^t \delta_{\hat x^t}.\]
    
Suppose instead that $(x^k_i, x^k_j)\in \{(\ux_i,\ux_j), (\lx_i,\lx_j)\}$ and $x^k\sim y^k$. By symmetry, we focus on the case in which $(x^k_i, x^k_j)=(\ux_i,\ux_j)$. Then $x^k_l<\ux_l$. By Axiom \ref{axiom_M}  and Axiom \ref{axiom_C}, we can find $\tilde x^k=(w^k, x^k_i,\tilde x^k_j, \tilde x^k_l)$ and $\hat x^k=(w^k,\hat x^k_i,\hat x^k_j, \tilde x^k_l)$
such that (i) $\hat x^k_i\not\in \{x^t_i: t<k\}\cup \{\hat x^t_i: t>k\}\cup \{y^t_i:t=1,\dots,n\}\cup\{\ux_i,\lx_i\}$, (ii) $\hat x^k_j, \tilde x^k_j\not\in \{x^t_j: t<k\}\cup \{\hat x^t_j: t>k\}\cup \{y^t_j:t=1,\dots,n\}\cup\{\ux_j,\lx_j\}$, and (iii) $\hat x^k\sim \tilde x^k\sim x^k$. Since $j\hra l$ and $i\hra j$, 
\begin{align*}
    \sum_{t=1}^{k-1}\pi^t \delta_{x^t} + \sum_{t=k}^n\pi^t \delta_{\hat x^t} & \sim \sum_{t=1}^{k-1}\pi^t \delta_{x^t} + \pi^k\delta_{\tilde x^k} + \sum_{t=k+1}^n\pi^t \delta_{\hat x^t} \sim \sum_{t=1}^k\pi^t \delta_{x^t} + \sum_{t=k+1}^n\pi^t \delta_{\hat x^t}\sim \sum_{t=1}^n\pi^t \delta_{x^t}.
\end{align*}

Similarly, we can construct $(\hat y^k)_{k=1}^n$ such that (i) $x^{k}\succsim \hat x^{k}\succsim \hat y^{k}\succsim y^{k}$ for all $k=1,\dots,n$, (ii) $\sum_{t=1}^n\pi^t \delta_{x^t}\succsim  \sum_{t=1}^n\pi^t \delta_{\hat x^t}$, (iii)  $\sum_{t=1}^n\pi^t \delta_{\hat y^t}\succsim  \sum_{t=1}^n\pi^t \delta_{ y^t}$, (iv) $\hat x^k=\hat y^k=x^k=y^k$ for all $k>m$, and (v) all values of $(\hat x^k_t)_{k=1}^m$ and $(\hat y^k_t)_{k=1}^m$ are distinct for both $t=i$ and $t=j$. It remains to show that $\sum_{t=1}^n\pi^t \delta_{\hat x^t}\succsim \sum_{t=1}^n\pi^t \delta_{\hat y^t}$. 

For $k=m$, without loss of generality, assume $\hat x^m_i>\hat y^m_i$ and $\hat x^m_l<\hat y^m_l$. The other cases are similar. By Axiom \ref{axiom_M}, Axiom \ref{axiom_C}, and the boundedness of $X$, there exist finitely many $\hat z^t, \tilde z^t\in X$ for $t=1,..,T$ such that (i) $\hat z^t_{\{i,j,l\}^\mathsf{c}}=\tilde z^t_{\{i,j,l\}^\mathsf{c}}=w^m$ for all $t=1,\dots,T$,
(ii) $ \tilde z^T=\hat y^m$, (iii) $\hat z_i^t=\tilde z_i^t$ and $\hat z_l^t<\tilde z_l^t$ for all $t=1,\dots,T$, (iv) $\hat z^1_l=\hat x^m_l$, $\tilde z^t_l = \hat z^{t+1}_l$, and $\tilde z^t_i > \hat z^{t+1}_i$ for all $t=1,\dots,T-1$, (v) all values of $(\hat z^t_j)_{t=1}^T$, $(\tilde z^t_j)_{t=1}^T$, and $(\hat x^k_j)_{k\neq m}$ are distinct, and (vi) $\hat x^m\succsim \hat z^1\succsim \tilde z^1\succsim\cdots \succsim \hat z^T  \succsim \tilde z^T= \hat y^m$. Since  $i\hra j$, and $j\hra l$, the above construction guarantees that 
\begin{align*}
    \sum_{t=1}^n\pi^t \delta_{\hat x^t} &  \succsim  \sum_{t=1}^{m-1}\pi^t \delta_{\hat x^t} + \pi^m \delta_{\hat z^1} + \sum_{t=m+1}^n\pi^t \delta_{\hat y^t}  \\
&\succsim \sum_{t=1}^{m-1}\pi^t \delta_{\hat x^t} + \pi^m \delta_{\tilde z^1} + \sum_{t=m+1}^n\pi^t \delta_{\hat y^t} \\
&\succsim \sum_{t=1}^{m-1}\pi^t \delta_{\hat x^t} + \pi^m \delta_{\hat z^2} + \sum_{t=m+1}^n\pi^t \delta_{\hat y^t}  \\
&\succsim \cdots\succsim \sum_{t=1}^{m-1}\pi^t \delta_{\hat x^t} + \pi^m \delta_{\tilde z^T} + \sum_{t=m+1}^n\pi^t \delta_{\hat y^t} =  \sum_{t=1}^{m-1}\pi^t \delta_{\hat x^t} + \sum_{t=m}^n\pi^t \delta_{\hat y^t}.
\end{align*}

By repeating the above construction for $k=1,\dots,m-1$, we can show that $\sum_{t=1}^n\pi^t \delta_{\hat x^t}\succsim \sum_{t=1}^n\pi^t \delta_{\hat y^t}$.  This finishes the proof of $i\hra l$.
\end{proof}

Second, the conditional preference on each $i$ has an expected utility (EU) representation.
 
\vspace{-5pt}

\begin{lemma}\label{lemma_conditional}
Suppose that $\succsim$ satisfies Axioms \ref{axiom_WO}-\ref{axiom_one} and  \ref{axiom_C}. For any $i\in I$ and $x\in X_{-i}$, the conditional preference $\succsim_x$ on $\Delta(X_i)$ has an EU representation with a continuous and strictly increasing Bernoulli index $v_{i|x}$, which is unique up to a positive affine transformation.
\end{lemma}

For any $B\subsetneq A$, we denote $B\vartriangleright A$
 if for all $x\in X_{A^{\cp}}$, $r\in \Delta(X_{A\setminus B})$, and $p,q\in \Delta(X_A)$ such that $p_{A\setminus B}=q_{A\setminus B}$, we have $p\succsim_x q \iff (p_B, r)\succsim_x (q_B,r).$  We say that  $A$ is \textit{bracket separable} if  there exists a nontrivial partition $\{B_k\}_{k=1}^n$ of $A$ such that $B_k\vartriangleright A$ for all $k=1,\dots,n$. In this case, we call  $\{B_k\}_{k=1}^n$ a \textit{bracket partition} of $A$. The following lemma provides a sufficient condition for the existence of a bracket partition.

\vspace{-5pt}

\begin{lemma}\label{lemma_suff_SUB}
    Suppose that $\succsim$ satisfies Axioms \ref{axiom_WO}-\ref{axiom_C}. If $A$ is not bracket separable, then there exists $i\in A$ such that $i\hra A$.
\end{lemma}

\vspace{-5pt}

The next result extends Axiom \ref{axiom_M} to lotteries.  For any $p\in \Delta(X_A)$ and $x\in X_A$, we say that $p$ \textit{dominates} $x$ if $p\neq \delta_x$ and $y_i\geqslant  x_i$ for all $i\in A$ and $y_i\in \supp(p_i)$. Similarly, $x$ \textit{dominates} $p$ if $p\neq \delta_x$ and $x_i\geqslant  y_i$ for all $i\in A$ and $y_i\in \supp(p_i)$. The dominance relation is weak if we allow for the possibility that $\delta_x=p$.

\vspace{-5pt}

\begin{lemma}\label{lemma_monotone}
 Suppose that $\succsim$ satisfies Axioms \ref{axiom_WO}-\ref{axiom_C}.   \emph{\text{(i)}} For any $A\subseteq I$, $p\in \Delta(X_A)$,  $x\in X_A$, and $x'\in X_{A^{\cp}}$, if $p$ dominates $x$, then $p\succ_{x'} x$, and if $x$ dominates $p$, then $x\succ_{x'} p$. \emph{\text{(ii)}}  For any $A\subseteq I$,  $p\in \Delta(X_A)$,  $x'\in X_{A^{\cp}}$, and $x,y\in X_A$ such that $p$ dominates $y$ and is dominated by $x$, there exists some $z\in X_A$ such that $p\sim_{x'} z$ and $x\geqslant  z\geqslant  y$.
\end{lemma}

\vspace{-5pt}

Denote $\ux=(\ux_i)_{i\in A}$ and $\lx=(\lx_i)_{i\in A}$. An immediate corollary of \lemmaref{lemma_monotone} is that if $p\neq \ux, \lx$, then $\ux \succ p\succ \lx$ and the set $\{y_i: y\sim p\}$ is uncountable for each $i\in I$.  The following lemma guarantees that a bracket-separable set has a ``finest'' bracket partition.

\vspace{-5pt}

\begin{lemma}\label{lemma_finest}
 Suppose that $\succsim$ satisfies Axioms \ref{axiom_WO}-\ref{axiom_C}. If $A$ is bracket separable, then $A$ must have a bracket partition $\{A_k\}_{k=1}^n$ in which $A_k$ is not bracket separable for all $k=1,\dots,n$. Moreover, such $\{A_k\}_{k=1}^n$ is unique and is finer than any other bracket partition of $A$.
 \end{lemma}

\vspace{-5pt}

We call $\{A_k\}_{k=1}^n$  in \lemmaref{lemma_finest} the \textit{finest bracket partition} of $A$. 

The final lemma states that if $i\hra B$, then we can find a representation of the conditional preference on $B$ with certain linearity properties.

\vspace{-5pt}

\begin{lemma}\label{lemma_linear}
 Suppose that $\succsim$ satisfies Axioms \ref{axiom_WO}-\ref{axiom_C} and $i\in B$. If $i\hra B$, then for every $x\in X_{B^\mathsf{c}}$, there exists a function $U_x: \Delta(X_B)\rightarrow \bR$ such that \emph{\text{(i)}}  $p\succsim_x q$ if and only if $U_x(p)\geqslant  U_x(q)$ for all $p,q\in \Delta(X_B)$; \emph{\text{(ii)}}  $U_x(p\alpha q ) = \alpha U_x(p) + (1-\alpha)U_x(q)$ for all $\alpha\in (0,1)$ and $p,q\in \Delta(X_B)$ with $\supp(p_i)\cap \supp(q_i)=\emptyset$; \emph{\text{(iii)}}  the function $w_x: X_B\rightarrow \bR$ defined by $w_x(y)=U_x(\delta_y)$ for every $y\in X_B$ is continuous and strictly increasing; and \emph{\text{(iv)}}  $U_x$ is unique up to a positive affine transformation.
 \end{lemma}

\medskip

\textit{Step 2: Construct an RCT.} 
\vspace{5pt}

Consider the following procedure to construct an RCT $\cT=(\cP,E)$ in which $\cP=\{\{1\},\dots,\{N\}\}$. We  write $(i,j)$ to represent $(\{i\},\{j\})\in E$.
\medskip

\noindent \textbf{Stage $1$}. Let $\cH_0=\{I\}$ and consider the following cases for $I$: 

\vspace{-5pt}

\begin{itemize}
    \item[(1)] If $I$ is bracket separable, then denote by $\{A_k\}_{k=1}^n$ the finest bracket partition of $I$ as defined in \lemmaref{lemma_finest}. Note that $A_k$ is not bracket separable for every $k$. Let $E_1=\emptyset$, $j^*(I)=o$, and $\cH_1=\{A_k\}_{k=1}^n$. Move on to Stage 2.
    
\vspace{-5pt}

    \item[(2)] If  $I$ is not bracket separable, then \lemmaref{lemma_suff_SUB} implies that $M(I)=\{i\in I: i\hra I\}\neq \emptyset$. Let $M(I)=\{l_1^1,\dots,l^1_{n_1}\}$. Denote    $j^*(I)=l^1_{n_1}$ and 
    $E_1=\{(o,l_1^1), (l_1^1,l_2^1),\dots,(l_{n_1-1}^1,l_{n_1}^1)\}$.
    
\vspace{-5pt}

    \begin{itemize}
        \item[(i)] If $I\setminus M(I)=\emptyset$, then the procedure terminates.
        \item[(ii)] If $I\setminus M(I)\neq \emptyset$ and is not bracket separable, then let $\cH_1=\{I\setminus M(I)\}$ and move on to Stage 2. 
        \item[(iii)] If  $I\setminus M(I)\neq \emptyset$ and is bracket separable, then denote by  $\{B_k\}_{k=1}^m$ the finest bracket partition of $I\setminus M(I)$ as defined in \lemmaref{lemma_finest}. Let $\cH_1=\{B_1,\dots,B_m\}$ and move on to Stage 2.
    \end{itemize}
\end{itemize}

\vspace{-5pt}

\noindent \textbf{Stage $t\geqslant  2$}. Consider any $A\in \cH_{t-1}$. By construction, $A$ is not bracket separable. Again by \lemmaref{lemma_suff_SUB}, $M(A)=\{l^t_1,\dots,l^t_{n_t}\}$ for some $n_t\geqslant  1$. Denote    $j^*(A)=l^t_{n_t}$ and let 
    $\{(j^*(B),l^t_1), (l^t_1,l^t_2),\dots,(l^t_{n_t-1},l^t_{n_t})\}\subseteq E_t$ in which $B\in \cH_{t-2}$ satisfies $A\subseteq B$.

\vspace{-5pt}

    \begin{itemize}
        \item[(i)] If $A\setminus M(A)=\emptyset$, then move on  to other $A'\in \cH_{t-1}$. 
        \item[(ii)] If $A\setminus M(A)\neq \emptyset$ and is not bracket separable, then let $A\setminus M(A)\in \cH_t$.
        \item[(iii)] If  $A\setminus M(A)\neq \emptyset$ and is bracket separable, then  denote by  $\{B_k\}_{k=1}^m$ the finest bracket partition of $A\setminus M(A)$ as defined in \lemmaref{lemma_finest}. Let $\{B_1,\dots,B_m\}\subseteq \cH_t$.
    \end{itemize}
    
\vspace{-5pt}

Repeat the process for all $A\in \cH_{t-1}$. Then, we obtain all elements of $\cH_t$ and $E_t$. Move on to Stage $t+1$.
This procedure terminates in finitely many stages $T$ 
where $A= M(A)$ for all $A\in \cH_T$. Define $E = \bigcup_{t=1}^T E_t$. It is clear that $(\cP,E)$ is an RCT.

\medskip

\textit{Step 3: Construct the Bernoulli indices for each $A\in  \cP_o $.} 
\vspace{5pt}

Define $H_1=\{i\in I:\bar d(\{i\})=\emptyset\}$ and $H_k=\{i\in I\setminus (\cup_{k'<k}H_{k'}): \bar d(\{i\})\subseteq \cup_{k'< k}H_{k'}\}$ recursively for all $k\geqslant 2$. The iteration ends at $H_m$ in which $\cup_{k=1}^m H_k=I$.


 We start with  $i\in H_1$.  
By the construction of $\cT$ and the definition of a bracket partition, $\succsim_x=\succsim_{x'}$ on $\Delta(X_i)$ for any $x,x'\in X_{-i}$ such that $x_{\bar a(\{i\})}=x'_{\bar a(\{i\})}$. In other words, the conditional preference $\succsim_x$ only depends on $x_{\bar a(\{i\})}$. 
By \lemmaref{lemma_conditional}, there exists  $u^i: X_{\bar a(\{i\})}\times X_{i}\to [0,1]$ such that for all $x\in X_{\bar a(\{i\})}$, (i) $u^i(x,\:\cdot\:): X_{i}\to [0,1]$ is continuous and strictly increasing, and satisfies $u^i(x,\ux_{i})=1$, $u^i(x,\lx_{i})=0$; and (ii) 
if one defines $U^i_{x}: \Delta(X) \to [0,1]$  by $U^i_x(p)=\bE^p_{i|x}(u^i(x,y))$, then 
for any $p,q\in \Delta(X)$ such that $p_{-i}=q_{-i}=\delta_z$ and $z_{\bar a(\{i\})}=x$,  we have $p\succsim q\iff U^i_x(p)\geqslant U^i_x(q)$.
 The normalization of $u^i(x,\ux_{i})=1$ and $u^i(x,\lx_{i})=0$ works, since $u^i(x,\:\cdot\:)$ is unique up to a positive affine transformation by \lemmaref{lemma_conditional}.
For any $z\in X_{-i}$ with $z_{\bar a(\{i\})}=x$, the function $U^i_x(\delta_{(z,y)})=u^i(x,y)$ is continuous and strictly increasing in $y\in X_{i}$. 

Next, we define the Bernoulli index for every $i\in H_k$ by induction on $k=2,\dots,m$.  
By the construction of $\cT$ and the definition of a bracket partition, $\succsim_x=\succsim_{x'}$ on $\Delta(X_{\{i\}\cup \bar d(\{i\}) })$ for any $x,x'\in X_{(\{i\}\cup \bar d(\{i\}))^\mathsf{c}}$ such that $x_{\bar a(\{i\})}=x'_{\bar a(\{i\})}$. In other words, the conditional preference $\succsim_x$ only depends on $x_{\bar a(\{i\})}$. 
 Since $i\hra \{i\}\cup \bar d(\{i\})$, for any $x\in X_{\bar a(\{i\})}$ and $z\in X_{(\{i\}\cup \bar d(\{i\}))^{\mathsf{c}}}$ with $z_{\bar a(\{i\})}=x$, \lemmaref {lemma_linear} ensures the existence of $\hat U^i_x:\Delta(X_{\{i\}\cup \bar d(\{i\})})\to [0,1]$ such that (i) $p\succsim_{z} q$ if and only if $\hat U^i_x(p)\geqslant  \hat U^i_x(q)$ for all $p,q\in \Delta(X_{\{i\}\cup \bar d(\{i\})})$; (ii) $\hat U^i_x(p\alpha q ) = \alpha \hat U^i_x(p) + (1-\alpha)\hat U^i_x(q)$ for all $\alpha\in (0,1)$ and $p,q\in \Delta(X_{\{i\}\cup \bar d(\{i\})})$ with $\supp(p_{i})\cap \supp(q_{i})=\emptyset$; and (iii) the function $w_x^i: X_{\{i\}\cup \bar d(\{i\})}\to [0,1]$ defined by $w_x^i(y)=\hat U_x^i(\delta_y)$ is continuous and strictly increasing, and satisfies $w^i_x(\ux_A)=1$ and $w^i_x(\lx_A)=0$. We can extend  $\hat U^i_x$ by defining $U^i_x: \Delta(X)\rightarrow[0,1]$ such that $U^i_x(p)=\hat U^i_x(p_{\{i\}\cup \bar d(\{i\})|x})$ if $x\in \supp(p_{\bar a(\{i\})})$ and $U^i_x(p)=0$ otherwise. Then the two statements in \defref{def_heu} hold for $U^i_x$ for any $x\in X_{\bar a(\{i\})}$ and  $z\in X_{(\{i\}\cup \bar d(\{i\}))^{\mathsf{c}}}$ with $z_{\bar a(\{i\})}=x$. 
 Define $u^i: X_{\bar a(\{i\})}\times X_{i}\times [0,1]^{ c(\{i\})}\to[0,1]$ by 
\begin{equation*}
   u^i(\:x,\:y,\:(U^j_{(x,y)}(p))_{\{j\}\in c(\{i\})}\:) = U^i_x(p)
\end{equation*}
for all $x\in X_{\bar a(\{i\})}$, $y\in X_i$, and $p\in \Delta_{X}$ such that $p_{i}=\delta_y$ and $p_{\bar a(\{i\})}=\delta_x$. 
By the construction of $U^j_{(x,y)}$ for all $\{j\}\in c(\{i\})$ in previous steps and the definition of a bracket partition (when $|c(\{i\})|\geq 2$), $u^i$ is well defined. 
Moreover, $u^i(x,\ux_{i},a)=1$ if $a_j=1$ for every $\{j\}\in  c(\{i\})$ and $u^i(x,\lx_{i},a)=0$ if $a_j=0$ for every $\{j\}\in  c(\{i\})$. It is easy to see that $U^i_x$ and $u^i(x,\cdot)$ are surjective.
Also, for any $p\in \Delta(X)$, we have the recursive equation
 \begin{equation}\label{appen_eq_recursive}
U^i_x(p)=\bE^p_{i|x}\:u^i(\:x,\:y,\:(U^j_{(x,y)}(p))_{\{j\}\in c(\{i\})}\:).
    \end{equation}

Since $\cup_{k=1}^m H_k=I$, it remains to define $u^o:[0,1]^{c(o)}\rightarrow[0,1]$ and $U^o:\Delta(X)\rightarrow[0,1]$. If $|c(o)|=1$, then we can simply set $u^o(a)=a$ for $a\in [0,1]$ and $U^o(p)=U^i(p)$ where $c(o)=\{\{i\}\}$. If $c(o)=\{\{l_1\},\dots,\{l_n\}\}$ with some $n\geqslant 2$, then $\{A_k\}_{k=1}^n$ is the finest bracket partition of $I$, where $A_k=\{l_k\}\cup \bar d(\{l_k\})$.  For any $p\in \Delta(X)$, we have $p\sim (p_{A_1},\dots,p_{A_n})\sim (x_{A_1},\dots,x_{A_n}):=x$, in which $U^{l_k}(\delta_{x}) = U^{l_k}(p)$ for all $k=1,\dots,n$. Since $\succsim$ restricted to degenerate lotteries $X$ is continuous and monotone, Debreu’s Theorem implies that there exists a continuous and strictly increasing function $w^o: X\to [0,1]$ such that $w^o(\ux)=1, w^o(\lx)=0$, and $\delta_x\succsim \delta_y$ if and only if $w^o(x)\geqslant  w^o(y)$ for all $x,y\in X$. Define $u^o: [0,1]^{ c(o)}\to[0,1]$ by 
$   u^o(U^{l_1}(\delta_{x}),\dots, U^{l_n}(\delta_{x}))=w^o(x)$
for all $x\in X$. Because $U^{l_k}(\delta_{\ux})=1, U^{l_k}(\delta_{\lx})=0$, and $U^{l_k}(\delta_{x})$ is continuous and strictly increasing in $x_{A_k}\in X_{A_k}$ for all $k=1,\dots,n$, the function $u^o$ is well defined, continuous, strictly increasing, and satisfies $u^o(1,\dots,1)=1, u^o(0,\dots,0)=0$. Define $U^o: \Delta(X) \to [0,1]$ by     $U^o(p)= u^o(U^{l_1}({p}),\dots, U^{l_n}({p}))$ for every $p\in \Delta(X)$. Then $U^o$ represents $\succsim$. 

To conclude,  $(\cT,(u^A)_{A\in\cP_o})$ is an SMEU representation of $\succsim$. 

\vspace{-8pt}

\subsection{Proof of \texorpdfstring{\thmref{thm_unique_h2}}{Theorem 1}}

Suppose that $\succsim$ has an SMEU representation. By \thmref{thm_main}, $\succsim$ satisfies Axioms \ref{axiom_WO}-\ref{axiom_C}. 
The next lemma validates a procedure for generating new RCTs in $\mathbb{T}(\succsim)$. 

\vspace{-5pt}

\begin{lemma}\label{lemma_modify}
    Suppose $(\cP,E)\in\mathbb{T}(\succsim)$ and $A\in \cP$ such that $c(A)=\{\{i\}\}$ for some $i\in M(A\cup\bar d(A))$. Denote $\cP'=\cP\cup\{A\cup\{i\}\}\backslash\{A,\{i\}\}$. For all $B,B'\in \cP'_o$, let $(B,B')\in E'$ if either (i) $(B,B')\in E$,  (ii) $B'=A\cup\{i\}$ and $(B,A)\in E$, or  (iii) $B=A\cup\{i\}$ and $(\{i\},B')\in E$. Then $(\cP',E')\in \mathbb{T}(\succsim)$.
\end{lemma}

\vspace{-5pt}

Starting with the RCT $\cT$ in the proof of \thmref{thm_main},   we construct $\cT^*$ as follows: 

Step 0: Following the procedure in \lemmaref{lemma_modify}, 
we can combine vertices recursively and generate $\cT_0=(\cP_0,E_0)\in \mathbb{T}(\succsim)$ in which $A=M(A\cup \bar d^{\cT_0}(A))$ and $A\cup \bar d^{\cT_0}(A)$ is not bracket separable for all $A\in \cP_0$.

Step 1: Suppose $c^{\cT_0}(o)=\{A_1,\dots,A_n\}$ and take any $i_k\in A_k$ for $k=1,\dots,n$. Define $J(k)=\{k':i_k\hra i_{k'} \text{~and~} i_{k'}\hra i_{k}\}$. Since $\hra$ is transitive, $\{J(k):k=1,\dots,n\}$ forms a partition of $\{1,\dots,n\}$.
If $|J(k)|\geqslant 2$ for some $k$, then by the proof of \lemmaref{lemma_nece_RI} (and \lemmaref{lemma_nece_EU} in \hyperref[OA_proof_necessity]{Online Appendix \RN{4}.1}), we can show that $c^{\cT_0}(A_{k'})=\emptyset$ for all $k'\in J(k)$ and the conditional preference on $\Delta(\prod_{k'\in J(k)}X_{A_{k'}})$ has an EU representation, with utility index additively separable across different $A_{k'}$. Define the intermediate tree $\mathcal{T}_1$ by specifying its child operator $c^{\mathcal{T}_1}$: (i) For the root $o$, let $c^{\mathcal{T}_1}(o) = \left\{ \bigcup_{k^{\prime} \in J(k)} A_{k^{\prime}} \right\}$ for all equivalence classes $k$; (ii) If $|J(k)| \geqslant 2$, set $c^{\mathcal{T}_1}\left(\bigcup_{k^{\prime} \in J(k)} A_{k^{\prime}}\right) = \emptyset$; (iii) If $|J(k)| = 1$, the subtree is inherited exactly from $\mathcal{T}_0$. It can be shown by the same proof of \lemmaref{lemma_nece_RI} that $\cT_1\in \mathbb{T}(\succsim)$. 

Step $t\geqslant 2$: 
Consider any $A\in \cP_{t-1}$ such that $c^{\cT_{t-1}}(A)\neq\emptyset$ and 
$depth^{\cT_{t-1}}(i)=t-1$ for $i\in A$. By construction, for all $B\subseteq \bar d^{\cT_{t-1}}(A)$, we have $B\in \cP_{t-1}$ if and only if $B\in \cP_0$. 
Then either $\bar d^{\cT_{t-1}}(A)$ is not bracket separable, or $\{B\cup\bar d^{\cT_{t-1}}(B)\}_{B\in c^{\cT_{t-1}}(A)}$ is the finest bracket partition of $\bar d^{\cT_{t-1}}(A)$. Following the procedure in step 1 by replacing vertex $o$ with every such $A$, we can construct a new $\cT_t\in \mathbb{T}(\succsim)$. 

The above procedure terminates in finitely many steps. Denote the resulting RCT by $\cT^*=(\cP^*,E^*)$. The following properties hold: (i) $depth^{\cT^*}(i)=depth^{\cT_0}(i)$ for all $i\in I$; (ii) $\cP_0$ is finer than $\cP^*$; (iii) if $H^{\cT_0}(i)\subsetneq H^{\cT^*}(i)$, then $M(H^{\cT^*}(i))=H^{\cT^*}(i)$ and $c^{\cT_0}(H^{\cT_0}(i)) = c^{\cT^*}(H^{\cT^*}(i))=\emptyset$; and (iv) if $c^{\cT^*}(A)\neq \emptyset$ for some $A\in \cP^*$, then $A\in \cP_0$ and $A\cup \bar d^{\cT^*}(A) = A\cup \bar d^{\cT_0}(A)$.

We claim that $depth^{\cT^*}(i)\leqslant depth^{\cT}(i)$ for all $i\in I$ and $\cT\in \mathbb{T}(\succsim)$. By property (i) above, it suffices to show that  
$depth^{\cT_0}(i)\leqslant depth^{\cT}(i)$ for all $i\in I$ and $\cT\in \mathbb{T}(\succsim)$. Consider any $i\in I$. If $depth^{\cT}(i)=1$, then $H^{\cT}(i)\in c^{\cT}(o)$ and $i\hra H^{\cT}(i)\cup \bar d^{\cT}(H^{\cT}(i))$. 
 By construction, either $I$ is not bracket separable or $\{A\cup\bar d^{\cT_0}(A)\}_{A\in c^{\cT_0}(o)}$ is the finest bracket partition of $I$. In the former case, $H^{\cT}(i)\cup \bar d^{\cT}(H^{\cT}(i))=I$ and hence $i\in M(I)\in \cP_0$, which implies that $depth^{\cT_0}(i)=1$.  In the latter case, let $i\in A\cup \bar d^{\cT_0}(A)$ for some $A\in c^{\cT_0}(o)$.  Since $\{A\cup\bar d^{\cT_0}(A)\}_{A\in c^{\cT_0}(o)}$ is the finest bracket partition of $I$, we have $A\cup \bar d^{\cT_0}(A)\subseteq H^{\cT}(i)\cup \bar d^{\cT}(H^{\cT}(i))$. Then $i\hra A\cup \bar d^{\cT_0}(A)$ and hence $i\in M(A\cup \bar d^{\cT_0}(A))=A$.  This also implies that $depth^{\cT_0}(i)=1$.

 Suppose by induction that 
$depth^{\cT}(i)\geqslant $
$depth^{\cT_0}(i)$ if $depth^{\cT}(i)\leqslant t$ for some $t\geqslant 1$. Now consider $i$ such that $depth^{\cT}(i) = t+1$. Suppose by contradiction that $depth^{\cT_0}(i) > t+1$. Then we  can find $A_1,\dots,A_{t+1}\in \cP$ and $A'_1,\dots,A'_{t+1}\in\cP_0$ such that $A_{k+1}\in c^{\cT}(A_k)$ and $A'_{k+1}\in c^{\cT_0}(A_k')$ for $k=0,\dots,t$, where $A_0=A_0'=o$, $i\in A_{t+1}$, and $i\in \bar d^{\cT_0}(A'_{t+1})$.  We claim that $A'_{k}\cup\bar d^{\cT_0}(A'_{k})\subseteq A_{k}\cup\bar d^{\cT}(A_{k})$ for all $k=1,\dots,t+1$.  For $k=1$, either $I$ is not bracket separable or $\{A\cup\bar d^{\cT_0}(A)\}_{A\in c^{\cT_0}(o)}$ is the finest bracket partition of $I$. 
In the former case, $A'_{1}\cup\bar d^{\cT_0}(A'_{1}) = A_{1}\cup\bar d^{\cT}(A_{1})=I$. In the latter case, the claim follows from the definition of the finest bracket partition.

By induction, suppose that the claim holds for $k\leqslant  m\leqslant t$ and consider $k=m+1$. Since $A_m\subseteq M(A_{m}\cup\bar d^{\cT}(A_{m}))$ and $A'_{m}\cup\bar d^{\cT_0}(A'_{m})\subseteq A_{m}\cup\bar d^{\cT}(A_{m})$, we know $A_m\cap (A'_{m}\cup\bar d^{\cT_0}(A'_{m}))\subseteq M(A'_{m}\cup\bar d^{\cT_0}(A'_{m}))=A'_{m}$ and hence $d^{\cT_0}(A'_{m})\subseteq d^{\cT}(A_{m})$. If $A'_{m+1}\cup\bar d^{\cT_0}(A'_{m+1})\not\subseteq A_{m+1}\cup\bar d^{\cT}(A_{m+1})$, then there exists $\hat A\in c^{\cT}(A_m)\setminus\{A_{m+1}\}$ such that $(\hat A\cup\bar d^{\cT}(\hat A))\cap (A'_{m+1}\cup\bar d^{\cT_0}(A'_{m+1}))\neq \emptyset$.  This implies that both $\bar d^{\cT_0}(A'_{m})$ and $A'_{m+1}\cup\bar d^{\cT_0}(A'_{m+1})$ are bracket separable. 
By \lemmaref{lemma_finest}, this contradicts the construction of $\cT_0$ in which $\{A\cup\bar d^{\cT_0}(A)\}_{A\in c^{\cT_0}(A_m')}$ is the finest bracket partition of $d^{\cT_0}(A'_{m})$, whenever $d^{\cT_0}(A'_{m})$ is bracket separable.  By induction, $A'_{k}\cup\bar d^{\cT_0}(A'_{k})\subseteq A_{k}\cup\bar d^{\cT}(A_{k})$ for all $k=1,\dots,t+1$. Since $i\in A_{t+1}\subseteq M(A_{t+1}\cup\bar d^{\cT}(A_{t+1}))$ and $i\in A'_{t+1}\cup\bar d^{\cT_0}(A'_{t+1})$, we have $i\in M(A'_{t+1}\cup\bar d^{\cT_0}(A'_{t+1}))= A'_{t+1}$, and hence $depth^{\cT_0}(i) = t+1$, contradicting our assumption. Hence, we conclude that $depth^{\cT}(i)\geqslant depth^{\cT_0}(i)=depth^{\cT^*}(i)$ for all $i\in I$.

Consider any $\cT=(\cP,E)\in \mathbb{T}(\succsim)$ such that $depth^{\cT}(i)= depth^{\cT^*}(i)$ for all $i\in I$. We claim that $\cP$ is finer than $\cP^*$. Suppose that this is not true. Then there exist $A\in \cP$, $m\geqslant 1$, and distinct $A^*_1,A^*_2\in \cP^*$ such that  $A\cap A_k^*\neq\emptyset$ for $k=1,2$ and  $depth^{\cT}(l)=depth^{\cT^*}(l)=m$ for all $l\in A\cup A_1^*\cup A_2^*$. 
Let $A_1^*\in c^{\cT^*}(B_1^*)$ for some $B_1^*\in \cP^*_o$. For any $i_1\in A^*_1$ and $i_2\in A^*_2$, we have $i_1\hra i_2$ and $i_2\hra i_1$. If $A^*_2\in c^{\cT^*}(B_1^*)$, then by the construction of $\cT^*$, vertices $A^*_1$ and $A^*_2$ should have been combined in the construction of $\cT^*$, leading to a contradiction. If instead $A^*_2\not\in c^{\cT^*}(B_1^*)$, then there exists $B_2^*\in \cP^*\setminus\{B_1^*\}$ such that $A^*_2\in c^{\cT^*}(B_2^*)$. 
Let $A\in c^{\cT}(B)$ for some $B\in \cP$. We cannot have $B_1^*\cap B\neq \emptyset$ and $B_2^*\cap B\neq \emptyset$ simultaneously,  since otherwise we can find $j_1\in B_1^*$ and $j_2\in B_2^*$ such that $j_1\hra j_2$ and $j_2\hra j_1$, and hence $B_1^*$ and $B_2^*$ should have been combined in the construction of $\cT^*$. 
Assume without loss that $B_1^*\cap B=\emptyset$. 
This implies $B_1^*\cup \bar d^{\cT^*}(B^*_1)$ is bracket separable. However, since $c^{\cT^*}(B_1^*)\neq \emptyset$, it must be the case that $B_1^*\in \cP_0$ and $B_1^*\cup \bar d^{\cT^*}(B^*_1)=B_1^*\cup \bar d^{\cT_0}(B^*_1)$, which is not bracket separable by the definition of $\cT_0$. This leads to a contradiction.  We conclude that $\cP$ is finer than $\cP^*$.

 

We now show $\cT^*$ is  the unique simplest RCT in $\mathbb{T}(\succsim)$. If $\cP\neq \cP^*$, then we can use the same argument above.  If $\cP= \cP^*$ and $\cT\neq \cT^*$,  then there must exist $A^*\in \cP^*$ and distinct $A_1,\dots,A_n\in\cP^*$ such that $n\geqslant 2$, $c^{\cT^*}(A^*)\neq \emptyset$, $depth^{\cT^*}(l)=depth^{\cT^*}(l')$ for all $l,l'\in A^*\cup A_1\cup\cdots\cup A_n$, and $A^*\cup \bar d^{\cT^*}(A^*)\subseteq \bigcup_{k=1}^n (A_k\cup \bar d^{\cT}(A))$. This implies $A^*\cup \bar d^{\cT^*}(A^*)$ is bracket separable. This leads to a contradiction since $c^{\cT^*}(A^*)\neq \emptyset$ implies   $A^*\in \cP_0$ and hence $A^*\cup \bar d^{\cT^*}(A^*)=A^*\cup \bar d^{\cT_0}(A^*)$ is not bracket separable by the definition of $\cT_0$. 


\vspace{-8pt}

\subsection{Proofs in Section \ref{sect_special}} \label{OA_proof_sect_special}

\begin{proof}[Proof of \thmref{thm_special}]  
First, if $\succsim$ has a FATE representation, then  $(\{I\},\{(o,I)\})\in \mathbb{T}(\succsim)$. \lemmaref{lemma_nece_hra} implies that $i\hra j$ for all $i,j\in I$. Conversely, if
$i\hra j$ for all $i,j\in I$, then $M(I)=I$ and hence $ (\{I\},\{(o,I)\})\in \mathbb{T}(\succsim)$, which corresponds to a FATE representation.

Second, if $\succsim$ has a FETA representation, then $\mathbb{T}(\succsim)=\{\cT^*\}$ in which $c^{\cT^*}(o)=\{\{1\},\dots,\{N\}\}$. If $i\hra j$ for some $i\neq j$, then \lemmaref{lemma_nece_RI} implies the existence of $\cT'\in \mathbb{T}(\succsim)$ in which $\{i,j\}\subseteq  H^{\cT'}(i)\cup \bar d^{\cT'}(H^{\cT'}(i))$. It is clear that $\cT'\neq \cT^*$, which contradicts $\mathbb{T}(\succsim)=\{\cT^*\}$.  Conversely, if $i\hra j\iff i=j$, then 
\lemmaref{lemma_nece_hra} implies that $\cP=\{\{1\},\dots,\{N\}\}$ and $c^{\cT}(\{i\})=\emptyset$ for all $i=1,\dots,N$ and $\cT=(\cP,E)\in \mathbb{T}(\succsim)$. The unique RCT that satisfies these conditions is $\cT^*$ in which $c^{\cT^*}(o)=\{\{1\},\dots,\{N\}\}$. Hence, $\mathbb{T}(\succsim)=\{\cT^*\}$ and $\succsim$ has a FETA representation.

Finally, if $\succsim$ has a recursive representation, then $\mathbb{T}(\succsim)=\{\cT^*\}$. Denote $\cT^*=(\cP^*,E^*)$. We must have $\cP^*=\{\{1\},\dots,\{N\}\}$, since otherwise
 $|\mathbb{T}(\succsim)|>1$ by \lemmaref{lemma_nece_singleton} . Because $|c^{\cT^*}(A)|\leqslant 1$    for all $A\in\cP_o^*$,   there must exist an enumeration $\pi:I\rightarrow I$ such that  $c^{\cT^*}(o)=\{\{\pi(1)\}\}$ and $c^{\cT^*}(\{\pi(k)\})=\{\{\pi(k+1)\}\}$ for all $k=1,\dots,N-1$. By \lemmaref{lemma_nece_hra}, $\pi(i)\hra \pi(j)$ for all $j\geqslant i$ and thus $\hra$ is complete. By \lemmaref{lemma_transitive}, $\hra$ is transitive. Suppose by contradiction that $\hra$ is not antisymmetric. That is, there exists $i<j$ such that $\pi(i)\hra \pi(j)$ and $\pi(j)\hra \pi(i)$. By  \lemmaref{lemma_nece_RI},
we can find $\cT'\in \mathbb{T}(\succsim)$ in which $\pi(i)\in  H^{\cT'}(\pi(j))\cup \bar d^{\cT'}(H^{\cT'}(\pi(j)))$. It is clear that $\cT'\neq \cT^*$, a contradiction. 

Conversely, suppose that $\hra$ is a linear order. Then there exists an enumeration $\pi:I\rightarrow I$ such that $\pi(i)\hrra \pi(j)$ for all $j>i$. For any $\cT=(\cP,E)\in \mathbb{T}(\succsim)$, we must have  $\cP=\{\{1\},\dots,\{N\}\}$ and $\pi(i)\not\in \bar d^{\cT}(\{\pi(j)\})$ for all $j>i$. Suppose by contradiction that $|c^{\cT}(A)|>1$ for some $A\in \cP_o$. Without loss of generality, assume that $A=o$. In other cases it suffices to focus on the conditional preference on $\Delta(X_{\bar d^{\cT}(\pi(t))})$ where $t$ is the smallest number such that $|c^{\cT}(\{\pi(t)\})|>1$. It is clear that $\{\pi(1)\}\in c^{\cT}(o)$.
For any $\{\pi(k)\}\in c^{\cT}(o)$ such that $k>1$, since $\pi(1)\hra \pi(k)$, $\pi(k)\not\in \bar d^{\cT}(\{\pi(1)\})$, and $\hra$ is antisymmetric,  \lemmaref{lemma_nece_EU} in \hyperref[OA_proof_necessity]{Online Appendix \RN{4}.1} implies that $c^{\cT}(\{\pi(k)\}=\emptyset$. If $|c^{\cT}(o)|\geqslant 3$, then we can find $\{\pi(i)\},\{\pi(j)\}\in c^{\cT}(o)$ such that $1<i<j$. By $\pi(i)\hra \pi(j)$ and the proof of \lemmaref{lemma_nece_RI}, the utility function is additively EU for dimensions $\pi(i)$ and $\pi(j)$ and hence $\pi(j)\hra \pi(i)$, which is a contradiction with the antisymmetry of $\hra$.

Now suppose that $|c^{\cT}(o)|=2$, i.e., $c^{\cT}(o)=\{\{\pi(1)\},\{\pi(k)\}\}$ for some $k\geqslant 2$. If $k=N$ and $N=2$, then $c^{\mathcal{T}}(o) = \{\{\pi(1)\}, \{\pi(2)\}\}$. If $k=N$ and $N \geqslant 3$, then there exists $\mathcal{T}^{\prime} \in \mathbb{T}(\succsim)$ such that $c^{\mathcal{T}^{\prime}}(o)=\{\{\pi(1)\}\}$, $c^{\mathcal{T}^{\prime}}(\pi(i))=\{\{\pi(i+1)\}\}$ for all $i<N-2$, and $c^{\mathcal{T}^{\prime}}(\pi(N-2))=\{\{\pi(N-1)\},\{\pi(N)\}\}$. In both cases we can focus on the conditional preference on $\{\pi(N-1),\pi(N)\}$, which can be represented by (\ref{eq_nece_RI2}) in \hyperref[OA_proof_necessity]{Online Appendix \RN{4}.1}. The proof of \lemmaref{lemma_nece_RI} implies that the utility function is additively EU for dimensions $\pi(N)$ and $\pi(N-1)$, which contradicts the antisymmetry of $\hra$. If instead $k<N$, then   there exists $\cT'\in \mathbb{T}(\succsim)$ such that $c^{\cT'}(o)=\{\{\pi(1)\}\}$, $c^{\cT'}(\pi(i))=\{\{\pi(i+1)\}\}$ for all $i<k-1$, $c^{\cT'}(\pi(k-1))=\{\{\pi(k)\}, \{\pi(k+1)\}\}$, and $c^{\cT'}(\pi(i))=\{\{\pi(i+1)\}\}$ for all $i\geqslant k+1$. Again, we can focus on the conditional preference on $\bar d^{\cT'}(\{\pi(k-1)\})$. Using the same argument as above, 
it must be the case that $c^{\cT'}(\pi(k+1))=\emptyset$, implying that $k=N-1$. This reduces to the previous case and a contradiction follows.  Hence, for any $\cT=(\cP,E)\in \mathbb{T}(\succsim)$, we must have  $\cP=\{\{1\},\dots,\{N\}\}$, $c(o)=\{\{\pi(1)\}\}$, $c(\{\pi(i)\})=\{\{\pi(i+1)\}\}$ for all $i=1,\dots,N-1$. Since $\cT^*$ is the unique RCT that satisfies these conditions, we conclude that  $\mathbb{T}(\succsim)=\{\cT^*\}$ and  $\succsim$ has a recursive representation. 
\end{proof}


\begin{proof}[Proof of \thmref{thm_special2}]  

If $\succsim$ has a generalized bracketing representation with  $c(o)=\{A_k\}_{k=1}^n$, then $\succsim$ can be represented by $U(p)=v(\bE^p_{A_1}\:u^{A_1}(x_{A_1}),\dots,\bE^p_{A_n}\:u^{A_n}(x_{A_n})).$
This implies $i\hra j$ for all $i,j\in A_k$ and $k=1,\dots,n$. If $i\hra j$ for some $i\in A_k$ and $j\in A_{k'}$ with $k\neq k'$, then by the proof of \lemmaref{lemma_nece_RI}, $\succsim$ has a generalized bracketing representation with partition $\{A_{k}\cup A_{k'}\}\cup \{A_l\}_{l=1,\dots,n, l\neq k,k'}$. This implies $j\hra i$ and hence $\hra$ is symmetric.

Conversely, assume that $\hra$ is symmetric. By \lemmaref{lemma_transitive}, $\hra$ is transitive and hence an equivalence relation. If $M(I)\neq \emptyset$, then $i\hra j$ for all $i,j\in I$.  By \thmref{thm_special}, $\succsim$ has an FATE representation, which is a special case of the generalized bracketing representation. If $M(I)=\emptyset$, then again by symmetry of $\hra$, the simplest RCT $\cT^*$ must feature $c^{\cT^*}(o)=\{A_k\}_{k=1}^n$, where $\{A_k\}_{k=1}^n$ is a partition of $I$ and $M(A_k)=A_k$ for all $k=1,\dots,n$. Hence, $\succsim$ has a generalized bracketing representation.

If   $\succsim$ has a generalized recursive representation with RCT $\cT=(\cP,E)$ where $|c^{\cT}(A)|\leqslant 1$    for all $A\in\cP_o$,  then $\cP=\{A_1,\dots,A_n\}$ in which $c^{\cT}(o)=\{A_1\}$ and $c^{\cT}(A_{k})=\{A_{k+1}\}$ for all $k=1,\dots,n-1$. By \lemmaref{lemma_nece_hra}, for any $i\in A_k$ and $j\in A_{k'}$ for some $k,k'=1,\dots,n$, we have $i\hra j$ if $k\leqslant k'$ and $j\hra i$ if $k'\leqslant k$. This proves the completeness of $\hra$. 

Conversely, assume that $\hra$ is complete. Since $\hra$ is transitive,  $M(A)\neq \emptyset$ for all nonempty  $A\subseteq I$. Define $l_1\in M(I)$ and $l_k\in M(I\setminus\{l_1,\dots,l_{k-1}\})$ for all $k= 2,\dots,N$. Let $\cP=\{\{1\},\dots,\{N\}\}$ and $E=\{(o,\{l_1\})\}\cup \{(\{l_k\},\{l_{k+1}\})\}_{k=1}^{N-1}$. We can follow the construction of Bernoulli indices (Step 3) in the proof of the sufficiency of \thmref{thm_main} to generate an SMEU representation of $\succsim$ with RCT $\cT$. Since $|c^{\cT}(A)|=1$ if $A\in \cP_o\setminus\{\{l_N\}\}$ and $c^{\cT}(\{l_N\})=\emptyset$, the RCT $\cT$ corresponds to a  generalized recursive representation of $\succsim$.

If $\succsim$ has a recursive bracketing representation, then $\cP=\{\{1\},\dots,\{N\}\}$ for all  $\cT=(\cP,E)\in \mathbb{T}(\succsim)$. Suppose that $i\hra j$  and $j\hra i$ for some $i\neq j$. By transitivity, $i\hra l$ if and only if $j\hra l$ for all $l\in I$. Let $A=\{l\in I: i\hra l\}$. Then $i,j\in M(A)$. 
By the proof of \lemmaref{lemma_nece_RI} and \lemmaref{lemma_modify},  there exists another $\cT'=(\cP',E')\in \mathbb{T}(\succsim)$
such that $\{i,j\}\in \cP'$, which leads to a contradiction. Hence, $\hra$ is antisymmetric.

Conversely, suppose that $\hra$ is antisymmetric.
Then \lemmaref{lemma_nece_hra} implies that 
$|A|=1$ for all  $(\cP,E)\in \mathbb{T}(\succsim)$ and $A\in\cP$. That is,  $\succsim$ has a recursive bracketing representation.
\end{proof}

\vspace{-8pt}

\subsection{Proofs in Section \ref{section_app}} \label{OA_proof_sect_app}


\begin{proof}[Proof of \propref{prop_dominance2}]
For sufficiency, note that if $\succsim$ is represented by $U(p)=\mathbb{E}^{f[p]}\:v$, then $\succsim$ satisfies dominance. Suppose instead that  $v$ exhibits CARA---that is, $\text{CE}(p,v) = \frac{1}{a}\log \bE^q\: e^{ax}$ for $a\neq 0$ or $\text{CE}(p,v)=\bE^q x$. In the latter case, $\succsim$ is also represented by $U(p)=\mathbb{E}^{f[p]}v$ and hence dominance holds. In the former case, for any $p=(p_2,\dots,p_N)$, 
\begin{align*}
    U(p) &= \frac{1}{a}\: \sum_{i=1}^n\:\log \bE^{f[p_{A_i}]}\: e^{ax} = \frac{1}{a} \: \log \prod_{i=1}^n\: \bE^{f[p_{A_i}]}\: e^{ax} = \frac{1}{a} \: \log \: \bE^{f[p]}\: e^{ax}.
\end{align*}
The last equality holds because $p_{A_i}$ and $p_{A_j}$ are statistically independent for all $i\neq j$. Hence, $\succsim$ satisfies independent-sources dominance.

For necessity, suppose that $\succsim$ satisfies independent-sources dominance and is not represented by $U(p)=\mathbb{E}^{f[p]}v$. Then $\{A_k\}_{k=1}^n$ must be a nontrivial partition of $I$. By applying  independent-sources dominance to source-$l$ lotteries and source-$l'$ lotteries for any $l\neq l'$, it must be the case that $v^{A_k}$ is a positive affine transformation of $v^{A_1}:=v$ for all $k$. Without loss of generality, we can assume $v^{A_k}=v$ for all $k$.  Fixing any $i\in A_1$, $j\in A_2$, and $z\in X_{\{i,j\}^{\cp}}$ such that $z_l=0$ for all $l\in \{i,j\}^{\cp}$, we consider the conditional preference $\succsim_{z}$ on $\Delta(X_{\{i,j\}})$.  Its utility function can be written as $\hat{U}(p)= \text{CE}(p_i,v) + \text{CE}(p_j,v)$ for all $p\in \Delta(X_{\{i,j\}})$. By Proposition 1 of \citeAppendix{OA-RabinWeizsacker09}, if $v$ does not exhibit CARA, then $\succsim_{z}$ (and hence $\succsim$) must violate independent-sources dominance.
\end{proof}


\begin{proof}[Proof of \propref{prop_dominance}] 
The sufficiency part is trivial. For necessity, suppose that $\succsim$ satisfies dominance and $\succsim$ is not represented by $U(p)=\mathbb{E}^{f[p]}v$. Then \propref{prop_dominance2} implies that, after normalization, $v^{A_k}=v$ for all $k$ and  
$v$ exhibits CARA. We can construct $\succsim_{z}$ on $\Delta(X_{\{i,j\}})$ as in the proof of \propref{prop_dominance2}, which is represented by $\hat{U}(p)= \text{CE}(p_i,v) + \text{CE}(p_j,v)$. Normalize $v(0)=0$. By dominance, for any $x\in Z$, we have $\delta_{(x,0)}\frac{1}{2}\delta_{(0,x)}\sim \delta_{(x,0)}$, which implies $v(x/2)=v(x)/2$. Hence, $\text{CE}(q,v)=\bE^q x$ and $\succsim$ is represented by $U(p)=\mathbb{E}^{f[p]}v$, which leads to a contradiction. 
\end{proof}


\begin{proof}[Proof of \propref{prop_dominance3}] 
To prove the first statement, assume that $\succsim^1$ is represented by (\ref{eq_money3}) with the coarsest brackets of income sources $\{A_i\}_{i=1}^n$. If $B\subseteq A_i$ for some $i$, then for any $p\in \Delta(X)$ such that $p_l=\delta_0$ for all $l\not\in B$, its utility is $\text{CE}(f[p],\:v^{A_i})$ and hence $\succsim^1$ satisfies dominance on $B$. Conversely, suppose that $\succsim^1$ satisfies dominance on $B$ and both $B\cap A_{k}$ and $B\cap A_{k'}$ are nonempty for some $k\neq k'$. Following the proof of \propref{prop_dominance} and \propref{prop_dominance2}, we must have $\text{CE}(q,v^{A_k})=\text{CE}(q,v^{A_{k'}})=\bE^q x$. This implies that $\succsim^1$ is also represented by (\ref{eq_money3}) with coarser brackets $\{A_i\}_{i\neq k,k'}\cup\{A_{k}\cup A_{k'}\}$, a contradiction.

To prove the second statement, suppose $\succsim^2$ is also represented by (\ref{eq_money3}) with the coarsest brackets of income sources $\{A'_i\}_{i=1}^m$. If $\{A'_i\}_{i=1}^m$ is finer than $\{A_i\}_{i=1}^n$ and $\succsim^2$ satisfies dominance on $B$, then by the first statement, $B\subseteq A'_{k'}\subseteq A_{k}$ for some $k$ and $k'$,  $\succsim^1$ also satisfies dominance on $B$, and hence $\succsim^1$ is more dominance-consistent than $\succsim^2$. Conversely, suppose  $\succsim^1$ is more dominance-consistent than $\succsim^2$. Since  $\succsim^2$ satisfies dominance on $A'_i$ for all $i=1,\dots,m$, $\succsim^1$ also satisfies dominance on $A'_i$ for all $i=1,\dots,m$. 
The first statement implies that $\{A'_i\}_{i=1}^m$ must be finer than $\{A_i\}_{i=1}^n$. This completes the proof.\end{proof}


\begin{proof}[Proof of \propref{prop_avoidance}] 
For the first statement, the sufficiency part is trivial. Indeed, by \propref{prop_dominance2}, if $\succsim$ is represented by $U(p)=\mathbb{E}^{f[p]}\:v$ for some $v$ or $U(p)=\sum_{i=1}^n \text{CE}(f[p_{A_i}],\:v)$ for some $v$ that exhibits CARA, then $\succsim$ satisfies independent-sources dominance and the decision maker is indifferent regardless of whether her income comes from one source or multiple sources. For necessity, suppose that $\succsim$ is not represented by $U(p)=\mathbb{E}^{f[p]}\:v$. Since $p\sim q$ if $p$ is a source-$i$ lottery and $q$ is a source-$j$ lottery such that $f[p]=f[q]$  for all $i,j\in I$, we can normalize that $v^{A_k}=v$ for all $k=1,\dots,n$. 
We can construct $\succsim_{z}$ on $\Delta(X_{\{i,j\}})$ as in the proof of \propref{prop_dominance2}, which is represented by $\hat{U}(p)= \text{CE}(p_i,v) + \text{CE}(p_j,v)$. Set $p_j=\delta_x$ for some $x\in Z$. Since both gains and losses are allowed, avoidance of multidimensional risk implies that $(p_i,\delta_x)\sim_{z} (f[(p_i,\delta_x)],\delta_0)$---that is, $\text{CE}(p_i,v)+ x = \text{CE}(f[(p_i,\delta_x)],v)$. Hence, $v$ must exhibit CARA. 

For the second statement, suppose $\succsim$ is not represented by $U(p)=\mathbb{E}^{f[p]}\:v$. We will focus on the representation of $\succsim$ on $\Delta(Z_+^{N})$.  Since $p\sim q$ if $p$ is a source-$i$ lottery on $Z_+$ and $q$ is a source-$j$ lottery on $Z_+$ such that $f[p]=f[q]$  for all $i,j\in I$, we can normalize that $v^{A_k}=v$ on $Z_+$ for all $k=1,\dots,n$. It suffices to show that  $\succsim$ satisfies avoidance of multidimensional risk in gains if and only if $v$  exhibits DARA for $x\in Z_+$. To show the ``if'' part, suppose $v$ exhibits DARA for $x\in Z_+$.  Then $\text{CE}(q,v)+ \text{CE}(q',v) \leqslant \text{CE}(f[(q,q') ],v)$ for all $q,q'\in \Delta(Z_+)$. Hence, for any $p\in \Delta(Z_+^N)$ such that $p=(p_1,\dots,p_N)$, we have 
\begin{align*}
        U(p)=\sum_{i=1}^n \text{CE}(f[p_{A_i}],v) \leqslant \text{CE}\big(f\big[(f[p_{A_1}],\dots,f[p_{A_n}])\big], v \big) = \text{CE}(f[p],v),
\end{align*}
which implies $f[p]\succsim p$. To show the ``only if'' part, we can follow the same proof idea of the first statement by observing that avoidance of multidimensional risk in gains implies that $\text{CE}(q,v)+ x \leqslant \text{CE}(f[(q,\delta_x)],v)$ for $x\in Z_+$ and $q\in \Delta(Z_+)$. Hence, $v$  exhibits DARA for $x\in Z_+$.

The proof of the third statement is symmetric to that of the second one and is omitted.
\end{proof}

\begin{proof}[Proof of \propref{prop_OA_compare}]
    For necessity, suppose $\succsim_1$ exhibits stronger avoidance of multidimensional risk in gains than $\succsim_2$. We will focus on the representation of $\succsim^1$ on $\Delta(Z_+^{N})$. Since  for all $i\in I$,  $q\succsim^2 p \Longrightarrow q\succsim^1 p$ if $p,q$  are source-$i$ lotteries on $Z_+$,  we can normalize that $v^{A_k}=\hat v$ on $Z_+$ for all $k=1,\dots,n$. 
    Suppose there exists $A_i$ such that $A_i\cap B_j\neq\emptyset$ and $A_i\cap B_{j'}\neq\emptyset$ for some $j\neq j'$. Fix $l\in A_i\cap B_j$ and $l' \in A_i\cap B_{j'}$. Since $\hat v$ exhibits strictly DARA for $x\in Z_+$,  we can find $r\in \Delta(Z_+)$, $a\in Z_+$, and $\varepsilon>0$ such that $f[r,\delta_a]\in  \Delta(Z_+)$ and $\text{CE}(r,\hat v) + a <\text{CE}(f[r,\delta_{a-\varepsilon}],\hat v)$. Let $p,q\in \Delta(Z_+^{N})$ such that $p_l=f[r,\delta_{a-\varepsilon}], p_{l'}=\delta_0$, $q_l=r, q_{l'}=\delta_a$, and $p_i=q_i=\delta_0$ for all $i\neq l,l'$. Note that $p$ is a source-$l$ lottery.      For $\succsim^2$, the utility of $p$ is $\text{CE}(f[r,\delta_{a-\varepsilon}],\hat v)$, which is larger than $\text{CE}(r,\hat v) + a$, the utility of $q$. However, for $\succsim^1$, since $l,l'$ are in the same bracket $A_i$, the utility of $p$, $\text{CE}(f[r,\delta_{a-\varepsilon}],\hat v)$, is smaller than the utility of $q$, $\text{CE}(f[r,\delta_{a}],\hat v)$. That is, $p\succ_2 q $ and $ q \succ_1 p$ where $p$ is a source-$l$ lottery, a contradiction. 

    For sufficiency, suppose $\succsim^1$ and $\succsim^2$ are represented by (\ref{eq_money3}) with parameters $\{(A_\ell,\hat v)\}_{\ell=1}^n$ and $\{(B_\ell,\hat v)\}_{\ell=1}^m$, respectively.  Since $\{A_\ell\}_{\ell=1}^{n}$ is finer than $\{B_\ell\}_{\ell=1}^{m}$, for each $j=1,...,m$, the collection $\{A_\ell: A_\ell\cap B_j\neq \emptyset\}$ is a partition of $B_j$. Following the same proof of \propref{prop_avoidance}, for all $p=(p_1,\dots,p_N)\in \Delta(Z_+^{N})$,
    \[U^1(p) = \sum_{i=1}^n \text{CE}(f[p_{A_i}],\hat v) \leqslant  \sum_{j=1}^m \text{CE}(f[p_{B_j}],\hat v) =U^2(p).\] 
    Moreover, $U^1(q)=U^2(q)$ for every source-$l$ lottery $q\in \Delta(Z_+^{N})$.
    Hence, $q\succsim_2 p \Longrightarrow q\succsim_1 p$ and $\succsim_1$ exhibits stronger avoidance of multidimensional risk in gains than $\succsim_2$.
\end{proof}}

\begingroup
\setstretch{1}
\bibliographystyle{ecta}
\bibliography{MCU}
\endgroup

\newpage
\appendix
\pagestyle{fancy}
\setlength{\headheight}{14.49998pt}
\fancyhead{}
\fancyhead[C]{Online Appendix}
\begin{center}
    { \Large \bf Online Appendix}\\ \medskip
    {\large (For Online Publication Only)}
\end{center}

This is the online appendix to ``\textit{Decision Making Under Multidimensional Risk}.'' It is organized as follows. Section \hyperref[OA_examples]{\RN{1}} presents additional examples. Section \hyperref[OA_money]{\RN{2}} provides additional results regarding applications in Section \ref{section_app}. Section \hyperref[OA_results]{\RN{3}} characterizes the uniqueness of Bernoulli indices and the (generalized) recursive preferences with an exogenous order on dimensions. Section \hyperref[OA_proofs]{\RN{4}} contains omitted proofs.

\subsection*{Online Appendix \RN{1}: Additional Examples}\label{OA_examples}

\subsubsection*{\RN{1}.1: More on the FATE and FETA Approaches}\label{OA_opposite}
The two commonly used but opposite approaches to evaluating a risky multidimensional alternative, FATE and FETA, have appeared in many different contexts beyond those discussed in the main paper. We describe a few additional examples below.

\begin{enumerate}
    \item Suppose the decision maker is evaluating a risky consumption bundle that yields $(0,1)$ and $(1,0)$ with equal probability. She wants to use a constant-elasticity-of-substitution function $u$ to aggregate the quantities of different goods. Should she use the FATE approach, $\frac{1}{2}u(0,1)+\frac{1}{2}u(1,0)$, or the FETA approach, $u(1/2,1/2)$? The first approach may seem more natural, but the second may capture narrow bracketing and insensitivity to correlation across goods, patterns often discussed in choice behavior. See Section \ref{sect_lit} for further discussion.
    \item Consider a risky consumption stream. If we simply compute exponentially discounted expected utility, the decision maker will exhibit risk-seeking behavior in the time dimension. More general evaluation approaches have therefore been proposed to avoid such risk-seeking behavior.  Some approaches first evaluate risk within each period and then aggregate across periods, while others first aggregate across periods and then take expectations. They are incompatible with each other, and it is not clear which approach is more appropriate. See Section \ref{sect_lit} for further discussion.
    \item Choice models under subjective uncertainty face the same dilemma. Let $(x_1,x_2)$ denote the decision maker's utility in states $1$ and $2$, respectively, for a given act. Now consider a lottery over acts. There are then at least two natural ways to evaluate such a lottery:
    \begin{enumerate}
        \item FETA: First take expected utility state by state, and then apply an ambiguity aggregator across states, such as the maxmin operator in \citeAppendix{AO-GS89}. It is well known in the ambiguity literature that randomization can hedge ambiguity under this approach. To see why, consider the lottery $p=\tfrac{1}{2}\delta_{(1,0)}+\tfrac{1}{2}\delta_{(0,1)}$. When facing $\delta_{(1,0)}$ or $\delta_{(0,1)}$ alone, ambiguity aversion makes the decision maker pessimistic about receiving the high payoff $1$. Under the FETA approach, however, expected utility is taken first, and $p$  is evaluated as  the degenerate lottery   $(1/2,1/2)$, eliminating ambiguity in the evaluation.
        \item FATE: First aggregate across states using an ambiguity operator (such as the maxmin aggregator of \citeAppendix{AO-GS89}), and only then take expected utility. Under this approach, randomization cannot hedge ambiguity. To see why, consider again lottery $p$. In this case, the decision maker evaluates each realization through the aggregator before taking expectations, so the pessimistic assessment of state probabilities remains even after randomization. Indeed, if, for instance, $\delta_{(1,0)}\sim\delta_{(0,1)}$, then $p\sim\delta_{(1,0)}\sim\delta_{(0,1)}$, implying that hedging by randomization offers no benefit.
    \end{enumerate}

    \citeAppendix{OA-Saito15} characterizes an objective function that takes a convex combination of the two approaches, and \citeAppendix{OA-KeZhang20} further generalize that objective function. However, these formulations are incompatible with ours because they violate unidimensional independence (see Example \ref{eg_mix}). Our generalized bracketing representation provides a different way to combine FATE and FETA. Its RCT partitions the dimensions, applying the FATE approach within each partition element and the FETA approach across partition elements. Hence, the generalized bracketing representation naturally serves as an intermediate case between FATE and FETA.

    Using the interpretation in \citeAppendix{OA-KeZhang20}, under FATE (resp.\ FETA), it is as if the decision maker believes that the probability measure over states is determined before (resp.\ after) her randomization over acts, although it is unknown to her. Under the generalized bracketing representation, it is as if she believes that the probability measure over brackets of states is determined after the randomization, while within each bracket the (conditional) probability measure is determined before the randomization, although none of these probability measures is known to her.

\end{enumerate}

\subsubsection*{\RN{1}.2: Other Ways to Generalize the Three Approaches}\label{OA_eg_general}
We present two examples that also generalize the three approaches---FATE, FETA, and recursive evaluation---but in ways that differ from the SMEU representation. We identify which of our axioms fail in these examples. The first example features a convex combination of the FATE and FETA utility functions. 
\begin{example}\label{eg_mix} 
    Let $N=2$. Suppose the decision maker's utility function is 
    \[ U(p)=\alpha\: \bE^p\sqrt{x_1+x_2} \: +\: (1-\alpha) \: \sqrt{(\bE^p_1\sqrt{x_1})^2 + (\bE^p_2\sqrt{x_2})^2},    \]
    in which $\alpha\in (0,1)$. The idea is simple, but the preference represented by this utility function violates unidimensional independence and RCT perfection. To see why unidimensional independence is violated, fix $x_2=1$ and note that the conditional preference $\succsim_{x_2}$ is represented by
    \[U(p_1,1)=\alpha\: \bE^{p_1}\sqrt{x_1+1} \: +\: (1-\alpha) \: \sqrt{(\bE^{p_1}\sqrt{x_1})^2 + 1},\]
    which is not a monotone transformation of any expected utility function in dimension $1$. Moreover, one can show that $1\not\hra 2$,  $2\not\hra 1$, and $1\not\perp 2$. Hence, RCT perfection is also violated.
\end{example}

Recall that under the FETA approach, evaluation depends only on the marginal distributions of the different brackets. The next example combines expected utility with this marginal-based feature.
\begin{example}\label{eg_CN} 
    Let $N=2$. Suppose the decision maker's utility function is 
    $$U(p)=\sum_{x_1,x_2} \sqrt{x_1+x_2}\: p_1(x_1)\:p_2(x_2).$$
When $p=(p_1,p_2)$, this evaluation agrees with an expected utility function whose Bernoulli index is $u(x_1,x_2)=\sqrt{x_1+x_2}$. For more general lotteries $p$, the decision maker is indifferent between $p$ and $(p_1,p_2)$, so only the marginal distributions enter the evaluation. One can verify that her preference satisfies all axioms in Section \ref{sect_axiom} except uncorrelated separability.  To see this, since $p\sim (p_1,p_2)$ for all $p\in\Delta(X)$, we have $1\perp 2$. However, for each $i=1,2$, the conditional preference $\succsim_{x_i}$ over the other dimension depends on $x_i$, and there exist $p_1,q_1\in\Delta(X_1)$ and $x_2,y_2\in X_2$ such that  $(p_1,\delta_{x_2})\succsim (q_1,\delta_{x_2})$ and $(p_1,\delta_{y_2})\prec (q_1,\delta_{y_2})$. Hence, the preference violates uncorrelated separability.
\end{example}

\subsection*{Online Appendix \RN{2}: More on Applications}\label{OA_money}
In this section, we discuss several remarks related to Section \ref{section_app}.

First, if the decision maker strictly avoids multidimensional risk 
for some $p\in \Delta(Z^N)$ such that $p=(p_1,\dots,p_N)$, then  by \propref{prop_avoidance}, her preference cannot  be represented by $U(p)=\mathbb{E}^{f[p]}\:v$ or $U(p)=\sum_{i=1}^n \text{CE}(f[p_{A_i}],\:v)$ with $v$ exhibiting CARA. Hence, \propref{prop_dominance} and \propref{prop_dominance2} imply that the preference violates both dominance and independent-sources dominance. This observation reveals a connection between these two implications of generalized bracketing.

Second,  we axiomatize the utility representation (\ref{eq_money3}):
\begin{equation*}
    U(p)= \sum_{i=1}^n\text{CE}(f[p_{A_i}],\:v^{A_i}).
\end{equation*}
A decision maker with the above utility function satisfies \textit{dominance over deterministic prospects}: For any $x,y\in Z^N$, we have $\delta_x\succsim \delta_y$ if and only if $\sum_{i=1}^N x_i\geqslant \sum_{i=1}^N y_i$.
As discussed at the end of Section \ref{section_app}, this is the key property that distinguishes (\ref{eq_money3}) from the narrow-bracketing utility functions in \cite{V2023bracketing} and \cite{Camara21}.

\begin{proposition}\label{prop_OA_NB}
    Suppose the preference $\succsim$ has a generalized bracketing representation. It is represented by (\ref{eq_money3}) if and only if   it satisfies dominance over deterministic prospects.
\end{proposition}

\begin{proof}[Proof of \propref{prop_OA_NB}]
  The necessity of axioms is clear because $U(\delta_x)=u(\sum x_i)$.  

  For sufficiency, suppose $\succsim$ has a generalized bracketing representation:  
    \begin{equation*} 
        U(p) = \hat u(\bE^p_{A_1}\:\hat v^{A_1}(x_{A_1}),\dots,\bE^p_{A_n}\:\hat v^{A_n}(x_{A_n})).
    \end{equation*}

 By dominance over deterministic prospects, for every $i=1,\dots,n$  and $x_{A_i},y_{A_i}\in X_{A_i}$ with $\sum_{l\in A_i} x_l = \sum_{l\in A_i} y_l$, we have $\hat v^{A_i}(x_{A_i}) = \hat v^{A_i}(y_{A_i})$. Then there exists a continuous and strictly increasing function $v^{A_i}$ such that $\hat v^{A_i}(x_{A_i})=v^{A_i}(\sum_{l\in A_i} x_l)$ and hence $\bE^p_{A_i}\:\hat v^{A_i}(x_{A_i}) = \bE^{f[p_{A_i}]}v^{A_i}(x)$. Again by dominance over deterministic prospects, we can find a continuous and strictly increasing function $u$ such that $\hat u(v^{A_1}(x_{A_1}),\:\dots,\:v^{A_n}(x_{A_n})) = u(\sum_{i\geqslant 1}x_i)$. Hence, the representation   can be rewritten as (up to a monotone transformation) 
\begin{equation*}
    \hat U(p)=  \sum_{i=1}^n\text{CE}(f[p_{A_i}],\:v^{A_i}).
\end{equation*}
 This completes the proof.   
\end{proof}

\propref{prop_OA_NB} can be readily extended to provide an axiomatization for  (\ref{eq_money}).

\subsection*{Online Appendix \RN{3}: Additional Results}\label{OA_results}

\subsubsection*{\RN{3}.1: Uniqueness Properties of Bernoulli Indices}

Suppose that $\succsim$ has an SMEU representation $(\cT,(u^A)_{A\in\cP_o})$. We now explore the uniqueness properties of Bernoulli indices $(u^A)_{A\in\cP_o}$ by fixing the RCT $\cT$. For each $A\in\cP$, we say that the function $u^A$ is \textit{normalized} if $u^A:X_{\bar a(A)}\times X_A\times[0,1]^{c(A)}\to[0,1]$ and, for every $x\in X_{\bar a(A)}$, the section $u^A(x,\cdot)$ is surjective. Similarly, we say that the function $u^o$
is \textit{normalized} if $u^o$ is a surjective mapping  from $[0,1]^{c(o)}$ to $[0,1]$. 
An SMEU representation $(\cT,(u^A)_{A\in\cP_o})$ is \textit{normalized} if $u^A$ is normalized for all $A\in \cP_o$. The following result shows that a normalized SMEU representation  exists and $u^A$ is unique for all $A\in\cP$. Note that the root aggregator $u^o$ need not be unique under this normalization.

\begin{proposition}\label{prop_unique_u}
If $\succsim$ has an SMEU representation with RCT $\cT$, then it has a normalized SMEU representation with the same RCT $\cT$. If $(\cT,(u^A)_{A\in\cP_o})$ and $(\cT,(\hat u^A)_{A\in\cP_o})$ are both normalized SMEU representations of $\succsim$, then $u^A=\hat{u}^A$ for all $A\in\cP$.
\end{proposition}

\begin{proof}[Proof of \propref{prop_unique_u}] 
We first establish the existence of a normalized representation. Suppose that $(\mathcal{T}, (u^A)_{A \in \mathcal{P}_o})$ is an SMEU representation for $\succsim$. We construct a normalized representation $(\mathcal{T}, (\hat{u}^A)_{A \in \mathcal{P}_o})$ inductively, moving upward through the tree. For each $A \in \mathcal{P}$ with $c(A) = \emptyset$ and $x \in X_{\bar{a}(A)}$, the section $u^A(x, \cdot)$ is an expected-utility index. We apply a positive affine transformation $\alpha_A(x) u^A(x, \cdot) + \beta_A(x)$ with $\alpha_A(x) > 0$ to derive $\hat{u}^A(x, \cdot)$ such that its image is $[0,1]$. For an internal vertex $A \in \mathcal{P}$, assume inductively that for all children $B \in c(A)$, the normalized indices $\hat{u}^B$ and their associated affine coefficients $\alpha_B, \beta_B$ have been defined. To preserve the represented preference, the parent aggregator must absorb the inverse of these transformations. Define an intermediate parent index by evaluating the original $u^A$ at the inversely transformed continuation coordinates: $u^A_{mid}(x, y_A, \hat{v}) = u^A(x, y_A, (\frac{\hat{v}_B - \beta_B(x, y_A)}{\alpha_B(x, y_A)})_{B \in c(A)})$, where $\hat{v}_B \in [0,1]$. We then apply a positive affine transformation $\alpha_A(x) u^A_{mid}(x, \cdot, \cdot) + \beta_A(x)$ to map the image of this intermediate function to $[0,1]$, yielding the normalized index $\hat{u}^A$. Finally, for the root $o$, we apply a strictly monotone transformation to $u^o$ that absorbs the inverse affine transformations of its children $B \in c(o)$ to obtain $\hat{u}^o$. This generates a valid, normalized SMEU representation.

Next, we establish uniqueness. Suppose that $(\mathcal{T}, (u^A)_{A \in \mathcal{P}_o})$ and $(\mathcal{T}, (\hat{u}^A)_{A \in \mathcal{P}_o})$ are two normalized SMEU representations for $\succsim$. We proceed by induction, moving upward through the tree. For each $A \in \mathcal{P}$ with $c(A) = \emptyset$ and $x \in X_{\bar{a}(A)}$, both $u^A(x, \cdot)$ and $\hat{u}^A(x, \cdot)$ represent the same conditional preference over $\Delta(X_A)$. By the uniqueness of expected utility, $\hat{u}^A(x, \cdot) = \alpha u^A(x, \cdot) + \beta$ for some $\alpha > 0$. Because both indices have an exact image of $[0,1]$, we must have $\alpha = 1$ and $\beta = 0$, implying $u^A = \hat{u}^A$. For an internal vertex $A \in \mathcal{P}$, assume inductively that $u^B = \hat{u}^B$ for all $B \in c(A)$. Because the child scales are fixed and identical, both $u^A(x, \cdot, \cdot)$ and $\hat{u}^A(x, \cdot, \cdot)$ take identical continuation vectors $v \in [0,1]^{c(A)}$ as inputs. By \lemmaref{lemma_linear}, $U_x$ is unique up to a positive affine transformation. Therefore, the expected-utility indices for $A$ must be related by a single, global positive affine transformation across the entire joint domain $X_A \times [0,1]^{c(A)}$. Again, because both $u^A(x, \cdot, \cdot)$ and $\hat{u}^A(x, \cdot, \cdot)$ are surjective onto $[0,1]$, this affine transformation must be the identity. Thus $u^A = \hat{u}^A$ for all $A \in \mathcal{P}$, which completes the proof.\end{proof}


\subsubsection*{\RN{3}.2: Exogenous Order on Dimensions}\label{time_order}

In the recursive and generalized recursive representations in Section \ref{sect_special}, the order of dimensions need not be exogenously fixed. For instance, when $N=2$, both RCTs with $c(o)=\{\{1\}\}$ and $c(o)=\{\{2\}\}$ are associated with some recursive representation. However, in some applications, a natural sequentiality is already built into the primitive. For example, in intertemporal settings,   period $t$ precedes period $t+1$ for every $t$. In this case, it may be reasonable to sharpen predictions of the theory by imposing the following additional axioms.\footnote{We can easily accommodate other exogenous orders by relabeling the dimensions.} 

\begin{axiom}\label{axiom_EO}
    (Exogenous Order) If $i\leqslant j$, then $i\hra j$.
\end{axiom}
\begin{axiom}\label{axiom_EO2}
    (Exogenous Linear Order) For all $i,j\in I$, $i\hra j\iff i\leqslant j$.
\end{axiom}

The following corollary characterizes the implications of Axioms \ref{axiom_EO} and \ref{axiom_EO2}. Its proof is immediate and is omitted.

\begin{corollary}\label{prop_exogenous}
    Suppose the preference $\succsim$ has an SMEU representation. It satisfies exogenous order if and only if it has a generalized recursive representation in which $j\in H(i)\cup \bar d(H(i))$ for all $j\geqslant i$. It satisfies  exogenous linear order if and only if it has a recursive representation in which $c(o)=\{\{1\}\}$ and $c(\{i\})=\{\{i+1\}\}$ for all $i=1,\dots,N-1$.   
\end{corollary}

\subsubsection*{\RN{3}.3: The Number of RCTs}\label{number}
In this section, we derive the formula for $T(N)$---the number of RCTs with $N$ dimensions. We first consider the number of ways to partition $I$ into exactly $k\in \{1,\dots,N\}$ nonempty subsets, which is given by the Stirling number of the second kind  $S(N,k)$:  
\[S(N,k)=\frac{1}{k!}\sum_{i=0}^k(-1)^{k-i}\binom{k}{i}i^N.\]

When $|\cP|=k$, so that there are $k$ vertices in addition to the root $o$, the number of rooted trees on these $k+1$ labeled vertices with root fixed at $o$ is $(k+1)^{k-1}$ by Cayley's formula. To see this, note that when $k=1$, there is only one ($(1+1)^{1-1}=1$) rooted tree;  for $k=2$, with nonroot vertices $A$ and $B$, there are three rooted trees, since $c(o)$ can be ${A}$, ${B}$, or ${A,B}$. Combining the two formulas yields the following formula for the number of RCTs with $N$ dimensions: 
\[T(N) = \sum_{k=1}^N\: (k+1)^{k-1}\:S(N,k). \]
For example, $T(1)=1$, $T(2)=4$, $T(3)=26$, and $T(4)=243,\ldots$

\subsubsection*{\RN{3}.4: Falsification of RCT Perfection}\label{falsify}
Although the two binary relations $\hra$ and $\perp$ cannot be exactly elicited from finite datasets, one of our key behavioral axioms, RCT perfection, can be falsified with finite data. To see this, note that to falsify RCT perfection it suffices to find $i,j\in I$ such that $i\not\hra j$, $j\not\hra i$, and $i\not\perp j$. For instance, if we can find $z\in X_{\{i,j\}^\mathsf{c}}$ and $x^1,x^2,y^1,y^2\in X_{\{i,j\}}$ such that $x^1_l\neq x^2_l$, $y^1_l\neq y^2_l$ for $l=i,j$, $\delta_{x^1}\sim_z \delta_{y^1}, \delta_{x^2}\sim_z \delta_{y^2}$, and $\frac12\delta_{x^1}+\frac12\delta_{x^2}\not\sim_z \frac12\delta_{y^1}+\frac12\delta_{y^2}$, then we can conclude that $i\not\hra j$ and $j\not\hra i$. If, in addition, we can find $p\in \Delta(X_{\{i,j\}})$ such that $p\not\sim_z(p_i,p_j)$, we can conclude that $i\not\perp j$.

\subsection*{Online Appendix \RN{4}: Omitted Proofs}\label{OA_proofs}

For completeness, we restate the lemmas whose proofs are omitted from the main text. To streamline the exposition, the proofs of lemmas  used for the necessity part of \thmref{thm_main} (i.e., Lemmas \ref{lemma_nece_singleton}-\ref{lemma_nece_hra}) are given in
\hyperref[OA_proof_necessity]{Online Appendix \RN{4}.1}; the remaining proofs are collected in \hyperref[OA_proof_lemmas]{Online Appendix \RN{4}.2}.  

\subsubsection*{\RN{4}.1: Proof of the Necessity of Axioms in \thmref{thm_main}} \label{OA_proof_necessity}

Suppose that $\succsim$ has an SMEU representation 
 $(\cT,(u^A)_{A\in\cP_o})$ with $\cT=(\cP,E)$. The first lemma shows that the conditional preference of $\succsim$ on any nonempty subset of dimensions also has an SMEU representation.

\begin{lemma}\label{lemma_nece_restrict}
If $\succsim$ has an SMEU representation, then for any nonempty $I'\subseteq I$ and $z\in X_{{I'}^\mathsf{c}}$, the conditional preference $\succsim_z$ has an SMEU representation.
 \end{lemma}

\begin{proof}[Proof of \lemmaref{lemma_nece_restrict}] 
Fix any $I'\subseteq I$ and $z\in X_{{I'}^\mathsf{c}}$. We omit the dependence of functions on $z$ when there is no risk of confusion. Define $\cP'=\{A\subseteq I':A\neq \emptyset, A=I'\cap B \text{~for some~} B\in \cP\}$. Clearly, $\cP'$ is a partition of $I'$. For each $A'\in\cP'$, define $d'(A')=\{B'\in \cP': \text{~there exist~}A, B\in\cP \text{~with~} A'=A\cap I', B'=B\cap I', \text{~and~} B\in d(A)\}.$ Similarly, define $d'(o)=\{B'\in \cP': \text{~there exists~}  B\in\cP \text{~with~}   B'=B\cap I', \text{~and~} B\in d(o)\}$.
Define edges by
\[E'=\{(A,B): A,B\in \cP'_o , B\in d'(A),\text{~and there is no~}C\in \cP' \text{~s.t.~} C\in d'(A), B\in d'(C) \}.\]
It is straightforward to check that $\cT'=(\cP',E')$ is an RCT on $I'$. 
We can similarly define mappings $a'$, $\bar a'$, $c'$, and $ \bar d'$. 
Fix any $A\in \cP'$ and $z'\in X_{\bar a'(A)}$. Denote by $B_A$ the element in $\cP$ that includes $A$. Then $\bar a'(A)\subseteq \bar a(B_A)$. Define $\hat{U}^A_{z'}(p)$ for each $p\in \Delta(X_{I'})$ by 
\begin{equation*}
    \hat{U}^A_{z'}(p)= {U}^{B_A}_{(z', z_{\bar a(B_A)\backslash {I'}})}(p, \delta_z).
\end{equation*}
For the root $o$, we define the overall evaluation simply as $\hat{U}^o(p) = U^o(p, \delta_z)$. Then we recursively define the non-root indices $\hat{u}^{A}: X_{\bar{a}^{\prime}(A)} \times X_{A} \times \mathbb{R}^{c^{\prime}(A)} \rightarrow \mathbb{R}$ by letting   \begin{equation*}
   \hat{U}^A_{z'}(p)=\bE^p_{A|z'}\:\hat{u}^A(\:z',\:y,\:(\hat{U}^B_{(z',y)}(p))_{B\in c'(A)}\:).
    \end{equation*}
and define the root aggregator $\hat{u}^o: \mathbb{R}^{c'(o)} \rightarrow \mathbb{R}$ by composing the original root aggregator $u^o$ with the deterministic utility evaluations of any original vertices that were contracted because their intersection with $I'$ was empty.

    We can verify that $(\cT',(\hat{u}^A)_{A\in\cP'\cup\{o\}})$ is an SMEU representation of $\succsim_z$.
\end{proof}

The next lemma shows that if $\succsim$ has an SMEU representation, then it must have one in which all nonroot vertices are singleton sets.

\begin{restatedlemma}{lemma_nece_singleton}
If $\succsim$ has an SMEU representation, then there exists $(\hat\cP,\hat E)\in \mathbb{T}(\succsim)$ in which $\hat\cP=\{\{1\},\dots,\{N\}\}$.

\end{restatedlemma}

\begin{proof}[Proof of \lemmaref{lemma_nece_singleton}] 
Let $(\cT,(u^A)_{A\in\cP_o})$ be an SMEU representation of $\succsim$  with $\cT=(\cP,E)$.
Consider any $A\in\cP$ with $|A|\geqslant  2$. Denote $A=\{i_1,\dots,i_m\}$. Let $\hat\cP:=\{\{i_1\},\dots,\{i_m\}\}\cup \cP \setminus\{A\}$ and define $\hat E$ as follows: (i) $(B,B')\in \hat E$ if $(B,B')\in E$ with $B,B'\neq A$; (ii) $(B,\{i_1\})\in \hat E$ if $(B,A)\in E$; (iii) $(\{i_m\},B)\in \hat E$ if $(A,B)\in E$; and (iv) $(\{i_k,i_{k+1}\})\in \hat E$ for all $k=1,\dots,m-1$. It is clear that $(\hat \cP,\hat E)$ is an RCT. Define $\hat c$ and  $\bar{\hat a}$ accordingly. For any $B\in \cP_o$ with $B\neq A$, set $\hat u^B=u^B$. Since $\hat c(\{i_m\})=c(A)$, 
define $\hat u^{i_m}:X_{\bar{\hat a}(\{i_m\})}\times X_{i_m}\times\bR^{\hat c(\{i_m\})}\to\bR$ by $\hat u^{i_m}(x,y,a)= u^A(x_{\bar a(A)}, (x_{i_1},\dots,x_{i_{m-1}},y),a)$. For each $k=1,\dots,m-1$, note that $\hat c(\{i_k\})=\{i_{k+1}\}$. Define $\hat u^{i_k}:X_{\bar{\hat a}(\{i_k\})}\times X_{i_k}\times\bR\to\bR$ by $\hat u^{i_k}(x,y,a)=a$. We can verify that $U^A_x(p)=\hat U^{i_1}_x(p)$ for all $p\in \Delta(X)$ and $x\in X_{\bar a(A)}$ and that $(\hat\cP, \hat E,(\hat u^A)_{A\in\hat\cP_o})$ is also an SMEU representation of $\succsim$. Repeating the above procedure for all $A'\in \cP$ with $|A'|\geqslant  2$ proves the result.
\end{proof}

 Define $U_x^A$ as in \defref{def_heu} for all $A\in\cP_o$ and $x\in  X_{\bar a(A)}$. 
 Since $U^o(\delta_z)$ is continuous and strictly increasing in $z\in X$ and $p\succsim q$ if and only if $U^o(p)\geqslant  U^o(q)$ for all $p,q\in\Delta(X)$, Axiom \ref{axiom_WO}  and Axiom \ref{axiom_M}  hold.  The following lemma shows that $\succsim$ satisfies Axiom \ref{axiom_C}.

 \begin{lemma}\label{lemma_continuity}
Suppose that $\succsim$ has an SMEU representation with RCT $\cT$. The following  hold: \emph{\text{(i)}} For any $A\in\cP_o$ with $c(A)=\emptyset$, the function $u^A(x,y)$ is continuous and strictly increasing in $y\in X_A$ for all fixed $x\in X_{\bar a(A)}$.\: \emph{\text{(ii)}} For any $A\in\cP_o$ with $c(A)\neq \emptyset$, the function $u^A(x,y,a)$ is continuous and strictly increasing in $a\in \bigtimes_{B\in c(A)}U^{B}_{(x,y)}(X)$ for all  $x\in X_{\bar a(A)}$ and $y\in X_A$.
\emph{\text{(iii)}} For any $A\in\cP_o$, $x\in X_{(A\cup \bar d(A))^\mathsf{c}}$ and $p\in \Delta(X_{A\cup \bar d(A)})$, there exists $z\in X_{A\cup \bar d(A)}$ such that $p\sim_x \delta_z$.\: \emph{\text{(iv)}} The preference $\succsim$ satisfies Axiom \ref{axiom_C}.
 \end{lemma}

\begin{proof}[Proof of \lemmaref{lemma_continuity}] 
For any $A\in\cP_o$ with $c(A)=\emptyset$ and any $x \in X_{\bar a(A)}$, $z\in X_{A^\mathsf{c}}$ with $z_{\bar a(A)}=x$, we know $u^A(x,y)=U^A_x(\delta_{(z,y)})$, which is continuous and strictly increasing in $y\in X_A$ by \defref{def_heu}. This proves (i). 

For (ii), fix any $A\in \cP_o$ with $c(A)\neq \emptyset$, $x\in X_{\bar a(A)}$, and  $y\in X_A$. For any $a,a'\in \bigtimes_{B\in c(A)}U^{B}_{(x,y)}(X)$ with $a\geqslant  a'$ and $a\neq a'$, the strict monotonicity and continuity of $U^{B}_{(x,y)}(\delta_w)$ in dimensions $B\cup\bar d(B)$ guarantee the existence of $z,z'\in X$ such that $z_i=z'_i$ for all $i\not\in \bar d(A)$,  $z\neq z', z\geqslant  z', z_{A\cup \bar a(A)}= z'_{A\cup \bar a(A)}=(x,y)$,  $a=(U^{B}_{(x,y)}(\delta_{z}))_{B\in c(A)}$ and $a'=(U^{B}_{(x,y)}(\delta_{z'}))_{B\in c(A)}$. By the strict monotonicity of $U^A_x(\delta_w)$ in dimensions $A\cup\bar d(A)$, we have 
\[u^A(x,y,a)=U^A_x(\delta_z)>U^A_x(\delta_{z'})=u^A(x,y,a').\]
This implies that $u^A(x,y,a)$ is strictly increasing in $a$. 
For continuity, consider any sequence $(a^n)_{n\geqslant  1}$ with $a^n\in \bigtimes_{B\in  c(A)}U^{B}_{(x,y)}(X)$ for all $n\geqslant 1$ and $a^n\to a\in \bigtimes_{B\in  c(A)}U^{B}_{(x,y)}(X)$. By strict monotonicity and continuity of $U^{B}_{(x,y)}(\delta_w)$ in dimensions $B\cup\bar d(B)$, there exist $z^n,z\in X$  for $n\geqslant  1$ such that $z_i=z^n_i$ for all $i\not\in \bar d(A)$, $z_{A\cup \bar a(A)}= z^n_{A\cup \bar a(A)}=(x,y)$,  $z^n\to z$, $a=(U^{B}_{(x,y)}(\delta_{z}))_{B\in c(A)}$, and $a^n=(U^{B}_{(x,y)}(\delta_{z^n}))_{B\in c(A)}$. Then by continuity of $U^A_x(\delta_w)$ in dimensions $A\cup\bar d(A)$, we have 
\[u^A(x,y,a^n)=U^A_x(\delta_{z^n})\to U^A_x(\delta_{z})=u^A(x,y,a).\]
Hence,  $u^A(x,y,a)$ is continuous in $a$. This proves (ii).

Define $\cH_0=\{o\}$ and $\cH_k=\{A\in \cP: A\in c(B) \hbox{~for some~} B\in \cH_{k-1}\}$ for  $k\geqslant  1$. The iteration ends at $\cH_m$ where $c(B)=\emptyset$ for all $B\in \cH_m$ for some $m\geqslant  1$. Indeed, $\{\cH^k\}_{k=1}^m$ forms a partition of $\cP$.  We show (iii) by induction on $k$. 

 For any $A\in \cH^m, p\in \Delta(X)$, and $x\in \supp(p_{\bar a(A)})$, we have  $U^A_x(p)= \bE^p_{A|x}\:u^A(x,y)$. By (i), $u^A(x,y)$ is strictly increasing and continuous in $y$. Hence, there exists $z\in X$ such that $x= z_{\bar a(A)}$ and  $U^A_x(p)=U^A_x(\delta_z)$. Now assume that (iii) holds for all $A\in \bigcup_{i=k}^m \cH^i$ for some $k\leqslant m$. For any $A\in \cH^{k-1}, p\in \Delta(X)$, and $x\in \supp(p_{\bar a(A)})$,  we have 
 \[U^A_x(p)=\bE^p_{A|x}\:u^A(\:x,\:y,\:(U^B_{(x,y)}(p))_{B\in c(A)}\:).\]
Since $ c(A)\subseteq \bigcup_{i=k}^m  \cH^i$ and the utility function indexed by $B$ only depends on dimensions in $B\cup \bar a(B)\cup \bar d(B)$, 
the inductive hypothesis implies the existence of $z^y\in X$ for all $y\in \supp (p_{X_{A}})$ such that $z^y_{\bar a(A)}=x$, $z^y_A=y$, and $U^B_{(x,y)}(p)=U^B_{(x,y)}(\delta_{z^y})$ for all $B\in c(A)$. 
 We can rewrite $U^A_x(p)=\bE^p_{A|x}\:U^A_x(\delta_{z^y})$,
which is the expected utility of some lottery in $\Delta(X_{A\cup\bar d(A)})$ under a continuous and strictly increasing Bernoulli index. Hence, we can find some $z\in X$ such that $z_{\bar a(A)}=x$ and  $U^A_x(p)=U^A_x(\delta_z)$. This proves the inductive hypothesis for $A\in \cH^{k-1}$ and hence proves (iii) by induction.

 For (iv), Axiom \ref{axiom_C}  contains two parts. The second part follows from (iii) and the continuity of $U^o(\delta_y)$ on $y\in X$.   To prove the first part,  it suffices to show that for any $p,q\in \Delta(X)$, the function $f:[0,1]\rightarrow \bR$ defined by $f(\alpha)= U^o(p\alpha q)$ is continuous in $\alpha$. Again, we will prove it by induction on $\{\cH^k\}_{k=0}^m$. First, for any $A\in \cH^m$ and  $x\in \supp(p_{\bar a(A)})\cup \supp(q_{\bar a(A)})$, 
\begin{equation*}
    U^A_x(p\alpha q) = \bE^{p\alpha q}_{A|x} u^A(x,y) = \lambda_x(\alpha)   U^A_x(p) + (1-\lambda_x(\alpha) ) U^A_x(q),
\end{equation*}
in which $\lambda_x(\alpha) = \frac{\alpha p_{\bar a(A)}(x) }{\alpha p_{\bar a(A)}(x) + (1-\alpha)q_{\bar a(A)}(x)}$. Because of our convention that $U^A_x(r)=0$ for $x\not\in \supp(r_{\bar a(A)})$,  $U^A_x(p\alpha q)$ is continuous on the exact interval where $x$ has positive probability---which will be $[0,1)$, $(0,1]$, or $[0,1]$ depending on which of $p$ and $q$ support $x$. Because the consequence space is a product of compact intervals and the utility indices are continuous, all recursively defined continuation values are confined to compact ranges. 

Next we show that this boundedness resolves the endpoint issue at the parent vertex.  In the parent's expectation aggregator, the term associated with history $x$ is multiplied by its conditional marginal probability. If $x$ is supported by $p$ but not $q$, its probability vanishes as $\alpha$ goes to $0$. Because the conditional utility is bounded, the product of the vanishing probability and the conditional utility goes to zero in the limit. Since lotteries are simple, the parent's expectation is just a finite sum of these probability-weighted terms. Every term either remains perfectly continuous (if the history is supported by both lotteries) or converges smoothly to zero (if the history disappears). Since the off-support convention also yields zero at the endpoint, the aggregate finite sum at the parent level remains continuous on the entire closed interval $[0,1]$. By tracking these probability-weighted terms, the induction step holds together all the way to the root. For simplicity, we omit the endpoint issue in the following induction step.

Now assume  that $U^A_x(p\alpha q)$ is continuous in $\alpha$ for all $A\in \bigcup_{i=k}^m \cH^i $. Consider any $A\in  \cH^{k-1}$ and $x\in \supp(p_{\bar a(A)})\cup \supp(q_{\bar a(A)})$, 
\begin{align*}
    U^A_x(p\alpha q) =& \bE^{p\alpha q}_{A|x} u^A(x,y, \:(U^B_{(x,y)}(p\alpha q))_{B\in c(A)}\:)\\
    =& \sum_{y}\: (p\alpha q)_{A|x}(y)\cdot u^A(x,y, \:(U^B_{(x,y)}(p\alpha q))_{B\in c(A)}).
\end{align*}
Since $u^A(x,y, a)$ is continuous in $a$ by (ii) and $ c(A)\subseteq \bigcup_{i=k}^m \cH^i$, the inductive hypothesis implies that $u^A(x,y, \:(U^B_{(x,y)}(p\alpha q))_{B\in c(A)})$ is continuous in $\alpha$ for all $y\in \supp(p_A)\cup \supp(q_{A})$. In addition, $(p\alpha q)_{A|x}(y)=\lambda_x(\alpha)p_{A|x}(y) + (1-\lambda_x(\alpha))q_{A|x}(y)$ is continuous in $\alpha$. Hence, $U^A_x(p\alpha q)$ is continuous in $\alpha$ for $A\in  \cH^{k-1}$. By induction, $f(\alpha)= U^o(p\alpha q)$ is continuous.
\end{proof}

To verify Axiom \ref{axiom_one}, fix any $i\in A\in \cP$.
For any $p,q\in\Delta(X_i)$ and $x\in X_{-i}$, we know $p\succ_x q$ if and only if $U^A_{x_{\bar a(A)}}(p,\delta_{x})>U^A_{x_{\bar a(A)}}(q,\delta_{x})$. Since $U^A_{x_{\bar a(A)}}(p\alpha r,\delta_{x})=\alpha U^A_{x_{\bar a(A)}}(p,\delta_{x}) + (1-\alpha)U^A_{x_{\bar a(A)}}(r,\delta_{x})$ for all $r\in \Delta(X_i)$ and $\alpha\in (0,1)$, we conclude that $p\alpha r\succ_x q\alpha r$.

The following lemma is useful for proving the necessity of Axiom \ref{axiom_S}.

\begin{lemma}\label{lemma_nece_CN}
 Suppose that $\succsim$ has an SMEU representation.    If $A\perp B$, then $p\sim_z (p_A,p_B)$ for all $p\in \Delta(X_{A\cup B})$ and $z\in X_{(A\cup B)^\mathsf{c}}.$ Moreover, $\succsim_z$ admits an SMEU representation in which   for all $C\in c(o)$, 
 either $C\cup\bar d(C)\subseteq A$ or $C\cup\bar d(C)\subseteq B$.
\end{lemma}

\begin{proof}[Proof of  \lemmaref{lemma_nece_CN}]
We prove the result by induction on $|A\cup B|$. If $|A\cup B|=2$, then $A=\{i\}$ and $B=\{j\}$ for some $i\neq j$. Fix $z\in X_{\{i,j\}^\mathsf{c}}$  and we will omit the dependence of functions on $z$ when there is no risk of confusion. By \lemmaref{lemma_nece_restrict},
$\succsim_z$ over $\Delta (X_{\{i,j\}})$ has an SMEU representation with RCT $\cT$. If $c(o)=\{\{i\},\{j\}\}$, $\succsim_z$ is represented by $U(p)=u^o(\:\bE^p_i\:u^i(x_i)\:,\:\bE^p_j\:u^j(x_j)\:)$ and we are done. Otherwise, \lemmaref{lemma_nece_singleton} suggests that it is without loss to assume that $c(o)=\{\{i\}\}$ or $\{\{j\}\}$. By symmetry, assume the former holds. Then 
$\succsim_z$ is represented by $U(p)=\bE^p_i\:u^i(x_i,\:\bE^p_{j|x_i}\:u^j(x_i,x_j)\:)$.  Consider any $x_i>\lx_i$ and $x_j\in X_j$. Since $i\perp j$, we have $\frac{1}{2}\delta_{(x_i,x_j)} + \frac{1}{2}\delta_{(\lx_i,\lx_j)}\sim_z \frac{1}{2}\delta_{(x_i,\lx_j)} + \frac{1}{2}\delta_{(\lx_i,x_j)}$, which implies 
\[\frac{1}{2} u^i(x_i,u^j(x_i,x_j)) + \frac{1}{2} u^i(\lx_i,u^j(\lx_i,\lx_j)) = \frac{1}{2} u^i(\lx_i,u^j(\lx_i,x_j)) + \frac{1}{2} u^i(x_i,u^j(x_i,\lx_j)).\]
Define $\hat u^i(x_i)=u^i(x_i,u^j(x_i,\lx_j))-\frac{1}{2} u^i(\lx_i,u^j(\lx_i,\lx_j))$ and $\hat u^j(x_j)=u^i(\lx_i,u^j(\lx_i,x_j))-\frac{1}{2} u^i(\lx_i,u^j(\lx_i,\lx_j))$. Then $u^i(x_i,u^j(x_i,x_j)) = \hat u^i(x_i)+ \hat u^j(x_j)$ for all $x_i>\lx_i$ and $x_j\in X_j$. By definition, the equality also holds with $x_i=\lx_i$. Moreover, for any $x_i\in X_i$, $a= u^j(x_i,x_j)$, $a'=u^j(x_i,x'_j)$ for some $x_j,x'_j\in X_j$,  and $\alpha\in (0,1)$, we can find $x_i'$ sufficiently close to but distinct from $x_i$ and, by $i\perp j$, 
\[\frac{1}{2-\alpha}\delta_{(x_i,x_j)} + \frac{1-\alpha}{2-\alpha}\delta_{(x_i',x_j')}\sim_z \frac{1}{2-\alpha}(\delta_{x_i},\alpha\delta_{x_j} +(1-\alpha)\delta_{x_j'}) + \frac{1-\alpha}{2-\alpha}\delta_{(x_i',x_j)}.\]
This leads to 
\begin{align*}
    & u^i(x_i,a) + (1-\alpha)u^i(x_i',u^j(x_i',x_j'))\\
    = ~& u^i(x_i,\alpha a +(1-\alpha)a') + (1-\alpha)u^i(x_i',u^j(x_i',x_j)).
\end{align*}
 By continuity of $u^i(y_i,u^j(y_i,y_j))$ in $(y_i,y_j)$, let $x_i'$ converge to $x_i$ and we conclude that $\alpha u^i(x_i,a) + (1-\alpha)u^i(x_i, a') = u^i(x_i,\alpha a+(1-\alpha)a')$---i.e., $u^i$ is linear in the second argument. Hence, $\succsim_z$ is represented by 
\begin{align*}
U(p)=\bE^p_i\:u^i(x_i,\:\bE^p_{j|x_i}\:u^j(x_i,x_j)\:) = \bE^p\: u^i(x_i,u^j(x_i,x_j)) =\bE^p_i\: \hat u^i + \bE^p_j \:\hat u^j.
\end{align*}
That is,  $\succsim_z$  has another SMEU representation with RCT $\cT'$ in which $c^{\cT'}(o)=\{\{i\},\{j\}\}$. This completes the proof for $|A\cup B|=2$.  

Suppose by induction that the result holds for $|A\cup B|\leqslant m$ for some $m\geqslant  2$. Now consider $|A\cup B|=m+1$. Fix $z\in X_{(A\cup B)^\mathsf{c}}$  and we will omit the dependence of functions on $z$ when there is no risk of confusion. By \lemmaref{lemma_nece_restrict},
$\succsim_z$ over $\Delta (X_{\{A\cup B\}})$ has an SMEU representation with RCT $\cT$. First, suppose $c(o)=\{C_1,\dots,C_k\}$ with $k\geqslant  2$, where $\{C_i\cup \bar d(C_i)\}_{i=1}^k$ is a partition of $A\cup B$. Then $\succsim_z$ is represented by $U(p)=u^o(\:U^{C_1}(p),\dots,U^{C_k}(p)\:)$ for all $p\in \Delta (X_{\{A\cup B\}})$.
If $C_i\cup \bar d(C_i)\subseteq A$ or $C_i\cup \bar d(C_i)\subseteq B$ for all $i$, we are done. Otherwise, suppose $(C_i\cup \bar d(C_i))\cap A:=C_i^A\neq \emptyset$ and $(C_i\cup \bar d(C_i))\cap B:=C_i^B\neq \emptyset$. Since $A\perp B$ and $k\geqslant  2$, we have $C_i^A\perp C_i^B$ and $|C_i^A\cup C_i^B|= |C_i\cup \bar d(C_i)|\leqslant m$. By the inductive hypothesis, $U^{C_i}(p)$ can be rewritten as $\hat u^{C_i}(\:U^{C_{i1}}(p),\dots,U^{C_{il}}(p)\:)$, in which  $\{C_{ij}\cup \bar d'(C_{ij})\}_{j=1}^l$ is a partition of $C_i\cup \bar d(C_i)$,  each element of  the partition is a subset of  either $A$ or $B$, and $d'$ corresponds to the RCT over $C_i\cup \bar d(C_i)$. In this way, we derive a new SMEU representation of $\succsim_z$ over $\Delta (X_{A\cup B})$ in which the children of $o$ are $\{C_1,\dots,C_{i-1}, C_{i1},\dots, C_{il}, C_{i+1},\dots,C_k\}$. Repeat the procedure for all $i=1,\dots,k$ and we end up with an SMEU representation of $\succsim_z$ with RCT $T^*$  in which $D\cup\bar d^{\cT^*}(D)\subseteq A$ or $D\cup\bar d^{\cT^*}(D)\subseteq B$ for all $D\in c^{\cT^*}(o)$. The RCT also induces a partition of $A$ and $B$, respectively, which implies that  $p\sim_z (p_A,p_B)$ for all $p\in \Delta(X_{A\cup B})$. This proves the inductive hypothesis for the case in which  $|c(o)|\geqslant  2$.

Now suppose $|c(o)|=1$. By \lemmaref{lemma_nece_singleton}, it is without loss of generality to assume that 
$c(o)=\{\{1\}\}$ and $1\in A$. Fix $z_1\in X_1$ and consider the conditional preference $\succsim_{(z,z_1)}$ over $\Delta(X_{A\cup B\backslash\{1\}})$. 
We will consider two cases: $|c(\{1\})|=1$ and $|c(\{1\})|\geqslant  2$. Note that if $|A|\geqslant  2$, then the inductive hypothesis implies that  $\succsim_{(z,z_1)}$ has an SMEU representation that satisfies the condition in  \lemmaref{lemma_nece_CN} and we are in the latter case. Hence, if $|c(\{1\})|=1$, then it is without loss of generality to assume that $A=\{1\}$ and $c(\{1\})=\{\{j\}\}$ for some $j\in B$. Then $\succsim_z$ is represented by 
\[U(p)=\bE^p_1\:u^1(x_1,\:\bE^p_{j|x_1}\:u^j(x_1,x_j, \:(U^B_{(x_1,x_j)}(p))_{B\in c(\{j\})}\:)\:).\]
Following the same proof for the previous case with $|A\cup B|=2$, we can show that $u^1$ is linear in its second argument and 
\[U(p)=\bE^p_{1,j}\:u^1(x_1,\:u^j(x_1,x_j, \:(U^B_{(x_1,x_j)}(p))_{B\in c(\{j\})}\:)\:).\]
This implies that $\succsim_z$ has an alternative SMEU representation with RCT in which $c'(o)=\{\{j\}\}$. Since $j\in B$, $|A|=1$, and $|A\cup B|\geqslant 3$, we know $|B|\geqslant 1$. Again by the inductive hypothesis, it is without loss to assume that $|c'(\{j\})|\geqslant  2$  in this new RCT. Hence, it suffices to focus on the case in which $|c(\{1\})|\geqslant  2$. 

Denote $c(\{1\})=\{\{i_1\},\dots,\{i_n\}, \{j_1\},\dots,\{j_k\}\}$ for some $n\geqslant  1, k\geqslant  0$ with $k+n\geqslant  2$, $i_l\in B$, and $j_l\in A$.  When $k=0$, $A=\{1\}$ and $c(\{1\})=\{\{i_1\},\dots,\{i_n\}\}$.
Let $A_l= \{j_l\}\cup \bar d(\{j_l\})$ for all $l\leqslant k$ and $B_l= \{i_l\}\cup \bar d(\{i_l\})$ for all $l\leqslant n$. Then $\{A_l\}_{l=1}^k$ is a partition of $A\backslash\{1\}$ and $\{B_l\}_{l=1}^n$ is a partition of $B$. The conditional preference $\succsim_z$ can be represented by 
\[U(p)=\bE^p_1\:u^1(x_1,\: U^{i_1}_{x_1}(p), \:(U^{i_l}_{x_1}(p))_{l=2}^n\:, \:(U^{j_l}_{x_1}(p))_{l=1}^k\:).\]

Since $1 \perp i_1$, following the same proof for the previous case with $|A \cup B|=2$, we can show that $u^1$ is linear in its second argument. Since $k+n \geqslant 2$, $|B_l| < m$ for all $l=1, \ldots, n$, and we can apply the inductive hypothesis for $\{1\} \perp B_l$ to show that the conditional preference on $B_l$ does not depend on the outcome in dimension 1. Consequently, the history-specific continuation representations over $B_l$ differ across $x_1$ only by positive affine transformations. Because the parent aggregator $u^1$ is permitted to depend on $x_1$, we can choose a common normalized continuation representation for the subtree $B_l$ and adjust $u^1(x_1, \cdot)$ to absorb these $x_1$-dependent affine coefficients without altering the overall preference. Under this equivalent representation, $u^{i_l}$ and $u^C$ for all $C \in d(\{i_l\})$ strictly do not depend on the outcome in dimension 1. Hence, $U_{x_1}^{i_l}(p)=U_{x_1'}^{i_l}(p')$ if $p_{B_l \mid x_1}=p_{B_l \mid x_1'}^{\prime}$. Construct $\hat p, \tilde p\in \Delta(X)$ such that 
\begin{align*}
    \hat p = & \:\alpha\:(z, \:{x_1}, \:{x_{i_1}},\:q_{B_1\backslash\{i_1\}},\:(q_{B_l})_{l=2}^n, \:(q_{A_l})_{l=1}^k)  \: + \alpha\:(z, \:{\lx_1}, \:{\lx_{i_1}},\:q_{B_1\backslash\{i_1\}},\:(q_{B_l})_{l=2}^n, \:(q_{A_l})_{l=1}^k) \\
    & \: + (1-2\alpha)\:(z, \:{\lx_1}, \:{\lx_{i_1}},\:q_{B_1\backslash\{i_1\}},\:(\lx_{B_l})_{l=2}^n, \:(\lx_{A_l})_{l=1}^k),\\
    \tilde p = & \:\alpha\:(z, \:{\lx_1}, \:{x_{i_1}},\:q_{B_1\backslash\{i_1\}},\:(q_{B_l})_{l=2}^n, \:(q_{A_l})_{l=1}^k)  \: + \alpha\:(z, \:{x_1}, \:{\lx_{i_1}},\:q_{B_1\backslash\{i_1\}},\:(q_{B_l})_{l=2}^n, \:(q_{A_l})_{l=1}^k) \\
    & \: + (1-2\alpha)\:(z, \:{\lx_1}, \:{\lx_{i_1}},\:q_{B_1\backslash\{i_1\}},\:(\lx_{B_l})_{l=2}^n, \:(\lx_{A_l})_{l=1}^k)
\end{align*}
where $\alpha\in (0,\frac{1}{2})$, $x_1>\lx_1$, $x_{i_1}\in X_{i_1}$, and $q_C$ is a lottery in $\Delta(X_{C})$ for each $C$. We can verify that  $\hat p_{\{1,i_1\}^\mathsf{c}} = \tilde p_{\{1,i_1\}^\mathsf{c}}$ and $\hat p_{1|z'}=\tilde p_{1|z'}$, $\hat p_{i_1|z'}=\tilde p_{i_1|z'}$ for all $z'\in \supp(\hat p_{\{1,i_1\}^\mathsf{c}})$. By the definition of $1\perp i_1$, we know $\hat p\sim \tilde p$. By linearity of $u^1$ in its second argument, we can derive the following equation:
\begin{align*}
    & u^1(x_1,\: u^{i_1}(x_{i_1},(U^B_{(x_1,x_{i1})}(\hat p))_{B\in c(\{i_1\})}), \:(U^{i_l}_{x_1}(\hat p))_{l=2}^n\:, \:(U^{j_l}_{x_1}(\hat p))_{l=1}^k\:)\\
   &~ - u^1(\lx_1,\: u^{i_1}(x_{i_1},(U^B_{(\lx_1,x_{i1})}(\hat p))_{B\in c(\{i_1\})}), \:(U^{i_l}_{\lx_1}(\hat p))_{l=2}^n\:, \:(U^{j_l}_{\lx_1}(\hat p))_{l=1}^k\:)\\
   = \:& u^1(x_1,\: u^{i_1}(\lx_{i_1},(U^B_{(x_1,\lx_{i1})}(\hat p))_{B\in c(\{i_1\})}), \:(U^{i_l}_{x_1}(\hat p))_{l=2}^n\:, \:(U^{j_l}_{x_1}(\hat p))_{l=1}^k\:)\\
   &~ - u^1(\lx_1,\: u^{i_1}(\lx_{i_1},(U^B_{(\lx_1,\lx_{i1})}(\hat p))_{B\in c(\{i_1\})}), \:(U^{i_l}_{\lx_1}(\hat p))_{l=2}^n\:, \:(U^{j_l}_{\lx_1}(\hat p))_{l=1}^k\:).
\end{align*}
We only used $\hat p$ in the above equation because $\hat p$ and $\tilde p$ give the same value for all relevant functions. By Axiom \ref{axiom_C} (\lemmaref{lemma_continuity}), letting $\alpha$ converge to $0$, the above equation becomes 
\begin{equation}\label{eq_nece_CN1}
    U(p^1) - U(p^2) = U(p^3) - U(p^4),
\end{equation}
where
\begin{align*}
    p^1 & = (z, \:{x_1}, \:{x_{i_1}},\:q_{B_1\backslash\{i_1\}},\:(q_{B_l})_{l=2}^n, \:(q_{A_l})_{l=1}^k),\\
    p^2 &= (z, \:{\lx_1}, \:{x_{i_1}},\:q_{B_1\backslash\{i_1\}},\:(\lx_{B_l})_{l=2}^n, \:(\lx_{A_l})_{l=1}^k),\\
    p^3 & = (z, \:{x_1}, \:{\lx_{i_1}},\:q_{B_1\backslash\{i_1\}},\:(q_{B_l})_{l=2}^n, \:(q_{A_l})_{l=1}^k),\\
    p^4 &= (z, \:{\lx_1}, \:{\lx_{i_1}},\:q_{B_1\backslash\{i_1\}},\:(\lx_{B_l})_{l=2}^n, \:(\lx_{A_l})_{l=1}^k).
\end{align*}

Next, fix $\lx_{i1}$ and consider the conditional preference $\succsim_{(z,\lx_{i1})}$ over $A\cup B\backslash\{i_1\}$. Since $|A\cup B\backslash\{i_1\}|=m$ and $A\perp B\backslash\{i_1\}$, the inductive hypothesis implies the existence of an SMEU representation of $\succsim_{(z,\lx_{i1})}$ with an RCT $\cT^1$ in which $\{C\cup\bar d(C)\}_{C\in c^{\cT^1}(o)}$ includes a partition of $B_1\backslash\{i_1\}$, $A_l$ for $l\leqslant k$, and $B_l$ for $2\leqslant l\leqslant n$. By the corresponding SMEU representation, we have $\frac{1}{2}p^3 + \frac{1}{2}p^5\sim \frac{1}{2}p^4 + \frac{1}{2}p^6$ and hence, 
\begin{equation}\label{eq_nece_CN2}
    U(p^3) - U(p^4) = U(p^6) - U(p^5),
\end{equation}
where
\begin{align*}
    p^5 &= (z, \:{\lx_1}, \:{\lx_{i_1}},\:\lx_{B_1\backslash\{i_1\}},\:(\lx_{B_l})_{l=2}^n, \:(\lx_{A_l})_{l=1}^k)\\
    p^6 & = (z, \:{x_1}, \:{\lx_{i_1}},\:\lx_{B_1\backslash\{i_1\}},\:(q_{B_l})_{l=2}^n, \:(q_{A_l})_{l=1}^k).
\end{align*}
We claim that the above two equations (\ref{eq_nece_CN1}) and (\ref{eq_nece_CN2}) also hold for $x_1=\lx_1$. To see this, we can repeat the above arguments and derive counterparts of equations (\ref{eq_nece_CN1}) and (\ref{eq_nece_CN2}) by switching the roles of $x_1$ and $\lx_1$ and replacing $(\lx_{A_l})_{l=1}^k)$ with $(q_{A_l})_{l=1}^k)$ in the construction of $\hat p$ and $\tilde p$. These four equations deliver what we want.

Combining equations (\ref{eq_nece_CN1}) and (\ref{eq_nece_CN2}) and substituting the result into the expression of $U(p)$ for any $p$ such that $p_{(A\cup B)^\mathsf{c}}=\delta_z$, we get
\begin{align*}
    U(p)=&  ~\bE^p_{\{1,i_1\}}\: U(z, \:{x_1}, \:{x_{i_1}},\:p_{B_1\backslash\{i_1\}|x_1,x_{i_1}},\:(p_{B_l|x_1})_{l=2}^n, \:(p_{A_l|x_1})_{l=1}^k) \\
    =&~  \bE^p_{\{1,i_1\}}\: U(z, \:{\lx_1}, \:{x_{i_1}},\:p_{B_1\backslash\{i_1\}|x_1,x_{i_1}},\:(\lx_{B_l})_{l=2}^n, \:(\lx_{A_l})_{l=1}^k)\\
    & + \bE^p_{1}\: U(z, \:{x_1}, \:{\lx_{i_1}},\:\lx_{B_1\backslash\{i_1\}},\:(p_{B_l|x_1})_{l=2}^n, \:(p_{A_l|x_1})_{l=1}^k) \\
    & - U(z, \:{\lx_1}, \:{\lx_{i_1}},\:\lx_{B_1\backslash\{i_1\}},\:(\lx_{B_l})_{l=2}^n, \:(\lx_{A_l})_{l=1}^k).
\end{align*}
The third term is a constant. For the first term, note that $\{1\}\perp B_1$. Consider a sequence of lotteries $q^t\in \Delta(X_{B_1\cup\{1\}})$ such that $q^t_1 \rightarrow \delta_{\lx_1}$ as $t$ goes to $\infty$  and the distribution over conditional distributions on $B_1$ induced by $q^t$ agrees with that induced by $p$ for all $t\geqslant  1$. 
This implies that $q^t_{B_1}=p_{B_1}$ for all $t\geqslant  1$. 
By the inductive hypothesis, we have 
$q^t\sim (q^t_1,p_{B_1})$. That is, 
\begin{align*}
    &\bE^{q^t}_{\{1,i_1\}}\: U(z, \:{x_1}, \:{x_{i_1}},\:q^t_{B_1\backslash\{i_1\}|x_1,x_{i_1}},\:(\lx_{B_l})_{l=2}^n, \:(\lx_{A_l})_{l=1}^k)
    \\
   =  & ~\bE^{q^t}_{1} \bE^{p}_{i_1}\: U(z, \:{x_1}, \:{x_{i_1}},\:q^t_{B_1\backslash\{i_1\}|x_{i_1}},\:(\lx_{B_l})_{l=2}^n, \:(\lx_{A_l})_{l=1}^k). 
\end{align*}
We can view $x_1$ as a random variable with distribution $q^t_1$. Since the induced distribution of $q^t_{B_1\backslash\{i_1\}|x_1,x_{i_1}}$ is the same as that of $p_{B_1\backslash\{i_1\}|x_1,x_{i_1}}$, we have 
\begin{align*}
    &\bE^{q^t}_{\{1,i_1\}}\: U(z, \:{x_1}, \:{x_{i_1}},\:p_{B_1\backslash\{i_1\}|x_1,x_{i_1}},\:(\lx_{B_l})_{l=2}^n, \:(\lx_{A_l})_{l=1}^k)
    \\
   =  & ~\bE^{q^t}_{1} \bE^{p}_{i_1}\: U(z, \:{x_1}, \:{x_{i_1}},\:p_{B_1\backslash\{i_1\}|x_{i_1}},\:(\lx_{B_l})_{l=2}^n, \:(\lx_{A_l})_{l=1}^k). 
\end{align*}
Moreover, both expressions above can be viewed as an expectation of a function that is continuous in $x_1$. As $t$ goes to $\infty$, $q^t_1 \rightarrow \delta_{\lx_1}$ implies that 
\begin{align*}
    &\bE^p_{\{1,i_1\}}\: U(z, \:{\lx_1}, \:{x_{i_1}},\:p_{B_1\backslash\{i_1\}|x_1,x_{i_1}},\:(\lx_{B_l})_{l=2}^n, \:(\lx_{A_l})_{l=1}^k)
    \\
   =  & ~\bE^p_{i_1}\: U(z, \:{\lx_1}, \:{x_{i_1}},\:p_{B_1\backslash\{i_1\}|x_{i_1}},\:(\lx_{B_l})_{l=2}^n, \:(\lx_{A_l})_{l=1}^k). 
\end{align*}

For the second term, by applying the inductive hypothesis to $A\perp (B\backslash{B_1})$, we have
\begin{align*}
   & \sum p_1(x_1)\:(z, \:{x_1}, \:{\lx_{i_1}},\:\lx_{B_1\backslash\{i_1\}},\:(p_{B_l|x_1})_{l=2}^n, \:(p_{A_l|x_1})_{l=1}^k)\\
    \sim &~ \sum p_1(x_1)\:(z, \:{x_1}, \:{\lx_{i_1}},\:\lx_{B_1\backslash\{i_1\}},\:(p_{B_l})_{l=2}^n, \:(p_{A_l|x_1})_{l=1}^k).
\end{align*}
This leads to 
\begin{align*}
    U(p)=&  \bE^p_{i_1}\: U(z, \:{\lx_1}, \:{x_{i_1}},\:p_{B_1\backslash\{i_1\}|x_{i_1}},\:(\lx_{B_l})_{l=2}^n, \:(\lx_{A_l})_{l=1}^k)\\
    & + \bE^p_{1}\: U(z, \:{x_1}, \:{\lx_{i_1}},\:\lx_{B_1\backslash\{i_1\}},\:(p_{B_l})_{l=2}^n, \:(p_{A_l|x_1})_{l=1}^k) \\
    & - U(z, \:{\lx_1}, \:{\lx_{i_1}},\:\lx_{B_1\backslash\{i_1\}},\:(\lx_{B_l})_{l=2}^n, \:(\lx_{A_l})_{l=1}^k).
\end{align*}

Repeating the above arguments for $i_2,\dots, i_k$ yields (ignoring the constant  terms)
 \begin{align*}
    U(p) = &~\bE^p_{1}\: U(z, x_1,\:(\lx_{B_l})_{l=1}^n, \:(p_{A_l|x_1})_{l=1}^k) + \sum_{l=1}^n \bE^p_{i_l}\: U(z, \:{\lx_A}, \:{x_{i_l}},\:p_{B_l\backslash\{i_l\}|x_{i_l}},\:(\lx_{B_t})_{t\neq l}).
\end{align*}
This is an SMEU representation with RCT $\cT^*$ where $c^{\cT^*}(o)=\{\{1\},\{i_1\},\dots,\{i_n\}\}$, $\{1\}\cup\bar d^{\cT^*}(\{1\})=A$, and  $\{i_l\}\cup \bar d^{\cT^*}(\{i_l\})=B_l$ for all $l=1,\dots,n$. This completes the proof for the inductive step for $|A\cup B|=m+1$.
\end{proof}

Now we are ready to show that Axiom \ref{axiom_S} holds. 
Suppose that $B$ is minimally separable from $A$. By \lemmaref{lemma_nece_CN}, for any $z\in X_{A^\mathsf{c}}$, the conditional preference has an SMEU representation $(\cT,(u^C)_{C\in\cP_o})$ in which   either $C\cup\bar d(C)\subseteq B$ or $C\cup\bar d(C)\subseteq A\backslash B$ for all $C\in c(o)$. If $C\cup\bar d(C)\subsetneq B$ for some $C\in c(o)$, then $C\cup\bar d(C)\perp (A\backslash(C\cup\bar d(C)))$, which leads to a contradiction. Hence, $C^*\cup\bar d(C^*)= B$ for some $C^*\in c(o)$. For all $r\in \Delta(X_{A\setminus B})$, and $p,q\in \Delta(X_A)$ such that $p_{A\setminus B}=q_{A\setminus B}$, we have $p\succsim_z q \iff U^{C^*}(p)\geqslant  U^{C^*}(q)\iff U^{C^*}(p_B,r)\geqslant  U^{C^*}(q_B,r) \iff (p_B, r)\succsim_z (q_B,r).$

To prove Axiom \ref{axiom_RI}, we need the  following lemma. Recall that for an RCT $(\cP,E)$, $H(i)$  is the element in $\cP$ that includes dimension $i$.

\begin{restatedlemma}{lemma_nece_RI}
        Suppose that $\succsim$ has an SMEU representation. If $i\in A$ and $i\hra A$, then  for all $z\in X_{A^\mathsf{c}}$, there is an SMEU representation of  $\succsim_z$ in which $A\subseteq H(i)\cup \bar d(H(i))$.
\end{restatedlemma}
 
\begin{proof}[Proof of \lemmaref{lemma_nece_RI}]
    We will prove the result by induction on the cardinality of $A$. Suppose $i\in A$ and $i\hra   A$. Fix any $z\in X_{A^\mathsf{c}}$ and we will omit the dependence of functions on $z$ when there is no risk of confusion. By \lemmaref{lemma_nece_restrict}, $\succsim_z$ has an SMEU representation.

First, let $A=\{i,j\}$. For any  $\cT\in \mathbb{T}(\succsim_z)$, the result holds if $\{i\}\in c(o)$ or $\{i,j\}\in c(o)$. Hence we can focus on two cases: $c(o)=\{\{j\}\}$ or $c(o)=\{\{i\},\{j\}\}$. 

In the former case, $\succsim_z$ is represented by $U(p)=\bE^p_j\:u^j(x_j,\:\bE^p_{i|x_j}\:u^i(x_j,x_i)\:)$. It suffices to show that $u^j$ is linear in the second argument. or any $x_j \in X_j$ and $a = u^i(x_j, x_i) > a' = u^i(x_j, x_i')$ for some $x_i, x_i' \in X_i$, continuity and strict monotonicity of the utility functions ensure that we can find a nearby point $(\hat{x}_i, \hat{x}_j)$ on the same indifference contour as $(x_i',x_j)$ such that $\hat{x}_i \neq x_i, x_i'$ and $\hat{x}_j \neq x_j$. (At boundaries, we can construct this indifference around $(x_i, x_j)$ instead).

Let $\hat{a} = u^i(\hat{x}_j, \hat{x}_i)$. Because $(\hat{x}_i, \hat{x}_j)$ is overall indifferent to $(x_i', x_j)$, we have $u^j(\hat{x}_j, \hat{a}) = u^j(x_j, a')$. Since the two consequences are indifferent, applying the definition of $i \rightharpoonup j$ in both directions yields the mixture indifference:
$$ \alpha \delta_{(x_i,x_j)} + (1-\alpha) \delta_{(x_i',x_j)} \sim_z \alpha \delta_{(x_i,x_j)} + (1-\alpha) \delta_{(\hat{x}_i,\hat{x}_j)} $$
for all $\alpha \in(0,1)$. Evaluating both sides through the recursive representation, we obtain:
$$ u^j(x_j, \alpha a + (1-\alpha) a') = \alpha u^j(x_j, a) + (1-\alpha) u^j(\hat{x}_j, \hat{a}). $$
Substituting $u^j(\hat{x}_j, \hat{a}) = u^j(x_j, a')$ directly into the right-hand side yields $u^j(x_j, \alpha a + (1-\alpha) a') = \alpha u^j(x_j, a) + (1-\alpha) u^j(x_j, a')$, establishing the required linearity.  Hence, $U(p)=\bE^p\:u^j(x_j,\:u^i(x_j,x_i))$ and the conclusion holds. 

In the latter case, $\succsim_z$ is represented by $U(p)=u^o(\:\bE^p_i\:u^i(x_i)\:,\:\bE^p_j\:u^j(x_j)\:)$, which is continuous on $p\in \Delta(X_{\{i,j\}})$. Below, we  show that $\succsim_z$ satisfies independence. For any $p,q,r\in \Delta(X_{\{i,j\}})$ such that $p\succ q$, \lemmaref{lemma_continuity} implies the existence of $x,x',y\in X_{\{i,j\}}$ such that $u^l(x_l)= \bE^p_l\:u^l$, $u^l(x'_l)= \bE^q_l\:u^l$, and $u^l(y_l)= \bE^r_l\:u^l$ for $l=i,j$. Then $p\sim_z x\succ_z q\sim_z x'$ and $r\sim_z y$. For now assume that $y_i\in (\lx_i,\ux_i)$. By continuity and monotonicity of $u^o,u^i$, and $u^j$, there exist $\hat x,\hat x'\in X_{\{i,j\}}$ such that (i) $x\geqslant  \hat x$, (ii) $\hat x'\geqslant  x'$, (iii) $x\succ_z \hat x \sim_z \hat x'\succ_z x'$, and (iv) $y_i\neq \hat x_i, \hat x_i'$. Since $i\hra j$, for any $\alpha\in (0,1)$, we have 
\begin{equation}\label{eq_nece_RI1}
    \alpha \delta_{\hat x} + (1-\alpha)\delta_y \sim_z \alpha\delta_{\hat x'} + (1-\alpha)\delta_y.
\end{equation}
We claim that  (\ref{eq_nece_RI1}) also holds when $y_i=\lx_i$ or $\ux_i$. In either case, we can find a sequence $(y^n)_{n\geqslant  1}$ such that  $y^n\rightarrow y$ and $y^n_i\in (\lx_i,\ux_i)$. Then $\alpha \delta_{\hat x} + (1-\alpha)\delta_{y^n} \sim_z \alpha\delta_{\hat x'} + (1-\alpha)\delta_{y^n}$ for all $n\geqslant  1$. Since $\succsim_z$ is represented by a continuous function on $\Delta(X_{\{i,j\}})$, the indifference holds as $n$ goes to infinity, implying that $ \alpha \delta_{\hat x} + (1-\alpha)\delta_y \sim_z \alpha\delta_{\hat x'} + (1-\alpha)\delta_y$. 

For any $\alpha\in (0,1)$, it holds that
\begin{align*}
    U(\alpha p + (1-\alpha )r ) & =U(\alpha p_i + (1-\alpha )r_i, \alpha p_j + (1-\alpha )r_j)\\
    & > U(\alpha \delta_{\hat x_i} + (1-\alpha )\delta_{y_i}, \alpha \delta_{\hat x_j} + (1-\alpha )\delta_{y_j})\\
    & = U(\alpha \delta_{\hat x_i'} + (1-\alpha )\delta_{y_i}, \alpha \delta_{\hat x_j'} + (1-\alpha )\delta_{y_j})\\
    & > U(\alpha \delta_{ x_i'} + (1-\alpha )\delta_{y_i}, \alpha \delta_{ x_j'} + (1-\alpha )\delta_{y_j})\\
    & = U(\alpha q_i + (1-\alpha )r_i, \alpha q_j + (1-\alpha )r_j)\\
    & = U(\alpha q + (1-\alpha )r ).
\end{align*}
That is, $\alpha p + (1-\alpha )r\succ_z \alpha q + (1-\alpha )r$. Hence, $\succsim_z$ satisfies independence on $\Delta(X_{\{i,j\}})$ and has an expected utility representation. This proves the inductive hypothesis for $|A|=2$. 

Now suppose that the inductive hypothesis holds for all $|A|\leqslant n$ for some $n\geqslant  2$. Consider any $A$ with $|A|=n+1$ and $i\hra   A$. By \lemmaref{lemma_nece_singleton}, let the partition in the RCT for $\succsim_z$ consist of singleton sets. If $c(o)=\{\{i\}\}$, we are done. If  $c(o)=\{\{j\}\}$ with $j\neq i$, then by the inductive hypothesis, $\succsim_z$ has the following SMEU representation: $$U(p)=\bE^p_j\:u^j(x_j,\:\bE^p_{i|x_j}\:u^i(x_j,x_i, \:(U^B_{(x_j,x_i)}(p))_{B\in c(\{i\})}\:)\:).$$
Since $i\hra j$, using the same argument for the case with $|A|=2$ above, we can show that $u^j$ is linear in its second argument and hence $U(p)=\bE^p_{\{i,j\}}\:u^j(x_j,\:u^i(x_j,x_i, \:(U^B_{(x_j,x_i)}(p))_{B\in c(\{i\})}\:)\:).$ This proves the inductive hypothesis.

It remains to show the inductive hypothesis for the case in which $|c(o)|\geqslant  2$. Without loss of generality, let $c(o)=\{\{j_1\},\dots,\{j_m\}\}$ for some $m\geqslant  2$ such that $i\in \{j_1\}\cup \bar d(\{j_1\})$. Since $i\hra \{j_1\}\cup \bar d(\{j_1\})$ and $|\{j_1\}\cup \bar d(\{j_1\})|\leqslant n$, 
by the inductive hypothesis, it is without loss to assume that $j_1=i$. Denote $B_l=\{j_l\}\cup\bar d(\{j_l\})$ for all $l=1,\dots,m$. Then $\succsim_z$ is represented by
\begin{equation}\label{eq_nece_RI2}
    U(p)= u^o(U^{j_1}(p),\dots,U^{j_m}(p)),
\end{equation}
where $U^{j_l}(p)= \bE^p_{j_l}\:u^{j_l}(x_{j_l}, \:(U^B_{x_{j_l}}(p))_{B\in c(\{j_l\})}\:)$ for all $l=1,\dots,m$. The following lemma shows that the conditional preference on each $B_l$ satisfies independence for $l\geqslant  2$.

\begin{lemma}\label{lemma_nece_EU}
     Suppose $\succsim_z$ has an SMEU representation (\ref{eq_nece_RI2}). If $j_1\hra j_l$ for some $l\geqslant  2$, then for each $z'\in X_{A \backslash B_l}$, the conditional preference $\succsim_{(z,z')}$ over $\Delta(X_{B_l})$ satisfies independence. 
\end{lemma}

\begin{proof}[Proof of \lemmaref{lemma_nece_EU}]
   For any fixed $\hat z\in X_{(B_l\cup\{j_1\})^\mathsf{c}}$ with $\hat z_{A^{\mathsf{c}}}=z$, it suffices to work with the conditional preference $\succsim_{\hat z}$ over  $\Delta(X_{B_l\cup\{j_1\}})$, which is represented by the following adaptation of equation (\ref{eq_nece_RI2}): 
   \[\hat U(p)= \hat u^o(\:\bE^p_{j_1}\:\hat u^{j_1}(x_{j_1}),\: \bE^p_{j_l}\:\hat u^{j_l}(x_{j_l}, \:(\hat U^B_{x_{j_l}}(p))_{B\in c(\{j_l\})}\:)\:).\]

    For any $q\in \Delta(X_{B_l\backslash\{j_l\}})$, denote $q=\sum_{k=1}^K\pi^k\delta_{y^k}$ for some  $y^1,\dots,y^K\in X_{B_l\backslash\{j_l\}}$. Define $z^k\in X_{\{{j_1},{j_l}\}^c}$ such that $z^k_{(B_l\cup\{j_1\})^\mathsf{c}}=\hat z$ and $z^k_{B_l\backslash\{j_l\}}=y^k$ for all $k=1,\dots,K$. For any fixed $x_{j_l}\in X_{j_l}$, we can find $x_{j_1}^k,\hat x_{j_1}^k\in X_{j_1}$ and $x_{j_l}^k\in X_{j_l}$ for $k=1,\dots,K$ such that $x_{j_1}^k\neq x_{j_1}^{k'}$,  $x_{j_l}^k\neq x_{j_l}^{k'}$, and $\hat x_{j_1}^k\neq \hat x_{j_1}^{k'}$ for all $k\neq k'$, and $(z^k,x_{j_1}^k,x^k_{j_l})\sim (z^k,\hat x_{j_1}^k,x_{j_l})$. Since $j_1\hra j_l$, we have $p:=\sum_{k=1}^n\pi^k\delta_{(z^k,x_{j_1}^k,x^k_{j_l})}\sim \sum_{k=1}^n\pi^k\delta_{(z^k,\hat x_{j_1}^k,x_{j_l})}:=\hat p$. 
      That is, 
    \begin{align*}
        & \hat u^o\left(\:\sum_{k=1}^K \pi^k\: \hat u^{j_1}(x_{j_1}^k), \:\sum_{k=1}^K \pi^k\:\hat u^{j_l}(x_{j_l}^k, \:(\hat U^B_{x_{j_l}^k}(p_{B_l\cup\{j_1\}}))_{B\in c(\{j_l\})}\:) \right)\\
      = &~ \hat u^o\left(\:\sum_{k=1}^K \pi^k\: \hat u^{j_1}(\hat x_{j_1}^k), \:\hat U^{j_l}(\hat p_{B_l\cup\{j_1\}}) \right).
    \end{align*}
    Note that  $\hat U^B_{x_{j_l}^k}(p_{B_l\cup\{j_1\}})=\hat  U^B_{x_{j_l}^k}(x^k_{j_1}, x^k_{j_l}, y^k)$ for each $k\geqslant  1$ and $B\in c(\{j_l\})$, implying that  $\hat u^{j_l}(x_{j_l}^k, \:(\hat U^B_{x_{j_l}^k}(p_{B_l\cup\{j_1\}}))_{B\in c(\{j_l\})}\:) =\hat U^{j_l}(x^k_{j_1}, x^k_{j_l}, y^k)$, which does not depend on $x^k_{j_1}$ and is continuous in $x^k_{j_l}$. Also, $\hat u^{j_1}$ is continuous. Letting $ x_{j_1}^k$ converge to $\hat x_{j_1}^k$ for each $k$, we have $x^k_{j_l}$ converge to $x_{j_l}$ for each $k$ and 
      \begin{align*}
        & \hat u^o\left(\sum_{k=1}^K \pi^k\: \hat u^{j_1}(\hat x_{j_1}^k), \:\sum_{k=1}^K \pi^k\:\hat U^{j_l}(\hat x_{j_1}^k,x_{j_l}, y^k) \right) = \hat u^o\left(\sum_{k=1}^K \pi^k\: \hat u^{j_1}(\hat x_{j_1}^k), \:\hat U^{j_l}(\hat p_{B_l\cup\{j_1\}}) \right).
    \end{align*}
    By \lemmaref{lemma_continuity}, $\hat u^o$ is strictly increasing and hence 
    \[\hat U^{j_l}(\hat p_{B_l\cup\{j_1\}}) = \sum_{k=1}^K \pi^k\:\hat U^{j_l}(\hat x_{j_1}^k,x_{j_l}, y^k).\]
    Also, $\hat U^{j_l}(\hat p_{B_l\cup\{j_1\}})$ only depends on  $\hat p_{B_l} = (q, \delta_{x_{j_l}})$.  For each $z'\in X_{A \backslash B_l}$ and $r,r'\in \Delta(X_{B_l})$, we have 
    \begin{align*}
        r\succsim_{(z,z')}r'\iff  & \bE^r_{j_l}\:\hat U^{j_l}(z'_{j_1}, x_{j_l}, r_{B_l\backslash\{j_l\}|x_{j_l}})\geqslant \bE^{r'}_{j_l}\:\hat U^{j_l}(z'_{j_1}, x_{j_l}, r'_{B_l\backslash\{j_l\}|x_{j_l}})\\
        \iff & \bE^r_{B_l}\:\hat U^{j_l}(z'_{j_1}, x_{B_l})\geqslant \bE^{r'}_{B_l}\:\hat U^{j_l}(z'_{j_1}, x_{B_l}),
    \end{align*}
    where $U^{j_l}(z'_{j_1}, x_{B_l})$ does not depend on $z'_{j_1}$. 
This completes the proof.  
\end{proof}

Since $j_1=i\hra l$ for all $l\in A$, \lemmaref{lemma_nece_EU} implies that it is without loss of generality to assume that $U^{j_l}(p)$ in equation (\ref{eq_nece_RI2}) is an expected utility function for each $l\geqslant 2$. 

Now we continue the proof of \lemmaref{lemma_nece_RI}. First, consider the case in which $\bar d(\{i\})=\emptyset$ (i.e., $B_1=\{i\}$). Since $\succsim_z$ is represented by (\ref{eq_nece_RI2}), it is continuous. Below, we show that $\succsim_z$ satisfies independence on $\Delta(X_A)$ using a similar proof technique as in the case with $|A|=2$. For any $p,q,r\in \Delta(X_{A})$ such that $p\succ q$, \lemmaref{lemma_continuity} implies the existence of $x,x',y\in X_{A}$ such that $u^{j_l}(x_{B_l})= \bE^p_{B_l}\:u^{j_l}$, $u^{j_l}(x'_{B_l})= \bE^q_{B_l}\:u^{j_l}$, and $u^{j_l}(y_{B_l})= \bE^r_{B_l}\:u^{j_l}$ for $l=1,\dots,m$. Then $p\sim_z x\succ_z q\sim_z x'$ and $r\sim_z y$. For now assume that $y_i\in (\lx_i,\ux_i)$. 
By compactness of $X$ and continuity and monotonicity of $u^o$ and $u^{j_l}$ for all $l$, there exists a finite sequence $(\hat x^t)_{t=1}^T$ in $X_A$ such that (i) $x\geqslant  \hat x^1$,   (ii) $\hat x^T\geqslant  x'$, (iii) $x\succ_z \hat x^1 \sim_z\cdots\sim_z \hat x^T\succ_z x'$, (iv) $y_i\neq \hat x_i^t$ for all $t=1,\dots,T$, and (iv) $\hat x^{t}$ and $\hat x^{t+1}$ only differ in dimension $i$ and some other dimension $l_t\in A$ for all $t=1,..,T-1$. For each $t=1,\dots,T-1$, since $i\hra l_t$, $\hat x^t_{A\backslash\{i,l_t\}} = \hat x^{t+1}_{A\backslash\{i,l_t\}}$, $y_i\neq \hat x^t_i, \hat x^{t+1}_i$, and $\hat x^t\sim_z  \hat x^{t+1}$, we have $\alpha \delta_{\hat x^t} + (1-\alpha)\delta_y \sim_z \alpha\delta_{\hat x^{t+1}} + (1-\alpha)\delta_y$. By transitivity of $\sim_z$, 
\begin{equation}\label{eq_nece_RI3}
    \alpha \delta_{\hat x^1} + (1-\alpha)\delta_y \sim_z \alpha\delta_{\hat x^{T}} + (1-\alpha)\delta_y.
\end{equation}
We claim that (\ref{eq_nece_RI3}) also holds when $y_i=\lx_i$ or $\ux_i$. In either case, we can find a sequence $(y^n)_{n\geqslant  1}$ such that  $y^n\rightarrow y$ and $y^n_i\in (\lx_i,\ux_i)$. Then $\alpha \delta_{\hat x^1} + (1-\alpha)\delta_{y^n} \sim_z \alpha\delta_{\hat x^T} + (1-\alpha)\delta_{y^n}$ for all $n\geqslant  1$. Since $\succsim_z$ is represented by a continuous function on $\Delta(X_{A})$, the indifference holds as $n$ goes to infinity, implying that $\alpha \delta_{\hat x^1} + (1-\alpha)\delta_y \sim_z \alpha\delta_{\hat x^{T}} + (1-\alpha)\delta_y$. 

For any $\alpha\in (0,1)$, it holds that
\begin{align*}
    U(\alpha p + (1-\alpha )r ) & =U(\alpha p_{B_1} + (1-\alpha )r_{B_1},\cdots,\: \alpha p_{B_m} + (1-\alpha )r_{B_m})\\
    & = U(\alpha \delta_{x_{B_1}} + (1-\alpha )\delta_{y_{B_1}},\cdots,\: \alpha \delta_{x_{B_m}} + (1-\alpha )\delta_{y_{B_m}})\\
    & > U(\alpha \delta_{\hat x^1_{B_1}} + (1-\alpha )\delta_{y_{B_1}},\cdots,\: \alpha \delta_{\hat x^1_{B_m}} + (1-\alpha )\delta_{y_{B_m}})\\
    & =U(\alpha \delta_{\hat x^1} + (1-\alpha)\delta_y)\\
    & =U(\alpha \delta_{\hat x^T} + (1-\alpha)\delta_y)\\
    & = U(\alpha \delta_{\hat x^T_{B_1}} + (1-\alpha )\delta_{y_{B_1}},\cdots,\: \alpha \delta_{\hat x^T_{B_m}} + (1-\alpha )\delta_{y_{B_m}})\\
    & > U(\alpha \delta_{x'_{B_1}} + (1-\alpha )\delta_{y_{B_1}},\cdots,\: \alpha \delta_{x'_{B_m}} + (1-\alpha )\delta_{y_{B_m}})\\
    & = U(\alpha q_{B_1} + (1-\alpha )r_{B_1},\cdots,\: \alpha q_{B_m} + (1-\alpha )r_{B_m})\\
    & = U(\alpha q + (1-\alpha )r ).
\end{align*}
That is, $\alpha p + (1-\alpha )r\succ_z \alpha q + (1-\alpha )r$. Hence, $\succsim_z$ satisfies independence on $\Delta(X_{A})$ and has an expected utility representation $\hat U(p)=\bE^p \hat u(x)$. Since $\succsim_z$ is also represented by (\ref{eq_nece_RI2}), we can show that $\hat u$ must be additively separable across different $B_l$. The representation can be rewritten as $\hat U(p)=\sum_{l=1}^m\:\bE^p_{B_l}\: \hat u^{B_l}(x)$. This proves the inductive hypothesis for $A$.

Next, consider the case in which $\bar d(\{i\})\neq\emptyset$. Denote $\hat B_1=B_1\backslash\{i\}\neq \emptyset$. For any fixed $z'\in X_{\hat B_1}$, consider the conditional preference $\succsim_{(z,z')}$ over $\Delta(X_{A\backslash \hat B_1})$. By equation (\ref{eq_nece_RI2}), $\succsim_{(z,z')}$ admits an SMEU representation in which $c(\{i\})=\emptyset$. By the arguments in the previous case, $\succsim_{(z,z')}$ is represented by $\hat U_{z'}(p)= \bE^p_{i}\: \hat u^{i}(z',x) +\sum_{l=2}^m\:\bE^p_{B_l}\: \hat u^{B_l}(z', x)$. Moreover, there exists a continuous and strictly increasing function $f_{z'}$ such that 
\begin{equation}\label{eq_nece_RI4}
   U(z',p)=f_{z'}\left(\bE^p_{i}\: \hat u^{i}(z',x) +\sum_{l=2}^m\:\bE^p_{B_l}\: \hat u^{B_l}(z',x)\right) 
\end{equation}
for all $z'\in X_{\hat B_1}$ and $p\in \Delta(X_{A\backslash \hat B_1})$. 

Since $\succsim_z$ is also represented by (\ref{eq_nece_RI2}), by uniqueness of the expected utility function, $\hat u^{B_l}(z',\:\cdot\:)$ is a positive affine transformation of some $u^{B_l}$ defined on $X_{B_l}$ for all $z'\in X_{\hat B_1}$ and $l\geqslant 2$. Recall that the conditional preference on $B_1$ is represented by
$U^{i}(p)= \bE^p_{i}\:u^{i}(x_{i}, \:(U^B_{x_{i}}(p))_{B\in c(\{i\})}\:)$. 
Up to a normalization, we can rewrite (\ref{eq_nece_RI4})  as 
\begin{equation}\label{eq_nece_RI5}
   U(z',p)=f_{z'}\left(\bE^p_{i}\:u^{i}(x_{i}, \:(U^B_{x_{i}}(z',p))_{B\in c(\{i\})}\:) +\sum_{l=2}^m\:\alpha_l(z')\:\bE^p_{B_l}\: u^{B_l}(x)\right) 
\end{equation}
where $\alpha_l:X_{\hat B_1}\rightarrow (0,\infty)$ for all $l\geqslant 2$ and $U^B_{x_{i}}(z',p)$ does not depend on $p$ for all $B\in c(\{i\})$. Denote $g(x_i,z')=u^{i}(x_{i}, \:(U^B_{x_{i}}(z',p))_{B\in c(\{i\})}\:)$. Then $g: X_{B_1}\rightarrow \bR$ is continuous and strictly increasing.

For each $z'\in X_{\hat B_1}$, we can find an open neighborhood $\cB(z')\subseteq X_{\hat B_1}$  such that $g(X_i,z')\cap g(X_i,\hat z')$ is a nontrivial interval for all $\hat z'\in \cB(z')$. Then for each $\hat z'\in \cB(z')$ and $l\geqslant 2$ there exist $x_i,\hat x_i, x_i', \hat x'_i\in X_i$ and $y,y'\in X_{A\backslash B_1}$ such that (i) $a=g(x_i,z')=g(\hat x_i,\hat z')> a' = g(x'_i,z')=g(\hat x'_i,\hat z')$, (ii) $(x_i,z',y)\sim_z (x'_i,z',y')$, and (iii) $y_{B_{t}}=y'_{B_{t}}$ for all $t\geqslant  2$ and $t\neq l$. By representation (\ref{eq_nece_RI2}), $U(\hat x_i, \hat z', y)=U(x_i,z',y)=  U(x'_i,z',y')=U(\hat x'_i, \hat z', y')$, which implies that $(\hat x_i, \hat z', y)\sim_z (\hat x_i', \hat z', y')$. By representation (\ref{eq_nece_RI5}),  
$a + \alpha_l(z')  u^{B_l}(y_{B_l}) =a' + \alpha_l(z')  u^{B_l}(y'_{B_l})$ and $a + \alpha_l (\hat z')  u^{B_l}(y_{B_l}) =a' + \alpha_l(\hat z')  u^{B_l}(y'_{B_l})$. Since $a>a'$ and $\alpha_l(z'), \alpha_l(\hat z')>0$, we must have $  u^{B_l}(y_{B_l})\neq   u^{B_l}(y'_{B_l})$ and hence $\alpha_l(z') = \alpha_l(\hat z')$. This holds for all $z'\in X_{\hat B_1}$ and $\hat z'$ in its open neighborhood. We conclude that $\alpha_l(z') = \alpha_l(\hat z')$ for all $z',\hat z'\in X_{\hat B_1}$ for all $l$. Denote the constant by $\alpha_l>0$ and (\ref{eq_nece_RI5}) becomes 
\begin{equation}\label{eq_nece_RI6}
   U(z',p)=f_{z'}\left(\bE^p_{i}\:u^{i}(x_{i}, \:(U^B_{x_{i}}(z',p))_{B\in c(\{i\})}\:) +\sum_{l=2}^m\:\alpha_l\:\bE^p_{B_l}\:   u^{B_l}(x)\right).
\end{equation}

We claim that $f_{z'}(a)=f_{\hat z'}(a)$ if $a$ belongs to the intersection of their domain. By \lemmaref{lemma_continuity}, it suffices to consider degenerate $p$. Suppose there exist $x,\hat x\in X_{A\backslash \hat B_1}$ such that 
\begin{align*}
    a & = u^{i}(x_{i}, \:(U^B_{x_{i}}(z',x))_{B\in c(\{i\})}\:) +\sum_{l=2}^m\:\alpha_l\:  u^{B_l}(x_{B_l}) \\
    & = u^{i}(\hat x_{i}, \:(U^B_{x_{i}}(\hat z',\hat x))_{B\in c(\{i\})}\:) +\sum_{l=2}^m\:\alpha_l\: u^{B_l}(\hat x_{B_l}).
\end{align*}
Because the space of conditioning realizations $X_{\hat{B}_1}$ is connected, there exists a continuous path connecting $z'$ to $\hat{z}'$. Since the utility components are continuous and strictly monotonic on compact intervals, any point on this path where the total additive value $a$ is attained admits a local overlapping neighborhood where $a$ remains exactly attainable by offsetting the local change in the conditioning realization with adjustments to the independent branch coordinates. Extracting a finite subcover along this compact path guarantees that we can find two finite sequences  $(y^t)_{t=1}^T$ and $(\hat y^t)_{t=1}^T$ in $X_A$ such that (i) $y^1=(z',x)$ and $\hat y^T = (\hat z',\hat x)$; (ii) $y^t_{\hat B_1}=\hat y^t_{\hat B_1}$ and $u^{i}(\hat y^t_{i}, \:(U^B_{x_{i}}(\hat y^t))_{B\in c(\{i\})}\:) +\sum_{l=2}^m\:\alpha_l\:   u^{B_l}(\hat y^t_{B_l})=a$ for all $t=1,\dots,T$; (iii) $ u^{B_l}(\hat y^t_{B_l}) =   u^{B_l}( y^{t+1}_{B_l})$ and $u^{i}(\hat y^t_{i}, \:(U^B_{x_{i}}(\hat y^t))_{B\in c(\{i\})}\:) =u^{i}( y^{t+1}_{i}, \:(U^B_{x_{i}}( y^{t+1}))_{B\in c(\{i\})}\:)$ for all $t=1,\dots,T-1$ and $l=2,\dots,m$. By property (ii) and representation (\ref{eq_nece_RI6}), $y^t\sim_z \hat y^t$ for all $t=1,..,T$. By property (iii) and representation (\ref{eq_nece_RI2}), $\hat y^t\sim_z y^{t+1}$ for all $t=1,..,T-1$. Property (i) further implies that $(z',x)=y^1\sim_z\cdots\sim_z\hat y^T=(\hat z',\hat x)$. Again by representation (\ref{eq_nece_RI6}), we must have $f_{z'}(a)=f_{\hat z'}(a)$. Hence, it is without loss of generality to assume that $f_{z'}(a)=a$ for all $a$ and  (\ref{eq_nece_RI6}) becomes 
\begin{equation}\label{eq_nece_RI7}
   U(z',p)= \bE^p_{i}\:u^{i}(x_{i}, \:(U^B_{x_{i}}(z',p))_{B\in c(\{i\})}\:) +\sum_{l=2}^m\:\alpha_l\:\bE^p_{B_l}\:   u^{B_l}(x) 
\end{equation}
for all $z'\in X_{\hat B_1}$ and $p\in \Delta(X_{A\backslash \hat B_1})$. It remains to extend the above representation to $\Delta(X_A)$. For each $q\in \Delta(X_A)$, representation (\ref{eq_nece_RI2}) implies that $q\sim_z(q_{B_1},\dots,q_{B_m})$. By \lemmaref{lemma_continuity}, we can find $y\in X_{B_1}$ such that $q\sim_z(y,q_{B_2},\dots,q_{B_m})$ and 
\[u^{i}(y_{i}, \:(U^B_{y_{i}}(y,q_{A\backslash B_1}))_{B\in c(\{i\})}\:) = \bE^q_{i}\:u^{i}(x_i, \:(U^B_{x_{i}}(q))_{B\in c(\{i\})}\:). \]
The utility of $q$ is 
\begin{align*}
    U(q) &=U(y,q_{A\backslash B_1}) = u^{i}(y_{i}, \:(U^B_{y_{i}}(y,q_{A\backslash B_1}))_{B\in c(\{i\})}\:)+\sum_{l=2}^m\:\alpha_l\:\bE^q_{B_l}\:   u^{B_l}(x)  \\
    & = \bE^q_{i}\:u^{i}(x_i, \:(U^B_{x_{i}}(q))_{B\in c(\{i\})}\:) + \sum_{l=2}^m\:\alpha_l\:\bE^q_{B_l}\:  u^{B_l}(x).
\end{align*}
This corresponds to an SMEU representation over $\Delta(X_A)$ with RCT $\cT^*$ in which $c^{\cT^*}(o)=\{\{i\}\}$ and $\bar d^{\cT^*}(\{i\})=A\backslash\{i\}$. Hence, the inductive hypothesis holds. This finishes the proof of \lemmaref{lemma_nece_RI}.
\end{proof}

To show that Axiom \ref{axiom_RI} holds, suppose $i\in A$ and $i\hra A$. Fix any $z\in X_{A^\mathsf{c}}$ and we will omit the dependence of functions on $z$ when there is no risk of confusion. By \lemmaref{lemma_nece_RI} and \lemmaref{lemma_nece_singleton}, the conditional preference $\succsim_z$ over $\Delta(X_A)$ can be represented by $U(q)=  \bE^q_{i}\:u^{i}(x_i, \:(U^B_{x_{i}}(q))_{B\in c(\{i\})})$ for all $q\in \Delta(X_A)$. For any $\alpha\in (0,1)$ and $p,r\in\Delta(X_A)$ such that $\supp(p_i)\cap \supp(r_i)=\emptyset$, we have $(p\alpha r)_{A\backslash\{i\}|x_i} = p_{A\backslash\{i\}|x_i}$  if $x_i\in \supp(p_i)$ and $(p\alpha r)_{A\backslash\{i\}|x_i} = r_{A\backslash\{i\}|x_i}$ if $x_i\in \supp(r_i)$. This implies  $U^B_{x_{i}}(p\alpha r)= U^B_{x_{i}}(p)$ if $x_i\in \supp(p_i)$ and $U^B_{x_{i}}(p\alpha r)= U^B_{x_{i}}(r)$ if $x_i\in \supp(r_i)$ for all $B\in c(\{i\})$. Then Axiom \ref{axiom_RI} holds because 
\begin{align*}
    U (p\alpha r) =  &~\alpha\sum_{x_i\in \supp(p_i)} p_i(x)\:u^{i}(x_i, \:(U^B_{x_{i}}(p\alpha r))_{B\in c(\{i\})}) \\
     & + (1-\alpha)\sum_{x_i\in \supp(r_i)} r_i(x)\:u^{i}(x_i, \:(U^B_{x_{i}}(p\alpha r))_{B\in c(\{i\})})\\
     = &~\alpha\sum_{x_i\in \supp(p_i)} p_i(x)\:u^{i}(x_i, \:(U^B_{x_{i}}(p))_{B\in c(\{i\})}) \\
     & + (1-\alpha)\sum_{x_i\in \supp(r_i)} r_i(x)\:u^{i}(x_i, \:(U^B_{x_{i}}(r))_{B\in c(\{i\})})\\
     = &~ \alpha U(p) + (1-\alpha) U(r).
\end{align*}

 The next two results give  sufficient conditions for $i\hra j$ and $i\perp j$.  

\begin{restatedlemma}{lemma_nece_perp}
    Suppose that $\succsim$ has an SMEU representation.  If  $j\not\in H(i)\cup\bar d(H(i))$ and $i\not\in H(j)\cup\bar d(H(j))$ in some $\cT\in \mathbb{T}(\succsim)$, then $i\perp j$.
\end{restatedlemma}

\begin{proof}[Proof of \lemmaref{lemma_nece_perp}]
    By the proof of \lemmaref{lemma_nece_singleton}, it is without loss to assume that $\cP=\{\{1\},\dots,\{N\}\}$, $j\not\in \bar d(\{i\})$, and $i\not\in \bar d(\{j\})$. Then we can find $A,B\subseteq I$ such that $i\in A\setminus B, j \in B\setminus A$, $A\cup B=I$, $\bar a(\{i\})\cap \bar a(\{j\})\subseteq A\cap B$, and $r\sim r'$ whenever $r_A=r'_A$ and $r_B=r'_B$. Suppose 
    $p_{\{i,j\}^\mathsf{c}} = q_{\{i,j\}^\mathsf{c}}$ and $p_{i|z}=q_{i|z}$, $p_{j|z}=q_{j|z}$ for all $z\in \supp(p_{\{i,j\}^\mathsf{c}})$. Then $p_{\{i\}^\mathsf{c}}=q_{\{i\}^\mathsf{c}}$ and $p_{\{j\}^\mathsf{c}}=q_{\{j\}^\mathsf{c}}$  and hence $p_A=q_A$ and $p_B=q_B$. We conclude that $p\sim q$.
\end{proof}

\begin{restatedlemma}{lemma_nece_hra}
    Suppose that $\succsim$ has an SMEU representation. If $j\in H(i)\cup\bar d(H(i))$ in some $\cT\in \mathbb{T}(\succsim)$, then $i\hra j$. 
\end{restatedlemma}

\begin{proof}[Proof of \lemmaref{lemma_nece_hra}]
By the proof of \lemmaref{lemma_nece_singleton}, it is without loss to assume that $\cP=\{\{1\},\dots,\{N\}\}$ and $j\in \bar d(\{i\})$. Denote $\bar a(\{i\})=\{l_1,\dots,l_T\}$, where $l_k\in \bar d(\{l_{k-1}\})$ for all $k=2,\dots,T$. Let $l_{T+1}=i$. Denote $S_k=\{l\in I\setminus\{l_k\}: l\not\in \bar d (\{l_k\})\}$ for each $k=1,\dots,T+1$. Intuitively, $S_k$ contains dimensions that are either ancestors of $l_k$ or lie on branches other than the branch containing $l_k$. 

Consider any $x^k,y^k\in X$ and $\pi^k\in (0,1)$ for $k=1,\dots,n$ such that $\sum_{k=1}^n\pi^k=1$, $x^k_i\neq x^{k'}_i, y^k_i\neq y^{k'}_i$ for all $k\neq k'$, and $x^k_{\{i,j\}^\mathsf{c}} = y^k_{\{i,j\}^\mathsf{c}} $ and $ x^k\succsim y^k$ for all $k$. We will show by induction that 
 $\sum_{k=1}^n\pi^k\delta_{x^k}\succsim \sum_{k=1}^n\pi^k\delta_{y^k}$.

First, suppose that $x^k_{S_{T+1}}=x^{k'}_{S_{T+1}}:=z\in X_{S_{T+1}}$ for all $k,k'=1,\dots,n$. Let $\hat z= z_{\bar a(\{l_{T+1}\})}$.
Since $S^\mathsf{c}_{T+1}=\{i\}\cup \bar d(\{i\})$, the rankings between $x^k$ and $y^k$, as well as their mixtures, are fully determined by  
\[U^{l_{T+1}}_{\hat z}(p,\delta_z)=\bE^p_{l_{T+1}|\hat z}\:u^{l_{T+1}}(\:\hat z,\:y,\:(U^B_{(\hat z,y)}(p,\delta_z))_{B\in c(\{l_{T+1}\})}\:)\]
for all $p\in \Delta(X_{\{i\}\cup \bar d(\{i\})})$. Since $i=l_{T+1}$, $x^k_i\neq x^{k'}_i, y^k_i\neq y^{k'}_i$ for all $k\neq k'$, $x^k_{S_{T+1}}=y^k_{S_{T+1}}=z$, 
and $ x^k\succsim y^k$ for all $k$, we have 
\begin{align*}
U^{l_{T+1}}_{\hat z}(\sum_{k=1}^n\pi^k\delta_{x^k}) & =\sum_{k=1}^n\pi^k\: U^{l_{T+1}}_{\hat z}(x^k)\geqslant \sum_{k=1}^n\pi^k\: U^{l_{T+1}}_{\hat z}(y^k)= U^{l_{T+1}}_{\hat z}(\sum_{k=1}^n\pi^k\delta_{y^k}).
\end{align*}
That is, $\sum_{k=1}^n\pi^k\delta_{x^k}\succsim \sum_{k=1}^n\pi^k\delta_{y^k}$. 

Next, suppose by induction that $\sum_{k=1}^n\pi^k\delta_{x^k}\succsim \sum_{k=1}^n\pi^k\delta_{y^k}$ if $x^k_{S_{t}}=x^{k'}_{S_{t}}$ for all $k,k'=1,\dots,n$ and some $2\leqslant t\leqslant 
T+1$. Below, we show that this is also true when $x^k_{S_{t-1}}=x^{k'}_{S_{t-1}}$ for all $k,k'=1,\dots,n$. Denote $z=x^k_{S_{t-1}}\in X_{S_{t-1}}$ and  $\hat z= z_{\bar a(\{l_{t-1}\})}$.  Since $S^\mathsf{c}_{t-1}=\{l_{t-1}\}\cup \bar d(\{l_{t-1}\})$, the rankings between $x^k$ and $y^k$, as well as their mixtures, are fully determined by 
\begin{equation}\label{eq_hra1}
    U^{l_{t-1}}_{\hat z}(p,\delta_z)=\bE^p_{l_{t-1}|\hat z}\:u^{l_{t-1}}(\:\hat z,\:y,\:(U^B_{(\hat z,y)}(p,\delta_z))_{B\in c(\{l_{t-1}\})}\:)
\end{equation}
for all $p\in \Delta(X_{\{l_{t-1}\}\cup \bar d(\{l_{t-1}\})})$. Denote $r= \sum_{k=1}^n\pi^k\delta_{x^k}$, $r'=\sum_{k=1}^n\pi^k\delta_{y^k}$, and $r_{l_{t-1}}=r'_{l_{t-1}}=\sum_{m=1}^M\alpha^m \delta_{w^m}$ for some distinct $w^1,\dots,w^M\in X_{l_{t-1}}$. By relabeling, we can write $r=\sum_{m=1}^M \alpha^m r^m$ and $r'=\sum_{m=1}^M \alpha^m r'^{m}$ with $r^m$ and $r'^m$ defined by 
\[r^m = \frac{1}{\alpha^m}\sum_{k=n_{m-1}+1}^{n_m} \pi^k \delta_{x^k} \text{~~~~and~~~~}r'^m = \frac{1}{\alpha^m}\sum_{k=n_{m-1}+1}^{n_m} \pi^k \delta_{y^k}\]
where $n_0=0<n_1<\dots<n_M=n$ and $r^m_{l_{t-1}}=r'^m_{l_{t-1}}=\delta_{w^m}$ for all $m=1,\dots,M$. By representation (\ref{eq_hra1}) and the inductive hypothesis, we have $r^m\succsim r'^m$ and 
$U^{l_t}_{(\hat z,w^m)}(r)=U^{l_t}_{(\hat z,w^m)}(r^m)\geqslant U^{l_t}_{(\hat z,w^m)}(r'^m)=U^{l_t}_{(\hat z,w^m)}(r')$ for all $m=1,\dots,M$.
 For all $B\in c(\{l_{t-1}\})\setminus \{l_t\}$, the conditional distributions of $r$ and $r'$ on $B\cup\bar d(B)$ agree and hence $U^{B}_{(\hat z,w^m)}(r^m)=U^{B}_{(\hat z,w^m)}(r'^m) $
for all $m=1,\dots,M$. By part (ii) of \lemmaref{lemma_continuity}, 
\[u^{l_{t-1}}(\:\hat z,\:w^m,\:(U^B_{(\hat z,w^m)}(r^m))_{B\in c(\{l_{t-1}\})}\:)\geqslant u^{l_{t-1}}(\:\hat z,\:w^m,\:(U^B_{(\hat z,w^m)}(r'^m))_{B\in c(\{l_{t-1}\})}\:),\]
i.e., $U^{l_{t-1}}_{\hat z}(r^m)\geqslant U^{l_{t-1}}_{\hat z}(r'^m)$ for all $m=1,\dots,M$. This further implies  
\begin{align*}
U^{l_{t-1}}_{\hat z}(r) & =\sum_{m=1}^M\alpha^m\: U^{l_{t-1}}_{\hat z}(r^m)\geqslant \sum_{m=1}^M\alpha^m\: U^{l_{t-1}}_{\hat z}(r'^m)= U^{l_{t-1}}_{\hat z}(r').
\end{align*}
That is, $\sum_{k=1}^n\pi^k\delta_{x^k}\succsim \sum_{k=1}^n\pi^k\delta_{y^k}$. This proves the inductive hypothesis.

Finally, we need to drop the assumption that $x^k_{S_{1}}=x^{k'}_{S_{1}}$ for all $k,k'=1,\dots,n$. The argument is the same as the above induction by noting that $l_0=o$ and $U^o(p)= u^o((U^B(p))_{B\in c(o)})$. This proves $i\hra j$.
\end{proof}

We have proved that $\succsim$ satisfies Axioms \ref{axiom_WO}-\ref{axiom_RI} and Axiom \ref{axiom_C}. By \lemmaref{lemma_transitive}, $\hra$ is transitive.
 We are now ready to show that  Axiom \ref{axiom_H} holds. 
Suppose that $i^k\not\perp i^{k+1}$ for all $k=1,\dots,n-1$ and $i^1\hra i^\ell$, $i^n\hra i^\ell$ for all $\ell=2,\dots,n-1$.  Again by \lemmaref{lemma_nece_singleton}, assume that $\cP=\{\{1\},\dots,\{N\}\}$. We claim that there exists some $k^*$ such that $i^k\in \bar d(\{i^{k^*}\})$ for all $k\neq k^*$. Suppose by contradiction that this is not the case. Then we can find a finite partition $A_1,\dots,A_T$ of $\{1,\dots,n\}$ with $T \geqslant 2$ such that 
$i^k\not\in \bar d(\{i^{k'}\})$ and $i^{k'}\not\in \bar d(\{i^{k}\})$ whenever $i^k$ and $i^{k'}$ belong  to different elements of the partition. 
By  \lemmaref{lemma_nece_perp}, we have $i^k\perp i^{k'}$ for such $i^k$ and $i^{k'}$. 
Without loss of generality, suppose that $1\in A_1$. Since $i^1\not\perp i^2$, we must have $2\in A_1$. Since $i^2\not\perp i^3$, we must have $3\in A_1$. By repeating this argument, we conclude that $A_1=\{1,\dots,n\}$, which contradicts the fact that $T \geqslant 2$. Hence, there exists some $k^*$ such that $i^k\in \bar d(\{i^{k^*}\})$ for all $k\neq k^*$. By \lemmaref{lemma_nece_hra}, $i^{k^*}\hra i^k$ for all $k\neq k^*$. 
If $k^*=1$ or $n$, we are done. If $k^*\in \{2,\dots,n-1\}$, then by transitivity of $\hra$ and $i^1\hra i^{k^*}$, we must have $i^1\hra i^n$. This finishes the proof of the necessity.

\subsubsection*{\RN{4}.2:  Omitted Proofs of  Lemmas in Appendix \ref{appen2}} \label{OA_proof_lemmas}

\begin{restatedlemma}{lemma_conditional}
Suppose that $\succsim$ satisfies Axioms \ref{axiom_WO}-\ref{axiom_one} and  \ref{axiom_C}. For any $i\in I$ and $x\in X_{-i}$, the conditional preference $\succsim_x$ on $\Delta(X_i)$ has an EU representation with a continuous and strictly increasing Bernoulli index $v_{i|x}$, which is unique up to a positive affine transformation.
\end{restatedlemma}

\begin{proof}[Proof of \lemmaref{lemma_conditional}]
By Axioms \ref{axiom_WO}-\ref{axiom_one}, $\succsim_x$ admits an EU representation with a Bernoulli index $v_{i|x}$ defined on $X_i$, which is strictly increasing and unique up to a positive affine transformation. To see that $v_{i|x}$ is continuous, suppose by contradiction that there exists a sequence $(y^n)$ in $X_i$ such that $y^n\rightarrow y\in X_i$ and $v_{i|x}(y^n)\not\rightarrow v_{i|x}(y)$. Without loss of generality and passing to a subsequence if necessary, suppose $v_{i|x}(y^n)\rightarrow a<b= v_{i|x}(y)$ and $v_{i|x}(y^n)<(a+b)/2$ for every $n\geqslant 1$.
Since $\succsim_{i|x}$  admits an EU representation, we can find $r\in \Delta(X_i)$ with 
 $\bE^{r}(v_{i|x})=(a+b)/2$. That is, $y^n\prec_x r\prec_x y$ for every $n\geqslant 1$. Axiom \ref{axiom_C}  implies that $y\precsim_x r \prec_x y$, a contradiction. Hence, $v_{i|x}$ is continuous for each $x\in X_{-i}$.
 \end{proof}

\begin{restatedlemma}{lemma_suff_SUB}
Suppose that $\succsim$ satisfies Axioms \ref{axiom_WO}-\ref{axiom_C}. If $A$ is not bracket separable, then there exists $i\in A$ such that $i\hra A$.
\end{restatedlemma}

\begin{proof}[Proof of \lemmaref{lemma_suff_SUB}]
    Suppose that $A$ is not bracket separable. We claim that $A$ can be written as $A=\bigcup_{k=1}^n\{l_k\}$, in which $l_k\not\perp l_{k+1}$ for all $k=1,\dots,n-1$. Repetition is allowed in $(l_k)_{k=1}^n$. Suppose by contradiction that this is not the case. For any $i\in A$, define 
    \[S(i)=\{j\in A: \exists~m\geqslant 1 \text{~and~} l_1=i,\dots,l_m=j\in A~~\text{s.t.}~~l_k\not\perp l_{k+1}, \forall k=1,\dots,m-1\}.\]
    By definition, $S(i)$ is a connected component of the graph on $A$ whose edges are pairs of dimensions satisfying $\not\perp$. Therefore, for any non-empty proper subset $B^{\prime} \subsetneq S(i)$, there must exist at least one pair $j \in B^{\prime}$ and $k \in S(i) \backslash B^{\prime}$ such that $j \not\perp k$. Consequently, $B^{\prime} \not\perp \left(S(i) \backslash B^{\prime}\right)$, which in turn implies $B^{\prime} \not\perp \left(A \backslash B^{\prime}\right)$. Since this holds for every proper subset of $S(i)$, and because there are no edges between components meaning $S(i) \perp (A \backslash S(i))$, Axiom 4 directly implies that $S(i) \vartriangleright A$. Repeating this logic for every connected component of $A$ yields a partition $\left\{B_{k}\right\}_{k=1}^{n}$ of $A$ in which $B_{k} \vartriangleright A$ for all $k=1, \ldots, n$. Because $S(i) \subsetneq A$, there are at least two such components ($n \geqslant 2$). This constitutes a non-trivial bracket partition of $A$, which contradicts the assumption that $A$ is not bracket separable. 

    By \lemmaref{lemma_transitive}, $\hra$ is transitive. There must exist some $i^*\in A$ such that $j\not\hrra i^*$ for all $j\in A$, since otherwise the finiteness of $A$ implies the existence of a cycle $j_1\hrra j_2\hrra \cdots \hrra j_n\hrra j_1$, which contradicts the transitivity of $\hra$. Below, we show that $i^*\hra A$.

    Suppose by contradiction that $i^*\not\hra j^*$ for some $j^*\in A$. By the construction of $i^*$, we must have $j^*\not\hra i^*$, since otherwise $j^*\hrra i^*$. By Axiom \ref{axiom_H} applied to $n=2$ and $l_1=i^*, l_2=j^*$, we have $i^*\perp j^*$.
    Then we can find $l_1,\dots,l_{n}\in A$ with $n \geqslant 3$ such that $i^*=l_1$, $j^*=l_n$, and $l_k\not\perp l_{k+1}$ for all $k=1,\dots,n-1$. Moreover, it is without loss of generality to assume that (i) $l_k\perp l_{k+t}$ for all $t \geqslant 2$ and $k,k+t\in \{1,\dots,n\}$ and (ii) $i\hra l_k$ for all $k=2,\dots,n-1$. If condition (i) fails, then we can simply choose a shorter sequence by removing $l_{k+1},\dots,l_{k+t-1}$. If condition (ii) fails, then we can replace $j^*$ with $l_k$ and work with a shorter sequence $l_1,\dots,l_{k}$.

    By transitivity of $\hra$, we must have $l_k\not\hra j^*$ for all $k=2,\dots,n-1$. We further claim that $j^*\hra l_k$ for all $k=2,\dots,n-1$. We prove it by induction on $n \geqslant 3$. If $n=3$, we have $l_2\not\perp j^*$ and $l_2\not\hra j^*$. Axiom \ref{axiom_H} implies that $j^*\hra l_2$. Suppose by induction that the result holds for all $n \leqslant m$ with $m \geqslant 3$. Now consider the case with $n=m+1$. Since $l_m\not\perp j^*$ and $l_m\not\hra j^*$, Axiom \ref{axiom_H} implies that $j^*\hra l_m$. Suppose by contradiction that $j^*\not\hra l_{m-1}$. By transitivity, $l_m\not\hra l_{m-1}$. By applying the inductive hypothesis to $(j^*,l_m, l_{m-1})$, we derive that $l_{m-1}\hra l_m$. Since $l_{m-1}\not\perp l_m \not\perp j^*$ and $j^*\hra l_m$, Axiom \ref{axiom_H}  implies that either $j^*\hra l_{m-1}$ or $l_{m-1}\hra j^*$, with both   leading to a contradiction. Hence we have  $j^*\hra l_{m-1}$.
    By repeating this argument and applying the inductive hypothesis to  $(j^*,l_m,\dots, l_{k})$ for $2\leqslant k\leqslant m-1$, we can show that $j^*\hra l_{k}$ for all $2\leqslant k\leqslant m$. This proves the inductive hypothesis for $n=m+1$.

    Now we have $i^*=l_1\hra l_k$ and $j^*=l_n\hra l_k$ for all $k=2,\dots,n-1$ and $l_k\not\perp l_{k+1}$ for all $k=1,\dots,n-1$. Again by Axiom \ref{axiom_H}, either $i^*\hra j^*$ or $j^*\hra i^*$, which contradicts $i^*\not\hra j^*$ and $j^*\not\hra i^*$. Hence, we conclude that $i^*\hra A$. This proves \lemmaref{lemma_suff_SUB}.
\end{proof}

\begin{restatedlemma}{lemma_monotone}
 Suppose that $\succsim$ satisfies Axioms \ref{axiom_WO}-\ref{axiom_C}.   \emph{\text{(i)}} For any $A\subseteq I$, $p\in \Delta(X_A)$,  $x\in X_A$, and $x'\in X_{A^{\cp}}$, if $p$ dominates $x$, then $p\succ_{x'} x$, and if $x$ dominates $p$, then $x\succ_{x'} p$. \emph{\text{(ii)}}  For any $A\subseteq I$,  $p\in \Delta(X_A)$,  $x'\in X_{A^{\cp}}$, and $x,y\in X_A$ such that $p$ dominates $y$ and is dominated by $x$, there exists some $z\in X_A$ such that $p\sim_{x'} z$ and $x\geqslant  z\geqslant  y$.
\end{restatedlemma}

\begin{proof}[Proof of \lemmaref{lemma_monotone}]
Without loss of generality, we can focus on the case in which $A=I$ and prove the results using induction on the cardinality of $I$. If $|I|=1$, then both statements hold trivially by \lemmaref{lemma_conditional}. Now suppose that both statements hold for $|I|\leqslant t$ for some $t\geqslant  1$. We need to show that they hold for $|I|=t+1$. 

If $I$ is bracket separable, then there exists a nontrivial partition $\{A_k\}_{k=1}^n$ of $I$ such that $A_k\vartriangleright I$ for all $k=1,\dots,n$. Since the partition is nontrivial, $|A_k|\leqslant t$ for all $k=1,\dots,n$. Fix any $p\in \Delta(X)$. 
 Since $A_1\vartriangleright I$, we have $p\sim(p_{A_1},p_{A_1^{\cp}})$. Since $A_2\vartriangleright I$, we have $p\sim(p_{A_1},p_{A_1},p_{(A_1\cup A_2)^{\cp}})$, and so on. Iteratively, we can show that $p\sim (p_{A_1},\dots,p_{A_n}):=q$. 

If $p$ dominates $x$, then $p_{A_k}=q_{A_k}$ weakly dominates $x_{A_k}$ for every $k$ and the dominance is strict for some $k^*$. By the inductive hypothesis and the definition of $A_k\vartriangleright I$,  we have  $p\sim q\succsim (x_{A_1},\dots,x_{A_n})=x$. Since at least one dominance relation is strict, we must have $p\succ x$. The proof for the case in which $x$ dominates $p$ is symmetric and omitted. This proves the first statement. The second statement can be proved similarly by applying the inductive hypothesis to  each $A_1,\dots,A_n$ iteratively.


Next, suppose that $I$ is not bracket separable. By \lemmaref{lemma_suff_SUB},   there exists $i\in I$ with $i\hra I$. Recall that we are using induction on the cardinality of $I$. Now for each cardinality of $I$, we will use another inductive argument based on the cardinality of $\supp(p_i)$. If $|\supp(p_i)|=1$, then $p=(x,p_{-i})$ for some $x\in X_i$. By Axiom \ref{axiom_M}  and applying the inductive hypothesis to the conditional preference $\succsim_x$, the two statements hold trivially.

Assume that for some $n\geqslant  1$, the two statements hold if $|\supp(p_i)|\leqslant n$. Suppose $|\supp(p_i)|= n+1$ and $p$ dominates $x$. Then, we can choose some $a_i\in \supp(p_i)\setminus\{x_i\}$ with $a_i>x_i$ and write $p=p_i(a_i)(\delta_{a_i}, p_{-i|a_i}) + (1-p_i(a_i))p'$, in which $|\supp(p'_i)|=n$. Note that $(\delta_{a_i}, p_{-i|a_i})$ dominates $x$ and $p'$ weakly dominates $x$, which implies  $p'\succsim x$ and $(\delta_{a_i}, p_{-i|a_i})\succ x$ by the inductive hypothesis. Clearly, $\supp(\delta_{a_i})\cap \supp(\delta_{x_i})=\emptyset$ and $\supp(\delta_{a_i})\cap \supp(p'_i)=\emptyset$. Since $a_i>x_i$, by Axiom \ref{axiom_C}  and Axiom \ref{axiom_M}, there exists some $x_i'>x_i$ and $x'=(x_i',x_{-i})$ such that $(\delta_{a_i}, p_{-i|a_i})\succ x'\succ x$. Since $i\hra I$, Axiom \ref{axiom_RI} implies that $p=p_i(a_i)(\delta_{a_i}, p_{-i|a_i}) + (1-p_i(a_i))p'\succsim  p_i(a_i)(\delta_{a_i}, p_{-i|a_i}) + (1-p_i(a_i))\delta_x \succ p_i(a_i)\delta_{x'} + (1-p_i(a_i))\delta_x$. The first relation is $\sim$ if and only if $p'=x$. Since $x'$ and $x$ agree in all dimensions other than $i$, \lemmaref{lemma_conditional} implies $p_i(a_i)\delta_{x'} + (1-p_i(a_i))\delta_x\succ x$ and hence $p\succ x$. The case in which $x$ dominates $p$ is symmetric. This proves the first statement.

For the second statement, suppose that $p$ dominates $y$ and is dominated by $x$. By the first statement and  Axiom \ref{axiom_M}, $\overline{x}\succ p\succ \lx$. Denote $x^0:=\overline{x}, x^1=(\lx_1,\ux_{-1}), x^2=(\lx_{\{1,2\}},\ux_{\{1,2\}^{\mathsf{c}}}),\dots$, and $x^{t+1}=\lx$. Then $x^0\succ x^1\succ\dots\succ x^{t+1}$. We can find some $k\in \{0,\dots,t\}$ such that $x^k\succsim p\succsim x^{k+1}$. By Axiom \ref{axiom_C}, we can find $\alpha \in [0,1]$ such that $p\sim \delta_{x^k}\:\alpha\: \delta_{x^{k+1}}$, which, by \lemmaref{lemma_conditional}, is indifferent to some $z\in X$. Since $x\succ p\sim z \succ y$, again by Axiom \ref{axiom_C}  and Axiom \ref{axiom_M}, there exists $z'\sim z$ with $x\geqslant  z'\geqslant  y$.  
By induction, we conclude that both statements hold for any finite cardinality of $\supp(p)$ and $I$.
 \end{proof}


\begin{restatedlemma}{lemma_finest}
 Suppose that $\succsim$ satisfies Axioms \ref{axiom_WO}-\ref{axiom_C}. If $A$ is bracket separable, then $A$ must have a bracket partition $\{A_k\}_{k=1}^n$ in which $A_k$ is not bracket separable for all $k=1,\dots,n$. Moreover, such $\{A_k\}_{k=1}^n$ is unique and is finer than any other bracket partition of $A$.
\end{restatedlemma}

\begin{proof}[Proof of \lemmaref{lemma_finest}]

Suppose that $A$ has a bracket partition $\{A_k\}_{k=1}^n$ in which $A_1$ is bracket separable with bracket partition $\{B_k\}_{k=1}^m$. We claim that $\{A_k\}_{k=2}^n\cup \{B_k\}_{k=1}^m$ is a bracket partition of $A$. To see this, note that for all $p\in\Delta(X_{A})$ and $x\in X_{A^\mathsf{c}}$, we have $p\sim_x(p_{A_1},\dots,p_{A_n})$. By \lemmaref{lemma_monotone}, there exists $x_{A_k}\in X_{A_k}$ for every $k\geqslant  2$ such that $(p_{A_1},\dots,p_{A_n})\sim_x (p_{A_1}, x_{A_2},\dots,x_{A_n})$. Then, given $(x,(x_{A_k})_{k=2}^n) \in X_{A_1^\mathsf{c}}$, since $\{B_k\}_{k=1}^m$ is a bracket partition of $A_1$, 
 \begin{align*}
     p &\sim_x (p_{A_1}, x_{A_2},\dots,x_{A_n}) \\
     &\sim_x  (p_{B_1},\dots,p_{B_m}, x_{A_2},\dots,x_{A_n}) \\
     & \sim_x (p_{B_1},\dots,p_{B_m}, p_{A_2},\dots,p_{A_n}).
 \end{align*}
  The last indifference holds since $\{A_k\}_{k=1}^n$ is a bracket partition of $A$. 

  Fix $k\in\{1,\dots,m\}$. We want to show that $B_k\vartriangleright A$.  For any $x\in X_{A^{\cp}}$, $r\in \Delta(X_{A\setminus B_k})$, and $p,q\in \Delta(X_A)$ such that $p_{A\setminus B_k}=q_{A\setminus B_k}$, we have 
\begin{align*}
    & ~p\succsim_x q  \\
    \xLongleftrightarrow{}& ~(p_{B_k}, p_{A_1\setminus B_k}, p_{A\setminus A_1}) \succsim_x  (q_{B_k}, p_{A_1\setminus B_k}, p_{A\setminus A_1}) \\
    \xLongleftrightarrow{\hbox{$A_1 \vartriangleright A$} } &~(p_{B_k}, p_{A_1\setminus B_k}, x_{A\setminus A_1}) \succsim_x  (q_{B_k}, p_{A_1\setminus B_k}, x_{A\setminus A_1})\\
     \xLongleftrightarrow{\hbox{by definition} }&~ (p_{B_k}, p_{A_1\setminus B_k}) \succsim_{(x,x_{A\setminus A_1})}  (q_{B_k}, p_{A_1\setminus B_k})\\
      \xLongleftrightarrow{\hbox{$B_k \vartriangleright A_1$} }&~ (p_{B_k}, r_{A_1\setminus B_k}) \succsim_{(x,x_{A\setminus A_1})}   (q_{B_k}, r_{A_1\setminus B_k})\\
      \xLongleftrightarrow{\hbox{by definition} }&~ (p_{B_k}, r_{A_1\setminus B_k}, x_{A\setminus A_1}) \succsim_x  (q_{B_k}, r_{A_1\setminus B_k}, x_{A\setminus A_1}) \\
      \xLongleftrightarrow{\hbox{$A_1 \vartriangleright A$} } &~(p_{B_k}, r_{A_1\setminus B_k}, r_{A\setminus A_1}) \succsim_x  (q_{B_k}, r_{A_1\setminus B_k}, r_{A\setminus A_1})\\
       \xLongleftrightarrow{\hbox{$A_1 \vartriangleright A$} }&~ (p_{B_k},r)\succsim_x   (q_{B_k},r).
\end{align*}
Hence, $B_k\vartriangleright A$ for every $k$ and  $\{A_k\}_{k=2}^n\cup \{B_k\}_{k=1}^m$ is a bracket partition of $A$. Continue this process until all elements in the bracket partition of $A_1$ are not bracket separable, and then repeat this procedure for other $A_k$'s. After finitely many steps, we will end up with a bracket partition of $A$ in which none of its elements is bracket separable. We still denote it by $\{A_k\}_{k=1}^n$.
To prove that $\{A_k\}_{k=1}^n$ is the finest bracket partition of $A$, we need the following lemma, which is an adaptation of the results in \citeAppendix{OA-Gorman1968}.

\begin{lemma}\label{lemma_isolation}
If $B\vartriangleright A$, $B'\vartriangleright A$,  and $B\cap B', B\setminus B', B'\setminus B$ are nonempty, then $B\cap B'\vartriangleright A, B\setminus B'\vartriangleright A$, and $B'\setminus B\vartriangleright A$.
\end{lemma}

\begin{proof}[Proof of \lemmaref{lemma_isolation}]
Given the assumptions, we first verify that $B\cap B'\vartriangleright A$.  Take any $x\in X_{A^{\cp}}$, $r\in \Delta(X_{A\backslash (B\cap B')})$, and $p,q\in \Delta(X_A)$ such that $p_{A\setminus (B\cap B')}=q_{A\setminus (B\cap B')}$. Then $p_{A\setminus B}=q_{A\setminus B}$ and $p_{A\setminus B'}=q_{A\setminus B'}$. Since $B\vartriangleright A$ and  $B'\vartriangleright A$, we have
\begin{align*}
     p \succsim_x q  
    \iff & (p_{B}, p_{A\setminus B})\succsim_x (q_{B}, q_{A\setminus B}) \\
   \iff & (p_{B\cap B'}, p_{A\setminus (B\cup B')}, p_{B\setminus B'}, p_{B'\setminus B})\succsim_x (q_{B\cap B'}, q_{A\setminus (B\cup B')}, q_{B\setminus B'}, q_{B'\setminus B}) \\
   \iff  & (p_{B\cap B'}, r_{A\setminus (B\cup B')}, p_{B\setminus B'}, r_{B'\setminus B})\succsim_x (q_{B\cap B'}, r_{A\setminus (B\cup B')}, q_{B\setminus B'}, r_{B'\setminus B}) \\
   \iff & (p_{B\cap B'}, r_{A\setminus (B\cup B')}, r_{B\setminus B'}, r_{B'\setminus B})\succsim_x (q_{B\cap B'}, r_{A\setminus (B\cup B')}, r_{B\setminus B'}, r_{B'\setminus B}) \\
   \iff & (p_{B\cap B'}, r)\succsim_x (q_{B\cap B'}, r).
\end{align*}
Then we verify $B\setminus B'\vartriangleright A$. The case for $B'\setminus B$ is symmetric and hence omitted. Take any $x\in X_{A^{\cp}}$, $r\in \Delta(X_{A\setminus(B\setminus  B')})$, and $p,q\in \Delta(X_A)$ such that $p_{A\setminus (B\setminus  B')}=q_{A\setminus (B\setminus  B')}$. Since $B\vartriangleright A, B'\vartriangleright A$, and $B\cap B'\vartriangleright A$, we can use $\succsim_x$ to represent the conditional preference on $\Delta(X_B), \Delta(X_{B'})$ and $\Delta(X_{B\cap B'})$ when there is no risk of confusion. 
Using the previous argument again, $p \succsim_x q $ if and only if
\begin{align*}
 (p_{B\cap B'}, r_{A\setminus (B\cup B')}, p_{B\setminus B'}, \lx_{B'\setminus B})\succsim_x (q_{B\cap B'}, r_{A\setminus (B\cup B')}, q_{B\setminus B'}, \lx_{B'\setminus B}).
\end{align*}

We want to replace $p_{B\cap B'}=q_{B\cap B'}$ with $r_{B\cap B'}$ without changing the preference. This can be done immediately if $p_{B\cap B'}\sim_x r_{B\cap B'}$, because $B\cap B'\vartriangleright A$. 
Without loss of generality, assume that $p_{B\cap B'}\succ_x r_{B\cap B'}$.
First, by \lemmaref{lemma_monotone}, there exist $y,y'\in X_{B\cap B'}$ such that $y\geqslant  y'$ and $y\sim_x p_{B\cap B'}\succ_x r_{B\cap B'}\sim_x y'$. Since $B\cap B'\vartriangleright A$,  we have
\begin{align*}
p \succsim_x q\iff (y, r_{A\setminus (B\cup B')}, p_{B\setminus B'}, \lx_{B'\setminus B})\succsim_x (y, r_{A\setminus (B\cup B')}, q_{B\setminus B'}, \lx_{B'\setminus B}).
\end{align*}
Second, if $(y', \ux_{B'\setminus B})\succsim_x (y,\lx_{B'\setminus B})$, then  Axiom \ref{axiom_C}  and Axiom \ref{axiom_M}  imply that there exists $z\in X_{B'\setminus B}$ such that $(y',z)\sim_x (y,\lx_{B'\setminus B})$. Since $B\vartriangleright A$, $B'\vartriangleright A$, and $B\cap B'\vartriangleright A$, 
\begin{align*}
    p \succsim_x q 
   \iff & (y', r_{A\setminus (B\cup B')}, p_{B\setminus B'}, z)\succsim_x (y', r_{A\setminus (B\cup B')}, q_{B\setminus B'}, z) \\
   \iff  & (r_{B\cap B'}, r_{A\setminus (B\cup B')}, p_{B\setminus B'},  z)\succsim_x (r_{B\cap B'}, r_{A\setminus (B\cup B')}, q_{B\setminus B'}, z) \\
   \iff & (r_{B\cap B'}, r_{A\setminus (B\cup B')}, p_{B\setminus B'}, r_{B'\setminus B})\succsim_x (r_{B\cap B'}, r_{A\setminus (B\cup B')}, q_{B\setminus B'}, r_{B'\setminus B}) \\
   \iff & (p_{B\setminus B'}, r)\succsim_x (q_{B\setminus B'}, r).
\end{align*}

Third, suppose instead that $(y,\lx_{B'\setminus B})\succ_x (y', \ux_{B'\setminus B})$. 
Then there exist $y'',y'''\in X_{B\cap B'} \setminus \{\ux_{B\cap B'},\lx_{B\cap B'}\}$ such that $y''>y'''$, $(y'',\ux_{B'\setminus B})\succ_x (y, \lx_{B'\setminus B})$, and  $(y''',\lx_{B'\setminus B})\prec_x (y', \ux_{B'\setminus B})$. By strict consequence monotonicity, $(z, \bar{x}_{B^{\prime} \backslash B}) \succ_{x}(z, \underline{x}_{B^{\prime} \backslash B})$ for every $z \in X_{B \cap B^{\prime}}$. Because the set $\{z \in X_{B \cap B^{\prime}}: y^{\prime \prime \prime} \leqslant z \leqslant y^{\prime \prime}\}$ is compact and conditional preferences are continuous, the utility difference between these two profiles is uniformly bounded below by a strictly positive amount.  tarting with $y^{0}=y$, Axiom 7 and strict monotonicity ensure that we can iteratively find $y^{k+1}$ such that $(y^{k+1}, \bar{x}_{B^{\prime} \backslash B}) \sim_{x}(y^{k}, \underline{x}_{B^{\prime} \backslash B})$. Since each iteration decreases the underlying utility in the $B \cap B^{\prime}$ coordinates by at least this strictly positive lower bound, the process must surpass $y^{\prime}$ in finitely many steps. This ensures the existence of a finite sequence $ (y^{k})_{k=0}^{n}$ in $\{z\in  X_{B\cap B'}: y'''\leqslant z \leqslant y''\}\cup\{y\}$ such that (i) $y^0=y$, (ii) $(y^{k+1},\ux_{B'\setminus B})\sim_x (y^k,\lx_{B'\setminus B})$ for each $k=0,\dots,n-1$, and (iii)  $(y', \ux_{B'\setminus B})\succsim_x (y^n,\lx_{B'\setminus B})$. By applying (ii) and the implications of $B\vartriangleright A$ and $B'\vartriangleright A$ repeatedly, we obtain  
\begin{align*}
    p \succsim_x q 
   \iff & (y, r_{A\setminus (B\cup B')}, p_{B\setminus B'}, \lx_{B'\setminus B})\succsim_x (y, r_{A\setminus (B\cup B')}, q_{B\setminus B'}, \lx_{B'\setminus B}) \\
   \iff  &(y^1, r_{A\setminus (B\cup B')}, p_{B\setminus B'}, \ux_{B'\setminus B})\succsim_x (y^1, r_{A\setminus (B\cup B')}, q_{B\setminus B'}, \ux_{B'\setminus B}) \\
   \iff & (y^1, r_{A\setminus (B\cup B')}, p_{B\setminus B'}, \lx_{B'\setminus B})\succsim_x (y^1, r_{A\setminus (B\cup B')}, q_{B\setminus B'}, \lx_{B'\setminus B})\\
   \iff  &(y^2, r_{A\setminus (B\cup B')}, p_{B\setminus B'}, \ux_{B'\setminus B})\succsim_x (y^2, r_{A\setminus (B\cup B')}, q_{B\setminus B'}, \ux_{B'\setminus B}) \\
   \iff & (y^2, r_{A\setminus (B\cup B')}, p_{B\setminus B'}, \lx_{B'\setminus B})\succsim_x (y^2, r_{A\setminus (B\cup B')}, q_{B\setminus B'}, \lx_{B'\setminus B})\\
   \iff & \cdots 
   \iff   (y^n, r_{A\setminus (B\cup B')}, p_{B\setminus B'}, \lx_{B'\setminus B})\succsim_x (y^n, r_{A\setminus (B\cup B')}, q_{B\setminus B'}, \lx_{B'\setminus B}).
\end{align*}
Since we have $(y', \ux_{B'\setminus B})\succsim_x (y^n,\lx_{B'\setminus B})$ by (iii), we can apply the argument in the previous case to establish that $ p \succsim_x q \iff  (p_{B\setminus B'}, r)\succsim_x (q_{B\setminus B'}, r).$
This completes the proof for $B\setminus B'\vartriangleright A$.\end{proof}


Now consider a different bracket partition $\{B_l\}_{l=1}^m$ of $A$. If $\{A_k\}_{k=1}^n$ is not finer than $\{B_l\}_{l=1}^m$, then there exist $A_k$ and $B_{l_1},\dots,B_{l_t}$ with $t\geqslant  2$ such that $A_k\cap B_{l_i}\neq \emptyset$ for all $i=1,\dots,t$ and $A_k\subseteq \bigcup_{i=1}^t B_{l_i}$. By \lemmaref{lemma_isolation}, $(A_k\cap B_{l_i})\vartriangleright A$ for all $i=1,\dots,t$, which implies that $A_k$ has a bracket partition $\{A_k\cap B_{l_i}\}_{i=1}^t$, a contradiction. Hence, $\{A_k\}_{k=1}^n$ is finer than any other bracket partition $\{B_l\}_{l=1}^m$. Moreover, there must exist some bracket-separable $B_l$ with a bracket partition being a subset of $\{A_k\}_{k=1}^n$. This ensures the uniqueness of the finest bracket partition.
\end{proof}

\begin{restatedlemma}{lemma_linear}
 Suppose that $\succsim$ satisfies Axioms \ref{axiom_WO}-\ref{axiom_C} and $i\in B$. If $i\hra B$, then for every $x\in X_{B^\mathsf{c}}$, there exists a function $U_x: \Delta(X_B)\rightarrow \bR$ such that \emph{\text{(i)}}  $p\succsim_x q$ if and only if $U_x(p)\geqslant  U_x(q)$ for all $p,q\in \Delta(X_B)$; \emph{\text{(ii)}}  $U_x(p\alpha q ) = \alpha U_x(p) + (1-\alpha)U_x(q)$ for all $\alpha\in (0,1)$ and $p,q\in \Delta(X_B)$ with $\supp(p_i)\cap \supp(q_i)=\emptyset$; \emph{\text{(iii)}}  the function $w_x: X_B\rightarrow \bR$ defined by $w_x(y)=U_x(\delta_y)$ for every $y\in X_B$ is continuous and strictly increasing; and \emph{\text{(iv)}}  $U_x$ is unique up to a positive affine transformation.
\end{restatedlemma}

\begin{proof}[Proof of \lemmaref{lemma_linear}]
We first prove some intermediate results.

\begin{lemma}\label{lemma_independence}
   Suppose that $\succsim$ satisfies Axioms \ref{axiom_WO}-\ref{axiom_C} and $i\in B$. If $i\hra B$, then for all $\alpha\in(0,1)$, $x\in X_{B^\mathsf{c}}$, and $p,q,r,s\in\Delta(X_{B})$ such that $\supp(p_i)\cap \supp(r_i)=\emptyset$ and $\supp(q_i)\cap \supp(s_i)=\emptyset$, the following properties hold:
\emph{\text{(i)}}  $p\succ_x r \implies p\succ_x p\alpha r\succ_x r$; 
\emph{\text{(ii)}}  $p\sim_x r \implies p\sim_x p\alpha r\sim_x r$; 
\emph{\text{(iii)}}  $p\sim_x q, r\sim_x s \implies p\alpha r\sim_x q\alpha s$; and
\emph{\text{(iv)}}  $p\succ_x q, r\succ_x s \implies p\alpha r\succ_x q\alpha s$.
\end{lemma}

\begin{proof}[Proof of \lemmaref{lemma_independence}]
    For (i), we consider four cases. First, if $p=\ux_B$ and $r=\lx_B$, then the result is implied by \lemmaref{lemma_monotone}. Second, if $p=\ux_B$ and $r\succ_x\lx_B$, then by Axiom \ref{axiom_C}  and Axiom \ref{axiom_M}, we can find $\varepsilon\in \bR^{B}_+$ such that $\varepsilon_i>0$, $\varepsilon_j=0$ for all $j\in B\setminus\{i\}$, and $\ux_B\succ_x\ux_B-\varepsilon\succ_x r$. By the definition of $i\hra B$ and \lemmaref{lemma_conditional}, we have $p\alpha r\precsim_x \delta_{\ux_B} \:\alpha \: \delta_{\ux_B-\varepsilon}\prec_x \ux_B= p$. As $r\succ_x \lx_B$, we can find $y,y'\in X_{B}$ such that  $p\succ_x y\succ_x r\sim_x y'$, $y_i,y_i'\not\in \supp(p_i)\cup\supp(r_i)$, $y_i\neq y_i'$, and $y_j=y'_j$ for all $j\in B\setminus\{i\}$. Again by  the definition of $i\hra B$ and \lemmaref{lemma_conditional}, $p\alpha r \succsim_x  \delta_y \alpha  \delta_{y'} \succ_x y' \sim_x r$. Third, if $p\prec_x \ux_B$ and $r=\lx_B$, then the proof is symmetric to the second case. Finally, if $\ux_B \succ_x p\succ_x r\succ_x \lx_B$, then the proof is a simple combination of those of the above two cases. 

For (ii), if $p\sim_x r$ and $\supp(p_i)\cap \supp(r_i)=\emptyset$, then $\ux_B\succ_x p\sim_x r\succ_x \lx_B$. By Axiom \ref{axiom_C}  and Axiom \ref{axiom_M}, we can find $y,y'\in X_{B}$ such that $y\succ_x p\sim_x r\succ_x y'$, $y_i\neq y'_i$, and $y_i,y'_i\not\in \supp(p_i)\cup\supp(r_i)$. For any $\beta\in (0,1)$, by applying (i) twice we get $p\beta \delta_y\succ_x p \sim_x r \succ_x p\beta \delta_{y'}$. Applying (i) again to    $p\beta \delta_y, r$, and $p\beta \delta_{y'}$ yields $(p\beta \delta_y)\alpha r\succ_x r \succ_x (p\beta \delta_{y'})\alpha r$. Let $\beta$ go to $1$, and by  Axiom \ref{axiom_C}, $p\alpha r\succsim_x r \succsim_x p\alpha r$, which implies $p\alpha r\sim_x r\sim_x p$.

For (iii), if $p,r\in \{\ux_B,\lx_B\}$, then $p=q$ and $r=s$ and the result is trivial. Without loss of generality, assume that $\ux_B \succ_x p\sim_x q \succ_x \lx_B$. Using the argument in the proof for (ii), we can find $y,y'$ such that $y\succ_x p\sim_x q\succ_x y'$, $y_i\neq y'_i$, $y_i,y'_i\not\in \supp(p_i)\cup\supp(r_i)$, and $p\beta \delta_y\succ_x p\sim_x q \succ_x p\beta \delta_{y'}$ for every $\beta\in(0,1)$. Since $i\hra B$, Axiom \ref{axiom_RI}  implies  $(p\beta \delta_y)\alpha r\succ_x q\alpha s \succ_x (p\beta \delta_{y'})\alpha r$. Let $\beta$ go to $1$, and by  Axiom \ref{axiom_C}, we have $p\alpha r\succsim_x q\alpha s \succsim_x p\alpha r$, which implies that $p\alpha r\sim_x q\alpha s$.

For (iv), it suffices to consider the case in which  $p\succ_x q\succ_x r\succ_x s$, as the other cases are either symmetric or implied by (i). There exists $y\in X_{B}$ such that $r\sim_x y$ and $y_i\not\in \supp(q_i)$. Since $i\hra B$, applying Axiom \ref{axiom_RI} twice shows $p\alpha r \succ_x q\alpha \delta_y \succ_x q \alpha s$. 
\end{proof}

The next lemma shows when independence holds  if supports of lotteries overlap.

\begin{lemma}\label{lemma_independence_extreme}
 Suppose that $\succsim$ satisfies Axioms \ref{axiom_WO}-\ref{axiom_C} and $i\in B$.  If $i\hra B$, then for all $\alpha\in(0,1)$, $x\in X_{B^\mathsf{c}}$, and $p,q,r,s\in\Delta(X_{B})$ such that $\supp(p_i)\cap \supp(r_i)=\emptyset$ and $\supp(q)\cup \supp(s)\subseteq \{\ux_{B},\lx_{B}\}$, then  $p\sim_x q, r\sim_x s \implies p\alpha r\sim_x q\alpha s$.
\end{lemma}

\begin{proof}[Proof of \lemmaref{lemma_independence_extreme}]
First, if $p,r\in \{\ux_B,\lx_B\}$, then $p=q$, $r=s$ and the result is trivial.  Without loss of generality, assume that $\ux_B \succ_x p\sim_x q \succ_x \lx_B$. Then there exists $y\in X_{B}$ such that $y_i\not\in \{\ux_i,\lx_i\}$ and $p\sim_x y$. Since $p\sim_x y$, $r\sim_x s$, $\supp(p_i)\cap \supp(r_i)=\emptyset$, and $y_i\not\in \supp(s_i)$, by part (iii) of \lemmaref{lemma_independence}, we have $p\alpha r \sim_x \delta_y\alpha s$. Hence, it suffices to show that $\delta_y\alpha s\sim_x q\alpha s$ for every $y\sim q$ with $y_i\not\in \{\ux_i,\lx_i\}$. As $\ux_B \succ_x p\sim_x q \succ_x \lx_B$, by Axiom \ref{axiom_C}, we can find $\varepsilon\in \bR^{B}_+$ and $\gamma\in (0,1)$ such that $\varepsilon_i>0$, $\varepsilon_j=0$ for all $j\in B\setminus\{i\}$, $\ux_B-\varepsilon \succ_x y \succ_x \lx_B +\varepsilon$, and $y\sim_x q \sim_x \delta_{\ux_B-\varepsilon}\: \gamma\: \delta_{\lx_B+\varepsilon}$. Denote $\hat{q}=\delta_{\ux_B-\varepsilon}\: \gamma \: \delta_{\lx_B+\varepsilon}$ and $q^\beta =q\beta \hat{q}$ for each $\beta\in (0,1)$.  Part (ii) of \lemmaref{lemma_independence} implies that $q^\beta \sim_x q \sim_x y$.

We claim that $\delta_y\alpha s \sim_x q^{\beta}\alpha s$ for all $\beta,\alpha\in (0,1)$. 
To see this, first note that $q(\ux_B)>0$ and $ q(\lx_B)>0$ as $\ux_B \succ_x q\succ_x \lx_B$. Then 
\begin{align*}
    q^\beta & = [\beta q(\ux_B)\delta_{\ux_B} + (1-\beta)\gamma \delta_{\ux_B-\varepsilon}] + [\beta q(\lx_B)\delta_{\lx_B} + (1-\beta)(1-\gamma) \delta_{\lx_B+\varepsilon}]\\ 
    & \sim_x [\beta q(\ux_B) + (1-\beta)\gamma] \delta_{\ux_B-\varepsilon'} + [\beta q(\lx_B) + (1-\beta)(1-\gamma)] \delta_{\lx_B+\varepsilon''},
\end{align*}
in which the   indifference follows from part (iii) of \lemmaref{lemma_independence}, and $\varepsilon', \varepsilon''\in \bR^{B}_+$ satisfy $\varepsilon'_i,\varepsilon''_i>0$, $\varepsilon'_j=\varepsilon''_j=0$ for all $ j\in B\setminus\{i\}$, and
\begin{align*}
\delta_{\ux_B-\varepsilon'} &\sim_x \frac{\beta q(\ux_B)}{\beta q(\ux_B)+(1-\beta)\gamma}\delta_{\ux_B} + \frac{(1-\beta)\gamma}{\beta q(\ux_B)+(1-\beta)\gamma}\delta_{\ux_B-\varepsilon},     \\
\delta_{\lx_B+\varepsilon''} &\sim_x \frac{\beta q(\lx_B)}{\beta q(\lx_B) + (1-\beta)(1-\gamma)}\delta_{\lx_B} + \frac{(1-\beta)(1-\gamma)}{\beta q(\lx_B) + (1-\beta)(1-\gamma)}\delta_{\lx_B+\varepsilon}.
\end{align*}
The existence of $\varepsilon', \varepsilon''$ is guaranteed by \lemmaref{lemma_conditional}. Denote $\hat{q}^{\beta}:=[\beta q(\ux_B) + (1-\beta)\gamma] \delta_{\ux_B-\varepsilon'} + [\beta q(\lx_B) + (1-\beta)(1-\gamma)] \delta_{\lx_B+\varepsilon''}$. Then $\hat{q}^{\beta}\sim_x q\sim_x y$ and
\begin{align*}
    q^{\beta}\alpha s= & \big[\alpha(\beta q(\ux_B)\delta_{\ux_B} + (1-\beta)\gamma \delta_{\ux_B-\varepsilon}) + (1-\alpha)s(\ux_B)\delta_{\ux_B} \big] \\
     &+ \big[\alpha(\beta q(\lx_B)\delta_{\lx_B} + (1-\beta)(1-\gamma) \delta_{\lx_B+\varepsilon}) + (1-\alpha)s(\lx_B)\delta_{\lx_B}\big].
\end{align*}
Again by  applying \lemmaref{lemma_conditional} to the two terms above, respectively, and applying part (iii) of \lemmaref{lemma_independence}, we derive 
\begin{align*}
    q^{\beta}\alpha s\sim_x & \big[\alpha(\beta q(\ux_B) + (1-\beta)\gamma) \delta_{\ux_B-\varepsilon'} + (1-\alpha)s(\ux_B)\delta_{\ux_B} \big] \\
     &+ \big[\alpha(\beta q(\lx_B) + (1-\beta)(1-\gamma)) \delta_{\lx_B+\varepsilon''} + (1-\alpha)s(\lx_B)\delta_{\lx_B}\big] \\
     = &~ \hat{q}^{\beta}\alpha s.
\end{align*}
Note that $\supp(\hat{q}^{\beta}_i)\cap \supp(s_i)=\emptyset$, $y_i\not\in \supp(s_i)$, and $\hat{q}^{\beta}\sim_x y$. Part (iii) of \lemmaref{lemma_independence} implies that $\delta_y\alpha s \sim_x  \hat{q}^{\beta}\alpha s \sim_x q^{\beta}\alpha s$, which holds for all $\alpha,\beta\in (0,1)$. Let $\beta$ approach $1$ and by Axiom \ref{axiom_C}, we conclude that $q\alpha s \sim \delta_y\alpha s$.   
\end{proof}

For any $p\succ_x q$, the next result provides sufficient conditions for $p\alpha q$ to be preferred to $p\beta q$ whenever $\alpha>\beta$.

\begin{lemma}\label{lemma_mixture}
Suppose that $\succsim$ satisfies Axioms \ref{axiom_WO}-\ref{axiom_C} and $i\in B$.  If $i\hra B$, then for all $\alpha, \beta\in(0,1)$, $x\in X_{B^\mathsf{c}}$, and $p,q\in\Delta(X_B)$ such that  $\alpha>\beta$, $p\succ_x q$, and $\supp(p_i)\cap \supp(q_i)=\emptyset$, we have \emph{\text{(i)}}  $\delta_{\ux_B}\alpha \delta_{\lx_B} \succ_x \delta_{\ux_B}\beta \delta_{\lx_B}$ and \emph{\text{(ii)}}  $p\alpha q\succ_x p\beta q$.
 \end{lemma}

\begin{proof}[Proof of \lemmaref{lemma_mixture}]
For (i),  note that $\delta_{\ux_B}\beta \delta_{\lx_B}=(\delta_{\ux_B}\alpha \delta_{\lx_B})\frac{\beta}{\alpha} \delta_{\lx_B}$. By \lemmaref{lemma_monotone}, there exists $y\in X_B$ such that $y_i\neq \ux_i$, $y_i\neq \lx_i$, and $y\sim_x \delta_{\ux_B}\alpha \delta_{\lx_B}$. By \lemmaref{lemma_independence_extreme},  we have $\delta_{\ux_B}\beta \delta_{\lx_B}=(\delta_{\ux_B}\alpha \delta_{\lx_B})\frac{\beta}{\alpha} \delta_{\lx_B}\sim_x \delta_y \frac{\beta}{\alpha}\delta_{\lx_B}$. Since $y\succ_x \lx_B$ and $y_i\neq \lx_i$, part (i) of \lemmaref{lemma_independence} implies that $\delta_y \frac{\beta}{\alpha}\delta_{\lx_B} \prec_x y \sim_x \delta_{\ux_B}\alpha \delta_{\lx_B}$. Hence, $\delta_{\ux_B}\alpha \delta_{\lx_B}\succ_x\delta_{\ux_B}\beta \delta_{\lx_B}$. 

For (ii), by part (i), we can find unique $\gamma^1,\gamma^2\in [0,1]$ such that $\gamma^1>\gamma^2$ and $p\sim_x \delta_{\ux_B}\gamma^1 \delta_{\lx_B} \succ_x q \sim_x \delta_{\ux_B}\gamma^2 \delta_{\lx_B}.$ The existence of $\gamma^1$ and $\gamma^2$ follows from Axiom \ref{axiom_C} and  the standard mixture continuity argument. Then \lemmaref{lemma_independence_extreme} implies 
\begin{align*}
    p\alpha q \sim_x  \delta_{\ux_B}(\alpha\gamma^1 + (1-\alpha)\gamma^2)\delta_{\lx_B} \text{~~~and~~~}
    p\beta q \sim_x  \delta_{\ux_B}(\beta\gamma^1 + (1-\beta)\gamma^2)\delta_{\lx_B}.
\end{align*}
Since $\alpha>\beta$ and $\gamma^1> \gamma^2$, we know that $\alpha\gamma^1 + (1-\alpha)\gamma^2> \beta\gamma^1 + (1-\beta)\gamma^2$ and hence $p\alpha q\succ_x p\beta q$ by part (i).
\end{proof}

Now we are ready to prove \lemmaref{lemma_linear}. For any $p\in \Delta(X_B)$, by \lemmaref{lemma_mixture}, there exists a unique $\alpha (p)\in [0,1]$ such that $p\sim_x \delta_{\ux_B}\alpha(p) \delta_{\lx_B}$. Define  $U_x: \Delta(X_B)\rightarrow \bR$ such that $U_x(p)=\alpha(p)$ for every $p\in \Delta(X_B)$. Then $U_x(\delta_{\ux_B})=1$ and $U_x(\delta_{\lx_B})=0$.

\lemmaref{lemma_mixture} ensures that $p\succsim_x q$ if and only if $U_x(p)\geqslant  U_x(q)$ for all $p,q\in \Delta(X_B)$. Now we  check condition (ii). Fix any $\alpha\in (0,1)$ and $p,q\in \Delta(X_B)$ with $\supp(p_i)\cap \supp(q_i)=\emptyset$. By definition of $U_x$, we know that $p\sim_x \delta_{\ux_B}U_x(p) \delta_{\lx_B}$ and $q\sim_x \delta_{\ux_B}U_x(q) \delta_{\lx_B}$. Since $\supp(p_i)\cap \supp(q_i)=\emptyset$, \lemmaref{lemma_independence_extreme} implies $p\alpha q\sim_x \delta_{\ux_B}(\alpha U_x(p) + (1-\alpha)U_x(q)) \delta_{\lx_B}$. Again, the definition of $U_x$ implies $p\alpha q\sim_x \delta_{\ux_B} U_x(p\alpha q) \delta_{\lx_B}$. By \lemmaref{lemma_mixture}, we conclude that $U_x(p\alpha q)=\alpha U_x(p) + (1-\alpha)U_x(q)$. Hence, $U_x(p)=\sum_{y_i}U_x(\delta_{y_i}, p_{B\setminus \{i\}|y_i})p_i(y_i)$. To verify (iii), define  $w_x: X_B\rightarrow \bR$ by $w_x(y)=U_x(\delta_y)$ for all $y\in X_B$. By Axiom \ref{axiom_M}, $w_x$ is strictly increasing.  To see that $w_x$ is continuous, suppose by contradiction that there exists a sequence $(y^n)$ in $X_B$ such that $y^n\rightarrow y\in X_B$ and $w_x(y^n)\not\rightarrow w_x(y)$. Without loss of generality and passing to a subsequence if necessary, suppose $w_x(y^n)\rightarrow a<b= w_x(y)$ and $w_x(y^n)<(a+b)/2$ for all $n$. By part (ii), we can find $r\in \Delta(X_B)$ with $U_x(r)=(a+b)/2$.  That is, $y^n\prec_x r\prec_x y$ for all $n$. Axiom \ref{axiom_C}  implies $y\precsim_x r \prec_x y$, a contradiction.  Finally, by our construction, once  $U_x(\delta_{\ux_B})$ and $U_x(\delta_{\lx_B})$ are determined, the utility function $U_x$ is pinned down. Hence, $U_x$ is unique up to a positive affine transformation.\end{proof}

\begin{restatedlemma}{lemma_modify}
    Suppose $(\cP,E)\in\mathbb{T}(\succsim)$ and $A\in \cP$ such that $c(A)=\{\{i\}\}$ for some $i\in M(A\cup\bar d(A))$. Denote $\cP'=\cP\cup\{A\cup\{i\}\}\backslash\{A,\{i\}\}$ and let $(B,B')\in E'$ if either (i) $(B,B')\in E$,  (ii) $B'=A\cup\{i\}$ and $(B,A)\in E$, or  (iii) $B=A\cup\{i\}$ and $(\{i\},B')\in E$. Then $(\cP',E')\in \mathbb{T}(\succsim)$.
\end{restatedlemma}

\begin{proof}[Proof of \lemmaref{lemma_modify}]
Suppose  that $(\cT,(u^A)_{A\in\cP_o})$ is an  SMEU representation of $\succsim$ with $\cT=(\cP,E)$ and that $\cT'=(\cP',E')$ satisfies the conditions stated in \lemmaref{lemma_modify}. Let $A'\in \cP_o$ such that $A\in c^{\cT}(A')$.  Then  $ c^{\cT'}(A\cup\{i\})=  c^{\cT}(\{i\})$,  $\bar a^{\cT'}(A\cup\{i\})= \bar a^{\cT}(A)$, and $c^{\cT'}(A')=c^{\cT}(A')\cup\{A\cup\{i\}\}\setminus\{A\}$.
On the rest of the domain, $c^{\cT}=c^{\cT'}$ and $\bar a^{\cT}=\bar a^{\cT'}$.
 
Let $\hat{u}^B=u^B$ for all  $B\in \cP'_o\setminus\{A\cup\{i\}\}$ (noting that $|c^{\cT'}(A')|=|c^{\cT}(A')|$) and define 
 $\hat{u}^{A\cup\{i\}}:X_{\bar a^{\cT}(A)}\times X_{A\cup\{i\}}\times\bR^{ c^{T}(\{i\})}\to\bR$ by 
\[\hat{u}^{A\cup\{i\}}(z,x,a)=u^A\big(z,\: x_{A}, \:u^{i}  \big(\:(z, x_{A}),\: x_i, \:a\big)\big).\]
Define $U^B_z$ and $\hat{U}^B_z$ accordingly using the recursive equation (\ref{kp}). 

It is easy to see that $\hat{U}^{A\cup\{i\}}_z(\delta_x)= U^A_z(\delta_x)$ and 
$\hat{U}^B_z(\delta_x)= U^B_z(\delta_x)$ for all $B\in \cP'_o\setminus\{A\cup\{i\}\}, z\in X_{\bar a^{\cT}(A)}$, and $x\in X$. This guarantees that for any $B\in \cP'_o$ and 
$z'\in X_{(B\cup \bar d^{\cT'}(B))^\mathsf{c}}$ with $z'_{\bar a^{\cT'}(B)}=z$,   the function $\hat U^B_z(\delta_{(z',x)})$ is continuous and strictly increasing in $x\in X_{B\cup \bar d^{\cT'}(B)}$. 
Moreover, $\hat{U}^B_z(p)= U^B_z(p)$ for all $B\in \cP'$ with $B\subseteq \bar d^{\cT'}(A\cup\{i\})=\bar d^{\cT}(\{i\})$,  $p\in \Delta(X)$, and $z\in X_{\bar a^{\cT}(B)}$.
To show that $(\cT',(\hat{u}^B)_{B\in\cP'_o})$ is an SMEU representation of $\succsim$, it suffices to show that $\hat{U}^{A\cup\{i\}}_z(p)= U^A_z(p)$ for all $p\in \Delta(X)$ and $z\in X_{\bar a^{\cT}(A)}$, since by recursion, this property implies that $\hat{U}^B_z(p)= U^B_z(p)$ for all $B\in \cP'_o$ with $A\cup\{i\}\subseteq \bar d^{\cT}(B)$,  $p\in \Delta(X)$ and $z\in X_{\bar a^{\cT}(B)}$.

 Fix any $z'\in  X_{(A\cup \bar d^{\cT}(A))^\cp}$ and denote $z= z'_{\bar a^{\cT}(A)} \in X_{\bar a^{\cT}(A)}$. The utilities of $p\in \Delta(X)$ in the two representations are given by 
 \begin{align*}
      U^A_z(p) & = \bE^p_{A|z} \: u^A\big(z,y,   \bE^p_{i|(z,y)}\:u^{i}((z,y),y',\:(U^B_{(z,y,y')}(p))_{B\in c^{\cT}(\{i\})}) \big),\\
      \hat{U}^{A\cup\{i\}}_z(p)& =\bE^p_{A\cup\{i\}|z}\:\hat u^{A\cup\{i\}}\big(\:z,\:(y,y'),\:(U^B_{(z,y,y')}(p))_{B\in c^{\cT}(\{i\})})\big) \\
      & = \bE^p_{A\cup\{i\}|z}\: u^A\big(z,y,\:u^{i}((z,y),y',\:(U^B_{(z,y,y')}(p))_{B\in c^{\cT}(\{i\})})\big).
 \end{align*}
 Note that in the above expressions, $y\in X_A$ and $y'\in X_i$.

Since $i\in M(A\cup\bar d^{\cT}(A))$,  \lemmaref{lemma_linear} ensures the existence of a function $V_z: \Delta(X_{A\cup\bar d^{\cT}(A)})\to \bR$ such that  {\text{(i)}}  $p\succsim_{z'} q$ if and only if $V_z(p)\geqslant  V_z(q)$ for all $p,q\in \Delta(X_{A\cup\bar d^{\cT}(A)})$ and
 {\text{(ii)}}  $V_z(p\alpha q ) = \alpha V_z(p) + (1-\alpha)V_z(q)$ for all $\alpha\in (0,1)$ and $p,q\in \Delta(X_{A\cup\bar d^{\cT}(A)})$ with $\supp(p_i)\cap \supp(q_i)=\emptyset$.

Fix any $y\in X_{A}$ and $p \in \Delta(X_{\bar d^{\cT}(A)})$. Denote $\supp(p_i)=\{y'^1,\dots,y'^n\}$ such that $y'^1<y'^2<\dots<y'^n$. We can find $y^k\in  X_{A}$,  $y''^k\in X_i$, and $\hat y^k\in X_{\bar d^{\cT}(A)\setminus\{i\}}$ for every $k=1,\dots,n$ such that 
$(y,y'^k,p_{\bar d^{\cT}(A)\setminus \{i\}|y'^k})\sim_{z'} (y^k,y''^k,\hat y^k)$, elements in $\{y^1,\dots,y^n\}$ are mutually distinct, and  elements in $\{y''^1,\dots,y''^n\}$ are mutually distinct. Denote $q=\sum_{k=1}^n p_i(y'^k)\delta_{(y^k,y''^k,\hat y^k)}$. 
 Properties (i) and (ii) ensure that $V_z(\delta_y,p) = V_z(q)$ and hence $(\delta_y,p)\sim_{z'} q$. Since $(\cT,(u^A)_{A\in\cP_o})$ is an  SMEU representation of $\succsim$, we have 
\begin{align*}
      U^A_z(\delta_{z'},\delta_y,p)  
     = & ~u^A\big(z,y,   \bE^p_{i}\:u^{i}((z,y),y',\:(U^B_{(z,y,y')}(\delta_{z'},\delta_y,p))_{B\in c^{\cT}(\{i\})}) \big)\\
     =  & ~U^A_z(\delta_{z'},q) \\
    =  & ~\sum_{k=1}^n p_i(y'^k) \:u^A\big(z,y^k,\:u^{i}((z,y^k),y''^k,\:(U^B_{(z,y^k,y''^k)}(\delta_{z'},q)_{B\in c^{\cT}(\{i\})})\big)\\
     = & ~ \bE^p_{i}\: u^A\big(z,y,\:u^{i}((z,y),y',\:(U^B_{(z,y,y')}(\delta_{z'},\delta_y,p))_{B\in c^{\cT}(\{i\})})\big).
\end{align*}
 Combining the above two sets of equations, for any $r\in \Delta(X)$, we have 
\begin{align*}
    U^A_z(r) & = \bE^r_{A|z} \: u^A\big(z,y,   \bE^r_{i|(z,y)}\:u^{i}((z,y),y',\:(U^B_{(z,y,y')}(r))_{B\in c^{\cT}(\{i\})}) \big)\\
    & =  \bE^r_{A|z} \bE^r_{i|(z,y)}\: u^A\big(z,y,  \:u^{i}((z,y),y',\:(U^B_{(z,y,y')}(r))_{B\in c^{\cT}(\{i\})}) \big)\\
    & = \bE^r_{A\cup\{i\}|z}\:u^A\big(z,y,  \:u^{i}((z,y),y',\:(U^B_{(z,y,y')}(r))_{B\in c^{\cT}(\{i\})}) \big)\\
    & = \hat{U}^A_z(r). 
\end{align*}    
Hence, $\hat{U}^A_z(r)= U^A_z(r)$ for all $r\in \Delta(X)$ and $z\in X_{\bar a^{\cT}(A)}$. This implies that $(\cT',(\hat{u}^B)_{B\in\cP'_o})$ is also an SMEU representation of $\succsim$.
\end{proof}


 \setstretch{1.1}
 \bibliographystyleAppendix{ecta}
\bibliographyAppendix{MCU-online}

\end{document}